\documentclass[10pt,twocolumn,twoside,submit]{IEEEtran}

\usepackage{cite}
\usepackage{amsthm,amsmath,amssymb}
\usepackage{amsfonts}
\usepackage{txfonts}
\usepackage{algorithmic}
\usepackage{graphicx}
\usepackage{textcomp}
\usepackage{xcolor}
\usepackage{comment}
\usepackage{float}
\usepackage{flushend}
\usepackage{stfloats}
\usepackage{url}
\usepackage{caption}
\usepackage{subfigure}
\usepackage{comment}
\usepackage{epsfig}
\usepackage[T1]{fontenc}
\usepackage{epstopdf}

\def\BibTeX{{\rm B\kern-.05em{\sc i\kern-.025em b}\kern-.08em
    T\kern-.1667em\lower.7ex\hbox{E}\kern-.125emX}}
    
\newtheorem{proposition}{\textbf{Proposition}}
\newtheorem{definition}{\textbf{Definition}}
\newtheorem{observation}{\textbf{Observation}}
\newtheorem{guideline}{\textbf{Guideline}}

\begin{document}

%
\title{Anti-Aging Scheduling in Single-Server Queues:\\
A Systematic and Comparative Study}

\author{
Zhongdong Liu,
Liang Huang,
Bin Li,
and Bo Ji
\thanks{This work was supported in part by the NSF under Grants CCF-1657162, CNS-1651947, and CNS-1717108. A preliminary version of this work was presented at IEEE INFOCOM 2020 Age of Information Workshop \cite{Liu2020:Anti}.}
\thanks{Zhongdong Liu (zhongdong@vt.edu) and Bo Ji (boji@vt.edu) are with the Department of Computer Science, Virginia Tech, Blacksburg, VA. 
Liang Huang (lianghuang@zjut.edu.cn) is with the College of Computer Science and Technology, Zhejiang
University of Technology, Hangzhou, China. 
Bin Li (binli@uri.edu) is with the Department of Electrical, Computer and Biomedical Engineering, University of Rhode Island, Kingston, Rhode Island. Bo Ji is the corresponding author.
}}
\maketitle

\begin{abstract}
The Age-of-Information (AoI) is a new performance metric recently proposed for measuring the freshness of information in information-update systems. 
In this work, we conduct a systematic and comparative study to investigate the impact of scheduling policies on the AoI performance in single-server queues and provide useful guidelines for the design of AoI-efficient scheduling policies. Specifically, we first perform extensive simulations to demonstrate that the update-size information can be leveraged for achieving a substantially improved AoI compared to non-size-based (or arrival-time-based) policies. Then, by utilizing both the update-size and arrival-time information, we propose three AoI-based policies. Observing improved AoI performance of policies that allow service preemption and that prioritize informative updates, we further propose preemptive, informative, AoI-based scheduling policies. Our simulation results show that such policies empirically achieve the best AoI performance among all the considered policies. 
However, compared to the best delay-efficient policies (such as Shortest-Remaining-Processing-Time (SRPT)), the AoI improvement is rather marginal in the settings with exogenous arrivals.
Interestingly, we also prove sample-path equivalence between some size-based policies and AoI-based policies. This provides an intuitive explanation for why some size-based policies (such as SRPT) achieve a very good AoI performance. 
\end{abstract}


\section{\uppercase{Introduction}}
Recently, the study of information freshness has received increasing attentions, especially for time-sensitive applications that require real-time information/status updates, such as road congestion alerts, stock quotes, and weather forecast. In order to measure the freshness of information, a new metric, called the \textit{Age-of-Information (AoI)} is proposed. The AoI is defined as the time elapsed since the generation of the freshest update among those that have been received by the destination~\cite{kaul2012real}.
Prior studies reveal that the AoI depends on both the inter-arrival time and the delay of the updates. Due to the dependency between the inter-arrival time and the delay, this new AoI metric exhibits very different characteristics than the traditional delay metric and is generally much harder to analyze (see, e.g., \cite{kaul2012real}).

Although it is well-known that scheduling policies play an important role in reducing the delay in single-sever queues, it remains largely unknown how exactly scheduling policies impact the AoI performance.
To that end, we aim to holistically study the impact of various aspects of scheduling policies on the AoI performance in single-server queues and provide useful guidelines for the design of scheduling policies that can achieve a small AoI. 

While much research effort has already been exerted to the design and analysis of scheduling policies aiming to reduce the AoI, almost all of these policies are only based on the arrival time of updates, such as First-Come-First-Served (FCFS) and Last-Come-First-Served (LCFS), assuming that the update-size information is unavailable.
Here, the size of an update is the amount of time required to serve the update if there were no other updates around.
In some applications, such as smart grid and  traffic monitoring, the update-size information can be obtained or fairly well estimated \cite{wu2017optimal}.
It has been shown that scheduling policies that leverage the size information can substantially reduce the delay, especially when the system load is high or when the size variability is large \cite{harchol2013performance}. 
This motivates us to investigate the AoI performance of size-based policies in a G/G/1 queue.
Note that the update-size information is ``orthogonal" to the arrival-time information, both of which could significantly impact the AoI performance.
Therefore, it is quite natural to further consider AoI-based policies that use both the update-size and arrival-time information of updates.

\begin{figure}[!t]
    \centering
    \includegraphics[scale=0.5]{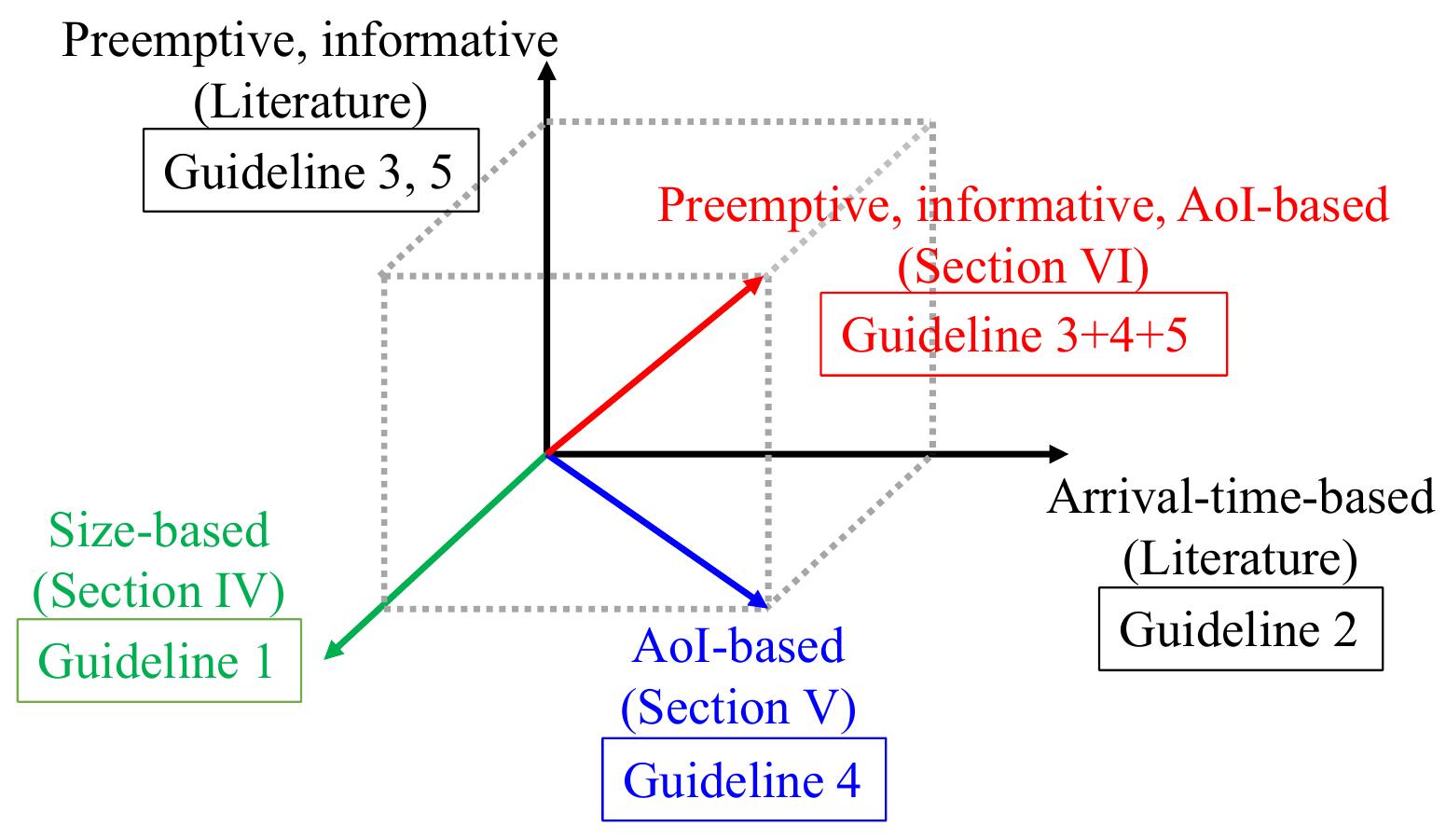}
    \caption{Our position in the design space of AoI-efficient scheduling policies for a G/G/$1$ queue}
    \label{fig:position}
\end{figure} 

\begin{table*}[!t]
\centering
\begin{tabular}{|c|c|c|}
\hline
Guideline  & Summary   & Representative policies  \\ \hline  \hline
\ref{guide:use-size-info}  & \textbf{ Prioritizing small updates}  & SJF, SJF\_P, SRPT \\ \hline
\ref{guide:arrival-time-info}  & Prioritizing recent updates  & LCFS, LCFS\_P \\ \hline
\ref{gui:preemption}   & Allowing service preemption  & PS, LCFS\_P, SJF\_P, SRPT \\ \hline
\ref{guide:use-size-arrival-info}  & \textbf{AoI-based designs} & \textbf{ADE}, \textbf{ADS}, \textbf{ADM} \\ \hline
\ref{guide:informative}  & Prioritizing informative updates  & Informative version of the above policies \\ \hline
\end{tabular}
\caption{Guidelines for the design of AoI-efficient scheduling policies for a G/G/$1$ queue}
\label{table:guidelines-summary}
\end{table*}

In addition, prior work has revealed that scheduling policies that allow service \emph{preemption} and that prioritize \emph{informative} updates (also called \emph{effective} updates, which are those that lead to a reduced AoI once delivered; see Section~\ref{subsec:infor-policies} for a formal definition) yield a good AoI performance \cite{bedewy2016optimizing,costa2014age,pappas2015age}. Intuitively, preemption prevents fresh updates from being blocked by a large and/or stale update in service; informative policies discard stale updates, which do not bring new information but may block fresh updates. To that end, we also consider AoI-based scheduling designs that both allow service preemption and prioritize informative updates.

In Fig.~\ref{fig:position}, we position our work in the literature by summarizing various design aspects of scheduling policies for a G/G/$1$ queue. Existing work mostly explores the design based on the arrival-time information along with considering service preemption and informative updates. We point out that the size-based design is an orthogonal dimension of great importance, which somehow has not received sufficient attentions yet. Unsurprisingly, designing AoI-efficient policies requires the consideration of all these dimensions.
In Table~\ref{table:guidelines-summary}, we summarize several useful guidelines for the design of AoI-efficient policies, which are also labeled in Fig.\ref{fig:position}. 
To the best of our knowledge, this is the first work that conducts a systematic and comparative study to investigate the design of AoI-efficient scheduling policies for a G/G/$1$ queue.
In the following, we summarize our key contributions along with an explanation of Fig.~\ref{fig:position} and Table~\ref{table:guidelines-summary}.

First, we investigate the AoI performance of size-based scheduling policies (i.e., the green arrow in Fig.~\ref{fig:position}), which is an orthogonal approach to the arrival-time-based design studied in most existing work.
We conduct extensive simulations to show that size-based policies that prioritize small updates significantly improve AoI performance. We also explain interesting observations from the simulation results and summarize useful guidelines (i.e., Guidelines \ref{guide:use-size-info}, \ref{guide:arrival-time-info}, and \ref{gui:preemption} in Table~\ref{table:guidelines-summary}) for the design of AoI-efficient policies.

Second, leveraging both the update-size and arrival-time information, we introduce Guideline \ref{guide:use-size-arrival-info} and propose AoI-based scheduling policies (i.e., the blue arrow in Fig.~\ref{fig:position}). These AoI-based policies attempt to optimize the AoI at a specific future time instant from three different perspectives: the AoI-Drop-Earliest (ADE) policy, which makes the AoI drop the earliest; the AoI-Drop-to-Smallest (ADS) policy, which makes the AoI drop to the smallest; the AoI-Drop-Most (ADM) policy, which makes the AoI drop the most. 
The simulation results show that such AoI-based policies indeed have a good AoI performance.

Third, we observe that informative policies can significantly improve the AoI performance compared to their non-informative counterparts, which leads to Guideline~\ref{guide:informative}. Integrating all the guidelines, we propose preemptive, informative, AoI-based policies (i.e., the red arrow in Fig.~\ref{fig:position}). The simulation results show that such policies empirically achieve the best AoI performance among all the considered policies.

Finally, we prove sample-path equivalence between some size-based policies and AoI-based policies. These results provide an intuitive explanation for why some size-based policies, such as Shortest-Remaining-Processing-Time (SRPT), achieve a very good AoI performance.

To summarize, our study reveals that among various aspects of scheduling policies we investigated, prioritizing small updates, allowing service preemption, and prioritizing informative updates play the most important role in the design of AoI-efficient scheduling policies.
However, compared to the best delay-efficient policies (such as SRPT), the AoI improvement of the preemptive, informative, and AoI-based policies is rather marginal in the settings with exogenous arrivals. Moreover, when the AoI requirement is not stringent or the update-size information is not available, some simple delay-efficient policies (such as LCFS with preemption (LCFS\_P)) are also good candidates for AoI-efficient policies.

The rest of this paper is organized as follows. We first discuss related work in Section~\ref{sec:relatedwork}. Then, we describe our system model in Section~\ref{sec:model}. In Section~\ref{sec:common-scheduling}, we evaluate the AoI performance of size-based scheduling policies. We further propose AoI-based scheduling policies in Section~\ref{sec:drop-to-point}. In addition, we evaluate the AoI performance of preemptive, informative, AoI-based policies in Section~\ref{sec:pre-infor}. Finally, we make concluding remarks in Section~\ref{sec:conclusion}.

\section{\uppercase{Related work}}
\label{sec:relatedwork}

The traditional queueing literature on single-server queues is largely focused on the delay analysis. 
In \cite{crovella1999connection}, the authors prove that all non-preemptive scheduling policies that do not make use of job	size information  have the same distribution of the number of jobs in the system.
The work of \cite{RePEc:inm:oropre:v:16:y:1968:i:3:p:687-690,smith1978new} proves that for a work-conserving queue, the SRPT policy minimizes the number of jobs in the system at any point and is therefore delay-optimal. 
The work of \cite{harchol2010queueing} derives a formula of the average delay for several common scheduling polices (which will be discussed in Section~\ref{sec:common-scheduling}). 
	
On the other hand, although the AoI research is still in a nascent stage, it has already attracted a lot of interests (see \cite{kosta2017age,sun2019age} for a survey). Here we only discuss the most relevant work, which is focused on the AoI-oriented queueing analysis. Much of existing work considers scheduling policies that are based on the arrival time (such as FCFS and LCFS). The AoI is introduced in \cite{kaul2012real}, where the authors study the average AoI in the M/M/1, M/D/1, and D/M/1 queues under the FCFS policy. In \cite{costa2016age}, the AoI performance of the FCFS policy in the M/M/$1$/$1$ and M/M/$1$/$2$ queues is studied, where new arrivals are discarded if the buffer is full. 
In \cite{9099557}, the authors study the average AoI performance of a multi-source FCFS M/G/1 queue. They derive the exact expression and three approximations of the average AoI for a special case of an M/M/$1$ queue and a general case of an M/G/$1$ queue, respectively.
The average AoI of the LCFS policy in the M/M/$1$ queue is also discussed in \cite{costa2016age}.

There has been some work that aims to reduce the AoI by making use of service preemption.  In \cite{kaul2012status}, the average AoI of LCFS in the M/M/$1$ queue with and without service preemption is analyzed. 
The work of \cite{kam2014effect} is quite similar to \cite{kaul2012status}, but it considers the average AoI in the M/M/$2$ queue. 
In \cite{najm2018status}, the average AoI for the M/G/$1$/$1$ preemptive system with a multi-stream updates source is derived.  
The age-optimality of the preemptive LCFS (LCFS\_P) policy is proved in \cite{bedewy2016optimizing}, where the service times are exponentially distributed.

In addition to taking advantage of service preemption, some of the prior studies also consider the strategy of prioritizing informative updates for reducing the AoI.
The work of \cite{costa2014age,pappas2015age} reveals that the AoI performance can be improved by prioritizing informative updates and discarding non-informative policies when making scheduling decisions. 
In \cite{inoue2018general}, the authors consider a G/G/$1$ queue with informative updates and derive the stationary distribution of the AoI, which is in terms of the stationary distribution of the delay and the Peak AoI (PAoI). With the AoI distribution, one can analyze the mean or higher moments of the AoI in GI/GI/$1$, M/GI/$1$, and GI/M/$1$ queues under several scheduling policies (e.g., FCFS and LCFS).

Recent research effort has also been exerted to understanding the relation between the AoI and the delay. In \cite{talak2019age}, the authors analyze the tradeoff between the AoI and the delay in a single-server M/G/$1$ system under a specific scheduling policy without knowing the service time of each individual update.  In \cite{devassy2018delay}, the violation probability of the delay and the PAoI is investigated under an additive white Gaussian noise (AWGN) channel, but the update size is assumed to be identical.

\section{\uppercase{System Model}}
\label{sec:model}
In this section, we consider a single-server queueing system and give the definitions of the Age-of-Information (AoI) and the Peak AoI (PAoI).

We model the information-update system as a G/G/$1$ queue where a single source generates updates (which contain current state of a measurement or observation of the source) with rate  $\lambda$. The updates enter the queueing system immediately after they are generated. Hence, the generation time is the same as the arrival time.
We use $S$ to denote the size of an update (i.e., the amount of time required for the update to complete service), which has a general distribution with mean $ {\mathbb{E}} \left[ S \right]{\rm{ = 1/}}\mu$. The system load is defined as $\rho  \triangleq \lambda {\rm{/}}\mu $.

\begin{figure}[!t]
	\centering
	\includegraphics[scale=0.5]{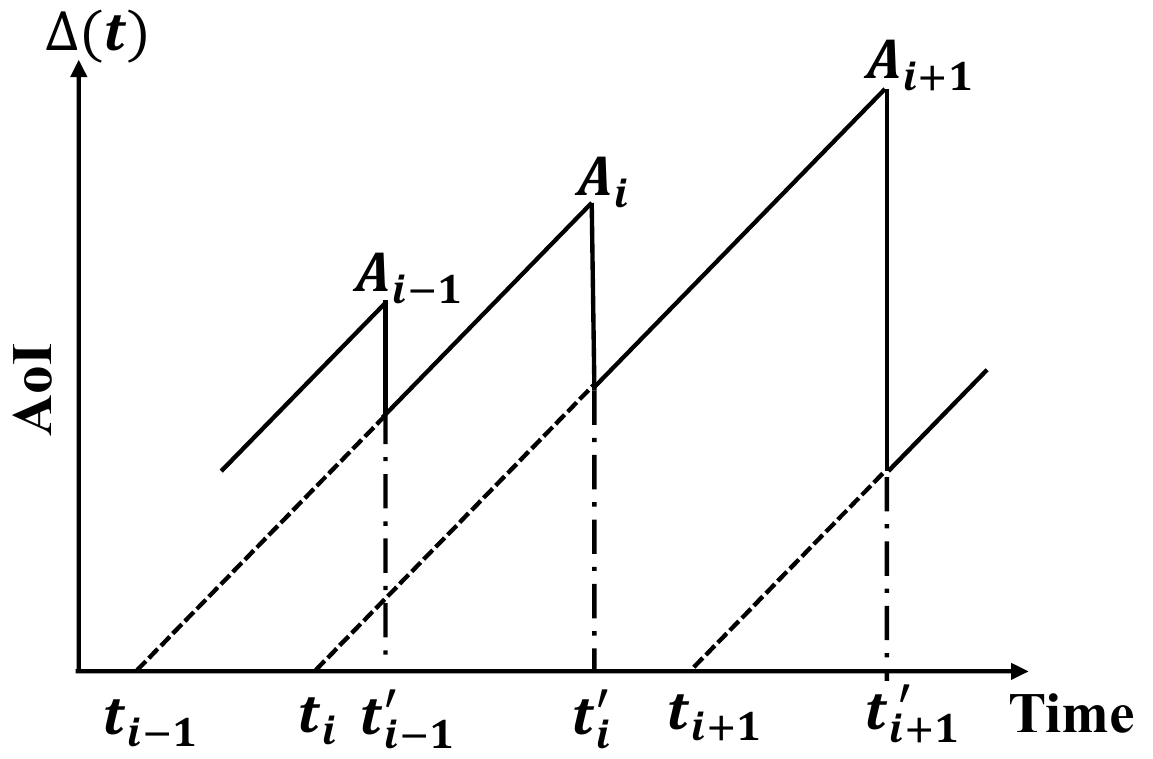}
	\caption{An example of the AoI evolution under the FCFS policy}
	\label{fig:aoi_ex}
\end{figure}

We use ${t_i}$ and ${t'_i}$ to denote the time at which the $i$-th update was generated at the source and the time at which it leaves the server, respectively. The AoI at time $t$ is then defined as  $\Delta(t) \triangleq t - U(t)$, where $U\left( t \right){\rm{ \triangleq max}}\left\{ {{t_i}:{t'_i}  \le t} \right\}$ is the generation time of the freshest update among those that have been processed by the server. An example of the AoI evolution under the FCFS policy is shown in  Fig.~\ref{fig:aoi_ex}. Then, the  average AoI can be defined as 

\begin{equation}
\Delta  = \mathop {\lim }\limits_{t  \to \infty } \dfrac{1}{t}\int_0^t  {\Delta \left( \tau \right)} d \tau.
\end{equation}

In general, the analysis of the average AoI is quite difficult since it is determined by two dependent quantities:  the inter-arrival time and the delay of updates\cite{kaul2012real}. We define the inter-arrival time between the $i$-th update and $(i-1)$-th update as ${X_i} \triangleq {t_i} - {t_{i - 1}}$ and define the delay of the $i$-th update as ${T_i} \triangleq {t^{\prime}_i} - {t_i}$. Alternatively, the Peak AoI (PAoI) is also proposed as an information freshness metric \cite{costa2014age}, which is defined as the maximum value of the AoI before it drops due to a newly delivered fresh update. Let $A_i$ be the $i$-th PAoI. From Fig.~\ref{fig:aoi_ex}, we can see $A_i = t^{\prime}_i - t_{i - 1}$. This can be rewritten as the sum of the inter-arrival time between the $i$-th update and the previous update (i.e., ${X_i}$) and the delay of the $i$-th update (i.e., ${T_i}$). Therefore, the PAoI of the $i$-th update can also be expressed as ${A_i} = {X_i} + {T_i}$, and its expectation is
$ {\mathbb{E}}[A_i] =  {\mathbb{E}}[ {X_i}] +  {\mathbb{E}}[ {T_i}]$.

\section{\uppercase{Size-based  policies}}
\label{sec:common-scheduling} 
In this section, we investigate the AoI performance of several common scheduling policies, including size-based policies and non-size-based policies, via extensive simulations.  
Note that these common scheduling policies may serve the non-informative updates (which do not lead to a reduced AoI). This is because in some applications, such as news and social network,  obsolete updates are still useful and need to be served~\cite{bedewy2016optimizing}. 
In Section~\ref{sec:pre-infor}, we will discuss the case where obsolete updates are discarded.

Following \cite{harchol2013performance}, we first give the definitions of several common scheduling policies that can be divided into four types: 
depending on whether they are size-based or not, where the size-based policies use the update-size information (which is available in some applications, such as smart grid  \cite{wu2017optimal}) for making scheduling decisions;
depending on whether they are preemptive or not. The definition of preemption is given below. 
In this paper, we do not consider the cost of preemption.

\begin{definition}
    A policy is preemptive if an update may be stopped partway through
its execution and then restarted at a later time without losing intermediary work. 
\end{definition}

\begin{figure*}[!t]
    \centering
    \subfigure[Exponential: $\mu=1$]{
		\label{fig:tradition-exp-AoI} 
		\includegraphics[width=0.322\textwidth]{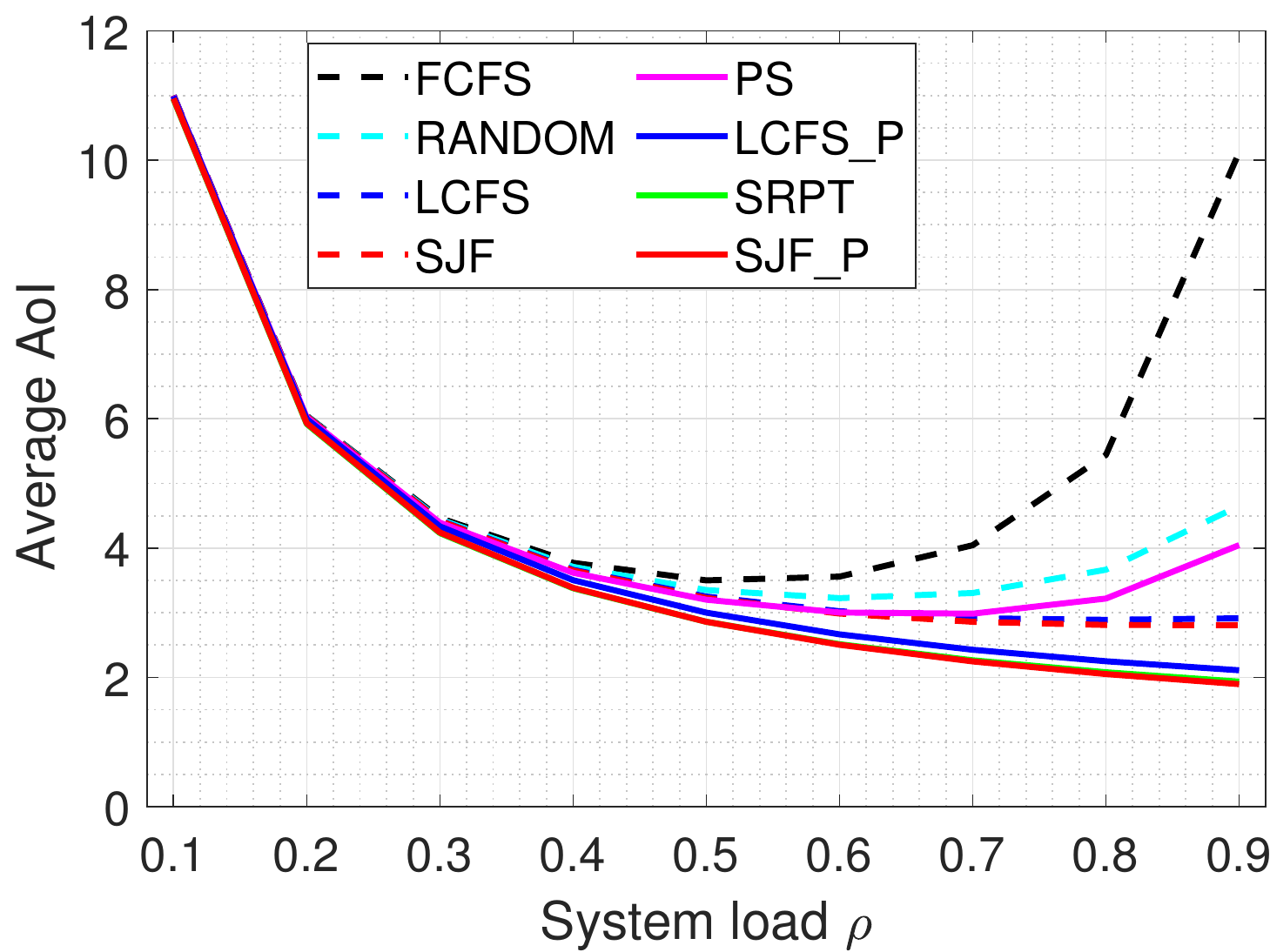}}
	\subfigure[Weibull: $\mu=1$ and  ${C^{\rm{2}}}{\rm{ = 10}}$]{
		\label{fig:tradition-wei-AoI} 
		\includegraphics[width=0.322\textwidth]{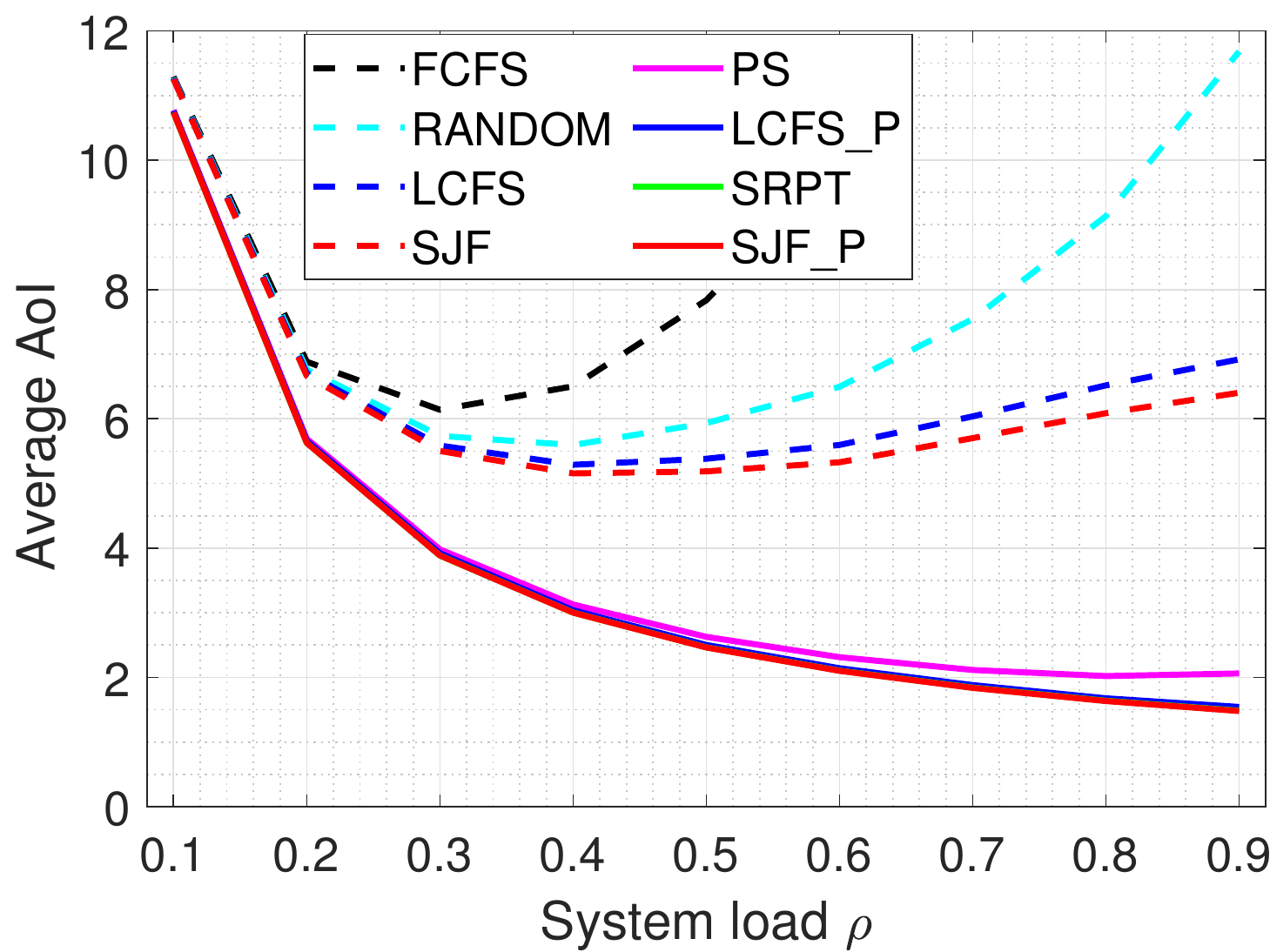}}
	\subfigure[Weibull: $\mu=1$ and $\rho {\rm{ = 0}}{\rm{.7}}$]{
		\label{fig:tradition-wei-variance-AoI} 
		\includegraphics[width=0.322\textwidth]{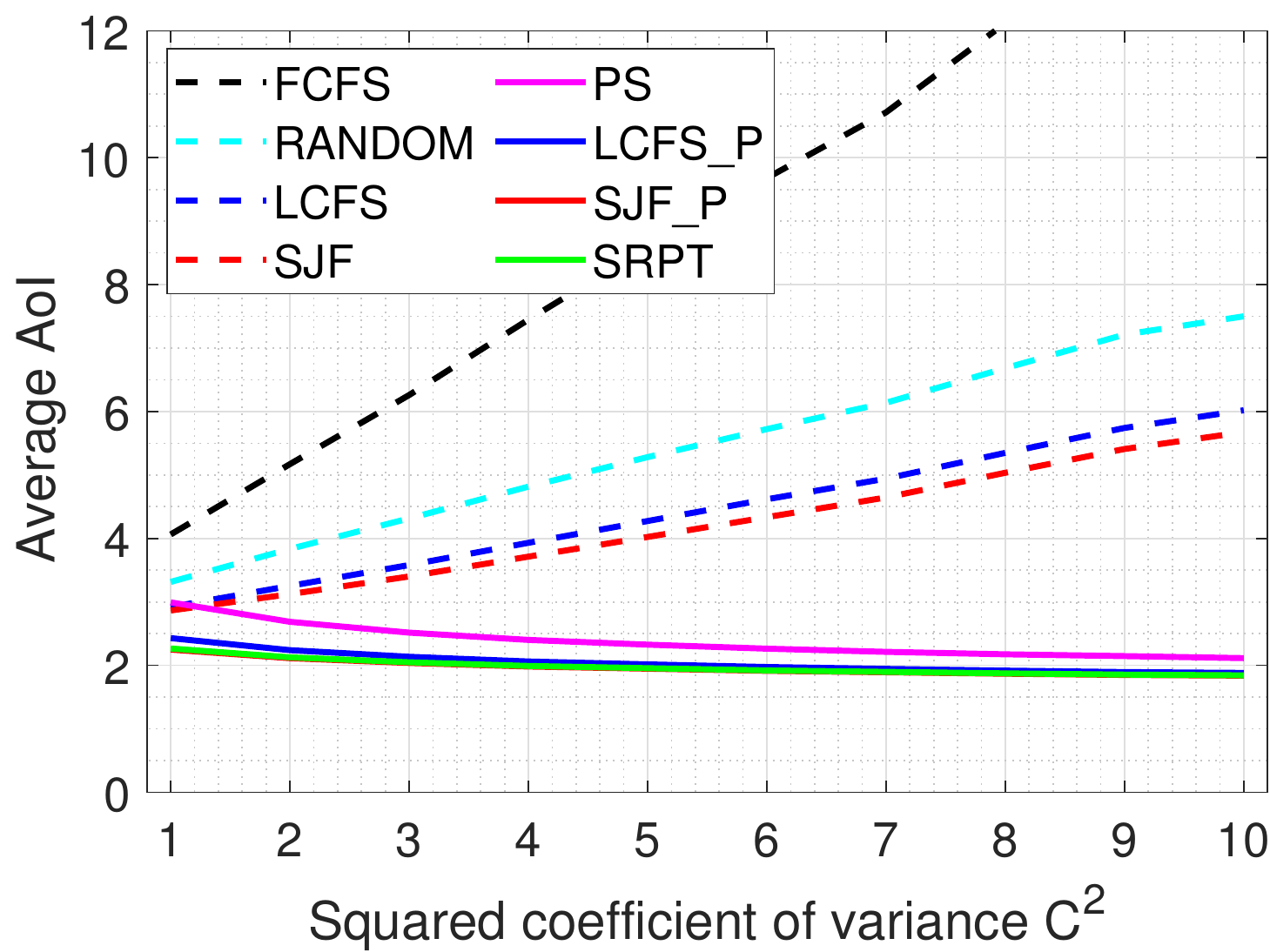}}
	\caption{Comparisons of the average AoI performance under several common scheduling policies}
	\label{fig:tradition-AoI}
\end{figure*}

\begin{figure*}[!t]
    \centering
    \subfigure[Exponential: $\mu=1$]{
		\label{fig:tradition-exp-PAoI} 
		\includegraphics[width=0.322\textwidth]{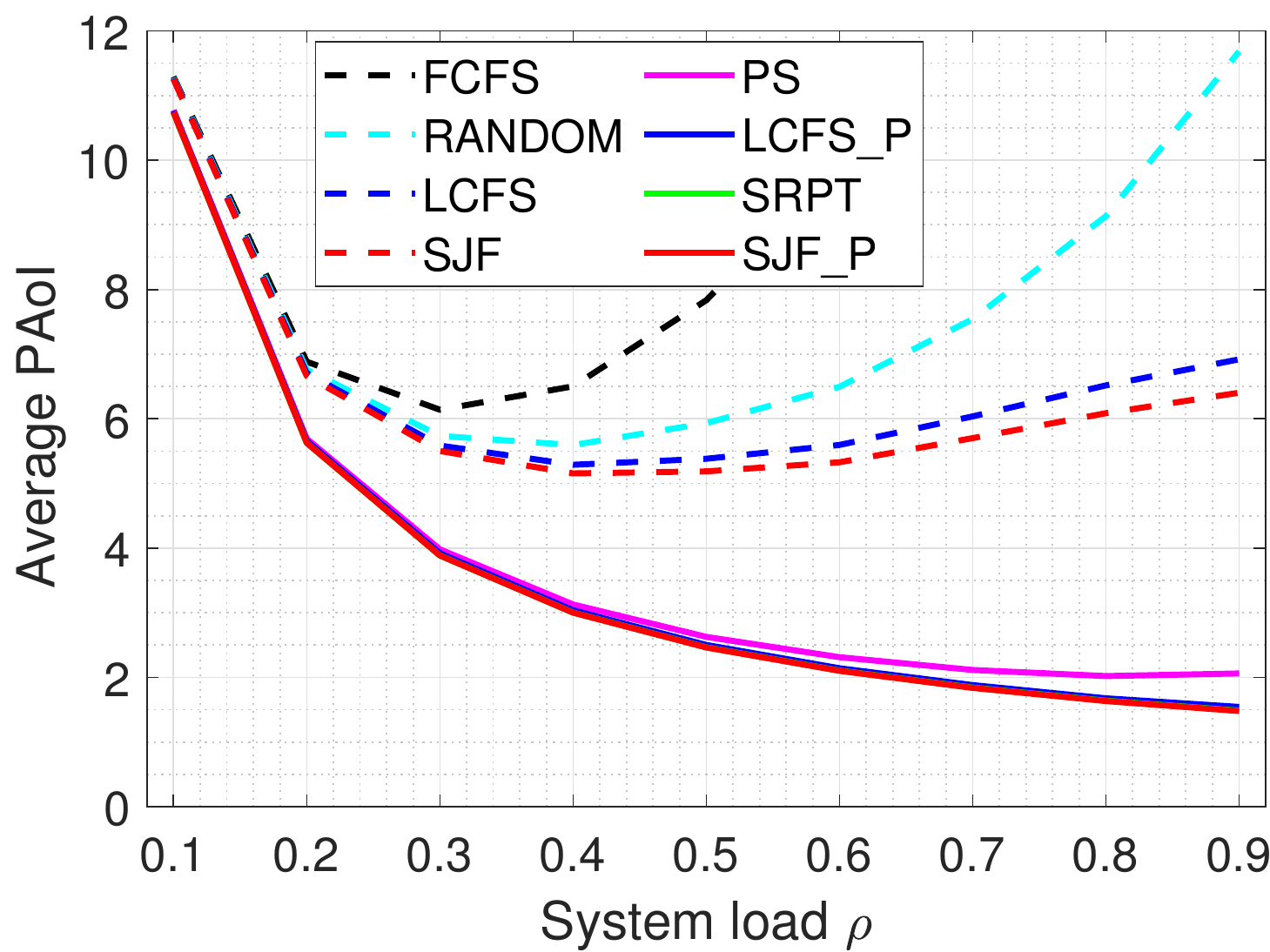}}
	\subfigure[Weibull: $\mu=1$ and ${C^{\rm{2}}}{\rm{ = 10}}$]{
		\label{fig:tradition-wei-PAoI} 
		\includegraphics[width=0.322\textwidth]{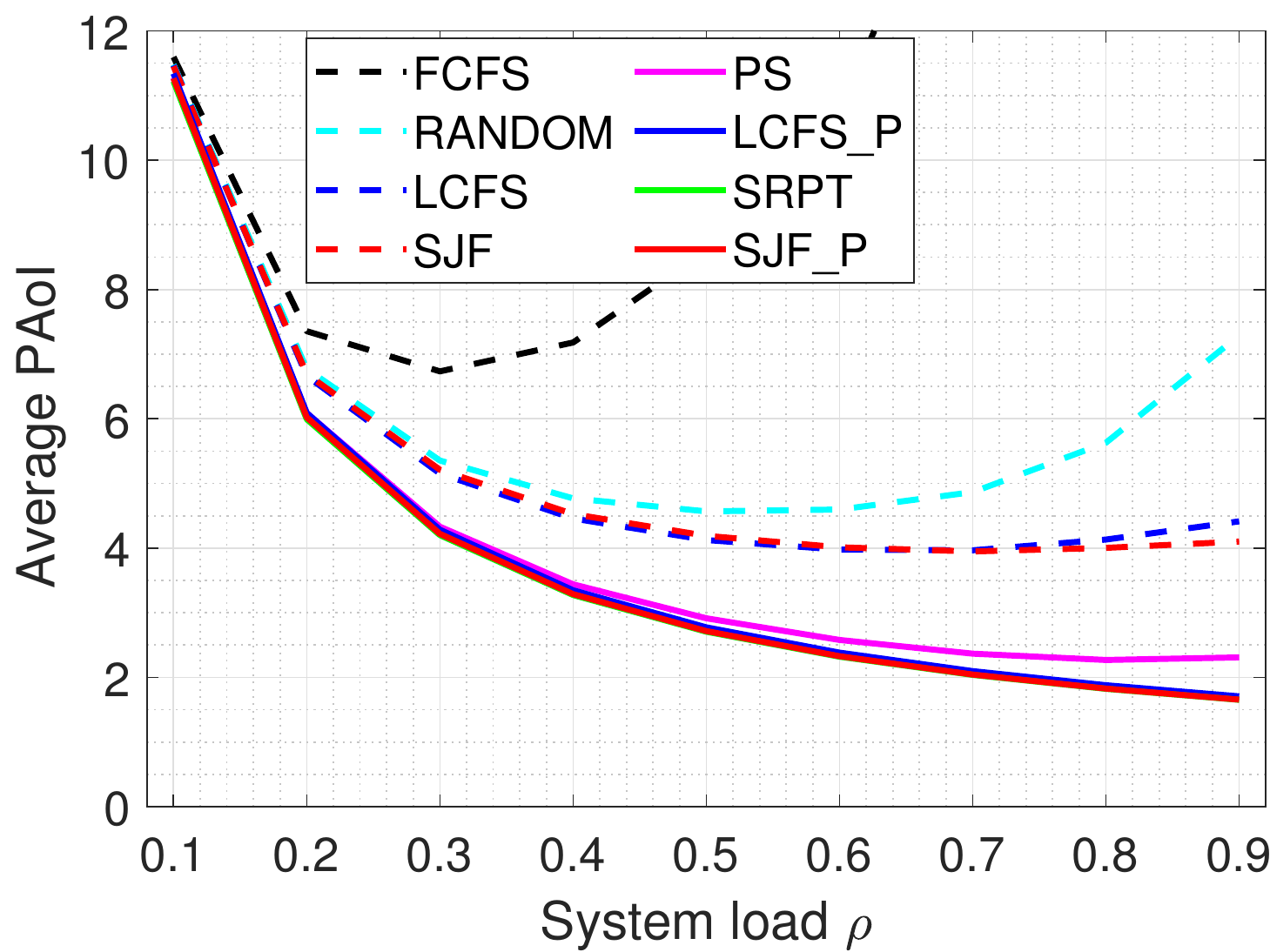}}
	\subfigure[Weibull: $\mu=1$ and $\rho {\rm{ = 0}}{\rm{.7}}$]{
		\label{fig:tradition-wei-variance-PAoI} 
		\includegraphics[width=0.322\textwidth]{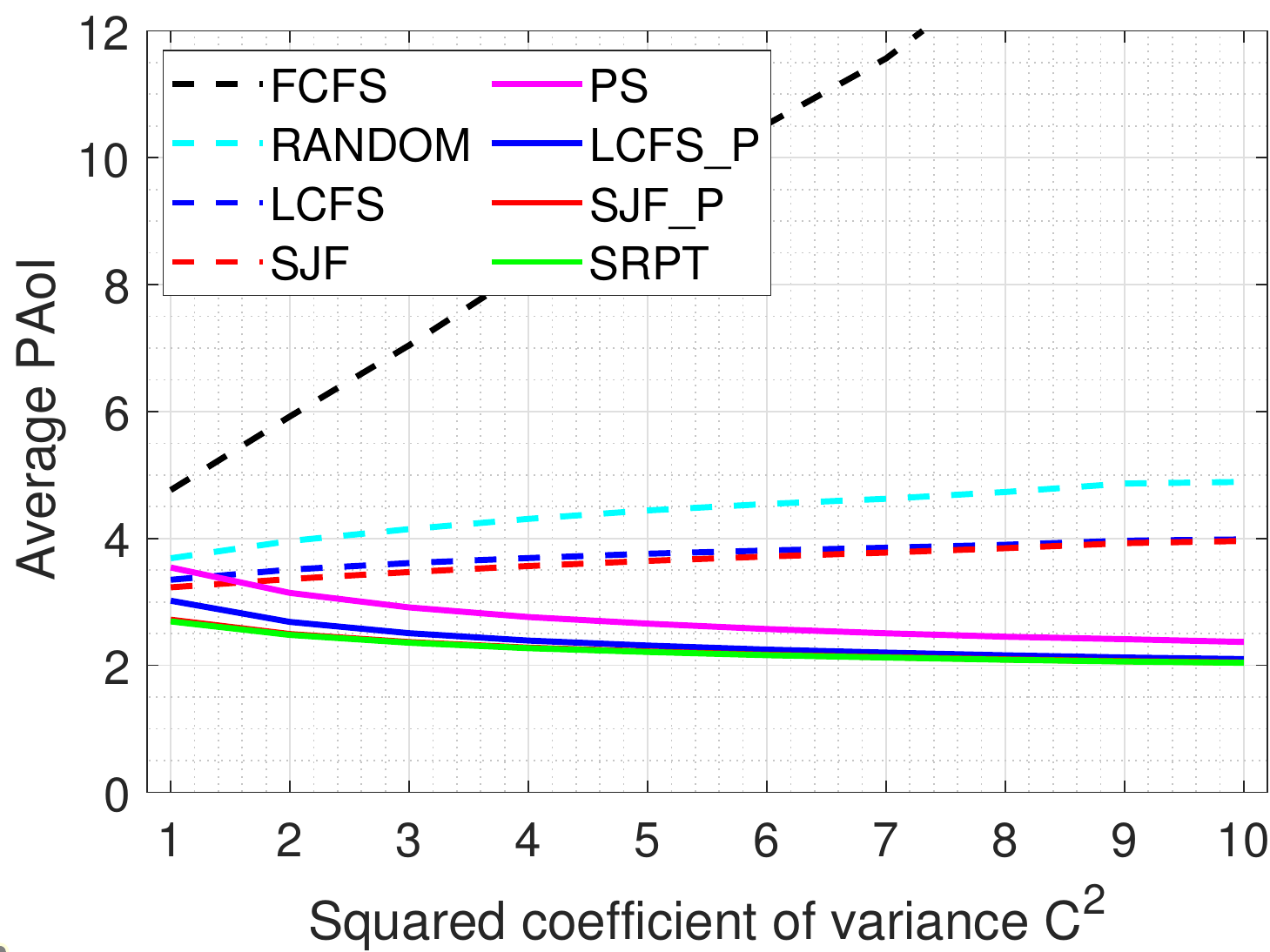}}
	\caption{Comparisons of the average PAoI performance under several common scheduling policies}
	\label{fig:tradition-PAoI}
\end{figure*}

The first type consists of policies that are non-preemptive and blind to the update size:

\begin{itemize}
	\item \emph{First-Come-First-Served (FCFS)}: When the server frees up, it chooses to serve the update that arrived first if any. 
	\item \emph{Last-Come-First-Served (LCFS)}: When the server frees up, it chooses to serve the update that arrived last if any.
	\item \emph{Random-Order-Service (RANDOM)}: When the server frees up, it randomly chooses one  update  to serve  if any.
\end{itemize}

The second type consists of policies that are non-preemptive and make scheduling decisions based on the update size:
\begin{itemize}
	\item \emph{Shortest-Job-First (SJF)}:  When the server frees up, it chooses to serve the update with the smallest size if any. 
\end{itemize}

The third type consists of policies that are preemptive and blind to the update size:
\begin{itemize}
	\item \emph{Processor-Sharing (PS)}: All the updates in the system are served simultaneously and equally (i.e., each update receives an equal fraction of the available service capacity).
	\item \emph{Preemptive Last-Come-First-Served (LCFS\_P)}: This is the preemptive version of the LCFS policy. Specifically, a preemption happens when there is a new update.
\end{itemize}

The fourth type consists of policies that are preemptive and make scheduling decisions based on the update size:
\begin{itemize}
    \item \emph{Preemptive Shortest-Job-First (SJF\_P)}: This is the preemptive version of the SJF policy. Specifically, a preemption happens when there is a new update that has the smallest size.
    \item \emph{Shortest-Remaining-Processing-Time (SRPT)}: When the server frees up, it chooses to serve the update with the smallest remaining size. In addition, a preemption happens only when there is a new update whose size is smaller than the remaining size of the update in service. 
\end{itemize}

Previous work (see, e.g., \cite[Section \uppercase\expandafter{\romannumeral7}]{harchol2013performance}) reveals that size-based policies can greatly improve the delay performance. Due to such results, we conjecture that size-based policies also achieve a better AoI performance given that the AoI is dominantly determined by the delay when the system load is high or when the size variability is large \cite{kaul2012real}. As we mentioned earlier, it is in general very difficult to obtain the exact expression of the average AoI except for some special cases (e.g., FCFS and LCFS) \cite{kaul2012real,inoue2018general}. Therefore, we attempt to investigate the AoI performance of size-based policies through extensive simulations.

In Figs.~\ref{fig:tradition-AoI} and \ref{fig:tradition-PAoI}, we present the simulation results of the average AoI and PAoI performance under the scheduling policies we introduced above, respectively. 
We conduct $50$ simulation runs and take the average values. In each simulation run, we consider a total number of ${10^5}$ updates to ensure that the steady state is reached. All the random numbers  are generated using the built-in random module in the Python standard library.
Here we assume that a single source generates updates according to a Poisson process with rate $\lambda$, and the update size is independent and identically distributed (\emph{i.i.d.}).
In Fig.~\ref{fig:tradition-exp-AoI}, we assume that the update size follows an exponential distribution with mean $1/\mu=1$. In Figs.~\ref{fig:tradition-wei-AoI}~and~\ref{fig:tradition-wei-variance-AoI}, we assume that the update size follows a Weibull distribution\footnote{The Weibull distribution is a heavy-tailed distribution with pdf $f(x;\alpha ,\beta ) = \frac{\alpha }{\beta }{(\frac{x}{\beta })^{\alpha  - 1}}{e^{ - {{(x/\beta )}^\alpha }}}$ for $x>0$, where $\alpha  > 0$ is the shape parameter and $\beta  > 0$ is the scale parameter.}
with mean $1/\mu=1$.
We define the squared coefficient of variation of the update size as ${C^{\rm{2}}} \triangleq  {\mathop{\rm Var}} \left( S \right){\rm{/}}  {\mathbb{E}} {\left[ S \right]^2}$, i.e., the variance normalized by the square of the mean \cite{harchol2013performance}. 
Hence, a larger ${C^{\rm{2}}}$ means a larger variability. In Fig.~\ref{fig:tradition-wei-AoI}, we fix ${C^{\rm{2}}=10}$ and change the value of system load $\rho$, while in Fig.~\ref{fig:tradition-wei-variance-AoI}, we fix system load $\rho=0.7$ and change the value of ${C^{\rm{2}}}$. 
Note that throughout the paper, these simulation settings are used as default settings unless otherwise specified.
In addition, the $95\% $ confidence intervals of Figs.~\ref{fig:tradition-AoI} and \ref{fig:tradition-PAoI} are also provided in Appendix~\ref{appendix:confidence-interval}.
We omit the $95\% $ confidence intervals for the other plots, as the margin of error is only a very small portion of the average (about $1\%$).

In the following, we will discuss key observations from the simulation results and propose useful guidelines for the design of AoI-efficient policies.

\begin{observation}
\label{obs:size_better}
Size-based policies achieve a better average AoI/PAoI performance than non-size-based policies in both non-preemptive and preemptive cases.
\end{observation}

In Fig.~\ref{fig:tradition-AoI}, we can see that for the non-preemptive case, SJF has a better average AoI performance than FCFS, RANDOM, and LCFS in various settings. Similarly, for the preemptive case, SJF\_P and SRPT have a better average AoI performance than PS and LCFS\_P. Similar observations can be made for the average PAoI performance in Fig.~\ref{fig:tradition-PAoI}.

\begin{observation}
\label{obs:pre_size_decreasing}
    Under preemptive, size-based policies, the average AoI/PAoI decreases as the system load increases.
\end{observation}

In Figs.~\ref{fig:tradition-exp-AoI}~and~\ref{fig:tradition-wei-AoI}, we can see that under SJF, SJF\_P, and SRPT, the average AoI decreases as the system load $\rho$ increases. There are two reasons. First, when $\rho$ increases, there will be more updates with small size arriving to the queue. Therefore, size-based policies that prioritize updates with small size lead to more frequent AoI drops. Second, preemption operations prevent fresh updates from being blocked by a large or stale update in service.
Similar observations can be made for the average PAoI performance in Figs.~\ref{fig:tradition-exp-PAoI}~and~\ref{fig:tradition-wei-PAoI}.

Observations~\ref{obs:size_better}~and~\ref{obs:pre_size_decreasing} lead to the following guideline:
\begin{guideline}
\label{guide:use-size-info}
When the update-size information is available, one should prioritize updates with small size.
\end{guideline}

However, in certain application scenarios, the update-size information may not be available or is difficult to estimate. Hence, the scheduling decisions have to be made without the update-size information. In such scenarios, we make the following observations from Figs.~\ref{fig:tradition-AoI}~and~\ref{fig:tradition-PAoI}.

\begin{observation}
\label{obs:LCFS_better}
LCFS and LCFS\_P achieve the best average AoI performance among non-preemptive, non-size-based policies and preemptive, non-size-based policies, respectively. 
\end{observation}

\begin{observation}
\label{obs:LCFS_P_decreasing}
Under LCFS\_P, the average AoI/PAoI decreases as the system load increases.
\end{observation}

\begin{figure}[!t]
    \centering
    \includegraphics[scale=0.6]{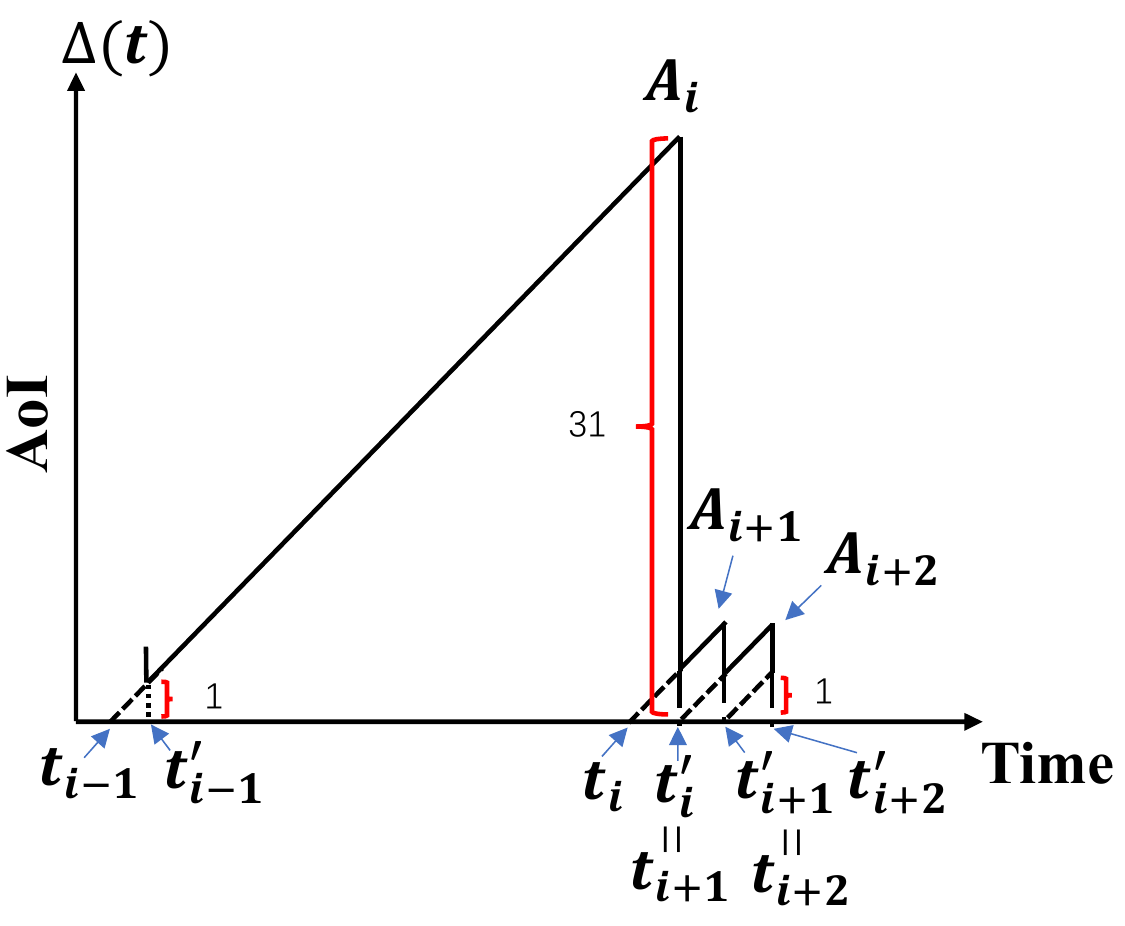}
    \caption{An example of the AoI/PAoI evolution where the interarrival time has a large variability}
    \label{fig:paoi_less_aoi}
\end{figure}

Observations~\ref{obs:LCFS_better}~and~\ref{obs:LCFS_P_decreasing} have also been made in previous work~\cite{costa2016age,yates2018age,bedewy2016optimizing}.  It is quite intuitive that when the update-size information is unavailable, one should give a higher priority to more recent updates. This is because while all the updates have the same expected service time, the most recent update arrives the last and thus leads to the smallest AoI once delivered. Therefore, Observations~\ref{obs:LCFS_better}~and~\ref{obs:LCFS_P_decreasing} lead to the following guideline:

\begin{guideline}
    When the update-size information is unavailable, one should prioritize recent updates. 
    \label{guide:arrival-time-info}
\end{guideline}

Note that Observations~\ref{obs:pre_size_decreasing}~and~\ref{obs:LCFS_P_decreasing} also suggest that under preemptive policies, the average AoI/PAoI decreases as the system load $\rho$ increases. This is because preemptions prevent fresh updates from being blocked by a large or stale update in service. In addition, we have also observed the following nice properties of preemptive policies.

\begin{observation}
\label{obs:pre-good}
Not only do preemptive policies achieve a better average AoI/PAoI performance than non-preemptive policies, but they are also less sensitive when the update-size variability changes, i.e., they are more robust.
\end{observation}

In Figs.~\ref{fig:tradition-exp-AoI}~and~\ref{fig:tradition-wei-AoI}, we can see that preemptive policies (e.g., LCFS\_P, SJF\_P, and SRPT) generally have a better average AoI performance than non-preemptive ones (e.g., FCFS, RANDOM, LCFS, and SJF), especially when the system load is high. In Fig.~\ref{fig:tradition-wei-variance-AoI}, we can see that the advantage of preemptive policies becomes larger as the update-size variability (i.e., ${C^{\rm{2}}}$) increases. Moreover, the AoI performance of preemptive policies is only very slightly impacted when the update-size variability changes, while that of non-preemptive policies varies significantly. Therefore, Observations~\ref{obs:pre_size_decreasing}, \ref{obs:LCFS_P_decreasing}, and \ref{obs:pre-good} lead to the following guideline:

\begin{guideline}\label{gui:preemption}
Service preemption should be employed when it is allowed.
\end{guideline}

Note that above observations not only hold for the M/G/$1$ queue, but also can be made for the G/G/$1$ queue. More simulation results for the G/G/$1$ queue (i.e., Figs.~\ref{fig:tradition-AoI-2}-\ref{fig:all-combined-paoi}) can be found in Appendix~\ref{appendix:more_simu}. In addition, we make the following interesting observations regarding the average PAoI and AoI in a G/G/$1$ queue.

\begin{observation}
\label{obs:PAoI-smaller}
The average PAoI could be much smaller than the average AoI when the interarrival time has a large variability. 
\end{observation}
In Figs.~\ref{fig:tradition-wei-exp-AoI}~and~\ref{fig:tradition-wei-exp-PAoI}, we can see that the average PAoI is much smaller than the average AoI for all the common scheduling policies we considered. This is due to the interarrival time has a large variability. 
We present an example in Fig.~\ref{fig:paoi_less_aoi} to illustrate that this phenomenon comes from the large variability of the interarrival time. 
We consider three updates: the $i$-th, the $(i+1)$-st and $(i+2)$-nd updates, which are served in sequence during $({t'_{i - 1}},{t'_{i + 2}})$. 
Their interarrival times are as follows: ${t_i} - {t_{i - 1}} = 30$, ${t_{i + 1}} - {t_i} = 1$, and ${t_{i + 2}} - {t_{i + 1}} = 1$; and their system times are as follows: ${t'_i} - {t_i} = 1$, ${t'_{i + 1}} - {t_{i + 1}} = 1$, and ${t'_{i + 2}} - {t_{i + 2}} = 1$. In addition, we also assume ${t'_{i - 1}} - {t_{i - 1}} = 1$. Therefore, the average AoI and the average PAoI during $({t'_{i - 1}},{t'_{i + 2}})$ are $\frac{{{{31}^2} + {2^2} + {2^2} - 3 \times {1^2}}}{{2 \times (30 + 1 + 1)}} \approx 15.09$ and $\frac{{31 + 2 + 2}}{3} \approx 11.67$, respectively. In this case, the average PAoI is indeed smaller than the average AoI. 

The importance of Observation~\ref{obs:PAoI-smaller} can be summarized as follows. First, in certain settings (e.g., where the interarrival time  has a large variability), the AoI can actually be higher than the PAoI. This observation is counterintuitive, given that the computation of the average PAoI includes the peak values of the AoI only. Second, given that the AoI and the PAoI exhibit different relationships in different settings, an AoI-efficient scheduling policy may not necessarily achieve a desired PAoI performance, and vice versa. In other words, one must carefully study the design of AoI-efficient scheduling policies with different goals in mind (i.e., minimizing the AoI or the PAoI).


\begin{observation}
\label{obs:PAoI-flat}
The average PAoI performance of several non-preemptive policies (such as RANDOM, LCFS, and SJF) is insensitive to the update-size variability. 
\end{observation}

\begin{figure}[!t]
    \centering
    \includegraphics[scale=0.6]{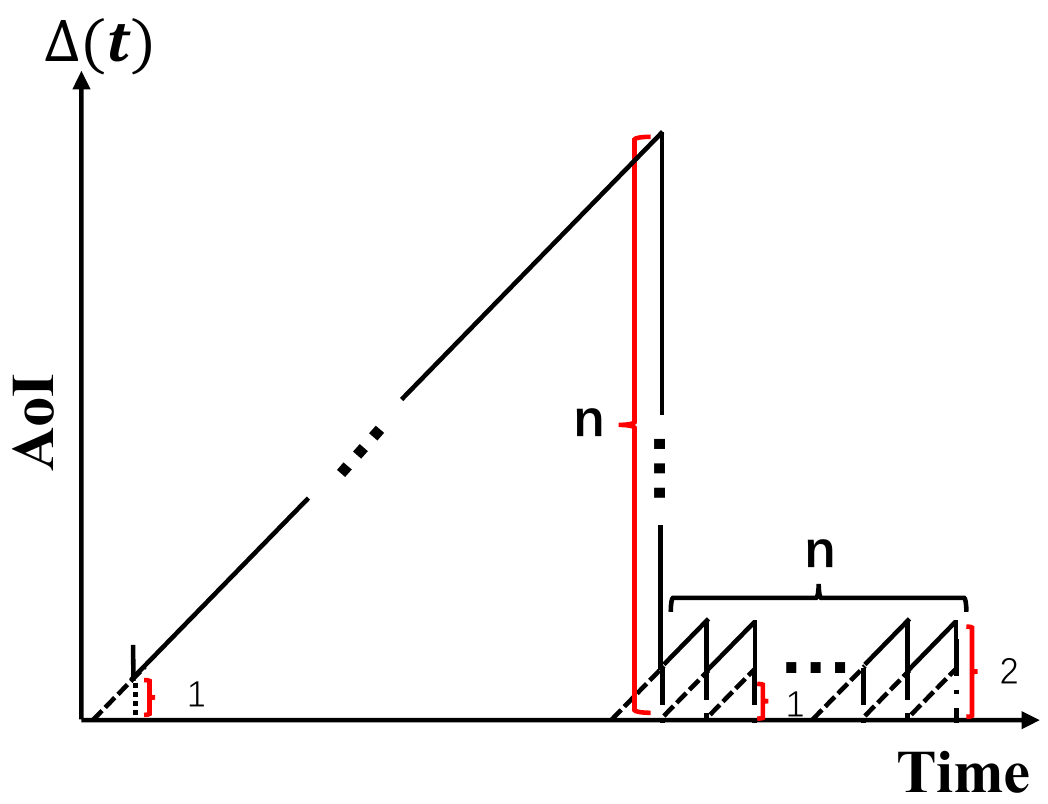}
    \caption{An example of the AoI/PAoI evolution where the service time has a large variability}
    \label{fig:obs7_explanation}
\end{figure}

In Fig.\ref{fig:tradition-wei-variance-PAoI}, we observe that while the average PAoI performance of FCFS is sensitive to the update-size variability, under several non-preemptive policies (such as RANDOM, LCFS, and SJF), the average PAoI performance is much less sensitive. An explanation for this observation is the following.  \\
\indent First, we explain why the PAoI under FCFS is still sensitive to the update-size variability. Note that a key difference between FCFS and other non-preemptive policies is that under FCFS, every update leads to an AoI drop and thus corresponds to an AoI peak. When a large update is in service, it will block all the following updates that are waiting in the queue, which results in a large delay for all such updates and thus a large PAoI corresponding to these updates. In contrast, under RANDOM, LCFS, and SJF, the impact of such a blocking issue is minimal for the updates that lead to an AoI drop. \\
\indent Next, we explain why the PAoI has a different behavior than the AoI under RANDOM, LCFS, and SJF. We first consider LCFS. In the setting we consider, there is a high chance that the newest update has a small size. Serving such small-size updates leads to a small PAoI. When the newest update has a large size, the corresponding PAoI would also be large. However, this happens less often. Therefore, the AoI trajectory would consist of a smaller percentage of large AoI peaks with many small AoI peaks in between. As the update-size variability increases, there will be fewer but larger AoI peaks. In such cases, while the average AoI is sensitive to the large AoI peaks (which comes from the large update-size variability), the average PAoI is much less sensitive. To illustrate this fact, we provide an example in Fig.~\ref{fig:obs7_explanation}, where there is a large update of size $n-1$, immediately followed by $n$ small updates of size $1$. In this case, we can compute the average AoI as $\Delta {\rm{ = }}\left[ {1 \times (\frac{{{n^2}}}{2} - \frac{{{1^2}}}{2}) + n \times \left( {\frac{{{2^2}}}{2} - \frac{{{1^2}}}{2}} \right)} \right]/\left( {\left( {n - 1} \right) + n} \right) = \frac{{{n^2} + 3n - 1}}{{4n - 2}} = {\cal O}(n)$ and compute the average PAoI as $A = \frac{{n + 2 \times n}}{{n + 1}} = \frac{{3n}}{{n + 1}}= {\mathcal{O}}(3)$.
This example shows that a larger update-size variability (i.e., a larger $n$ in this example) results in a larger average AoI but only minimally affects the average PAoI. A similar explanation also applies to SJF and RANDOM.

\begin{figure}[!t]
	\centering
	\includegraphics[scale=0.5]{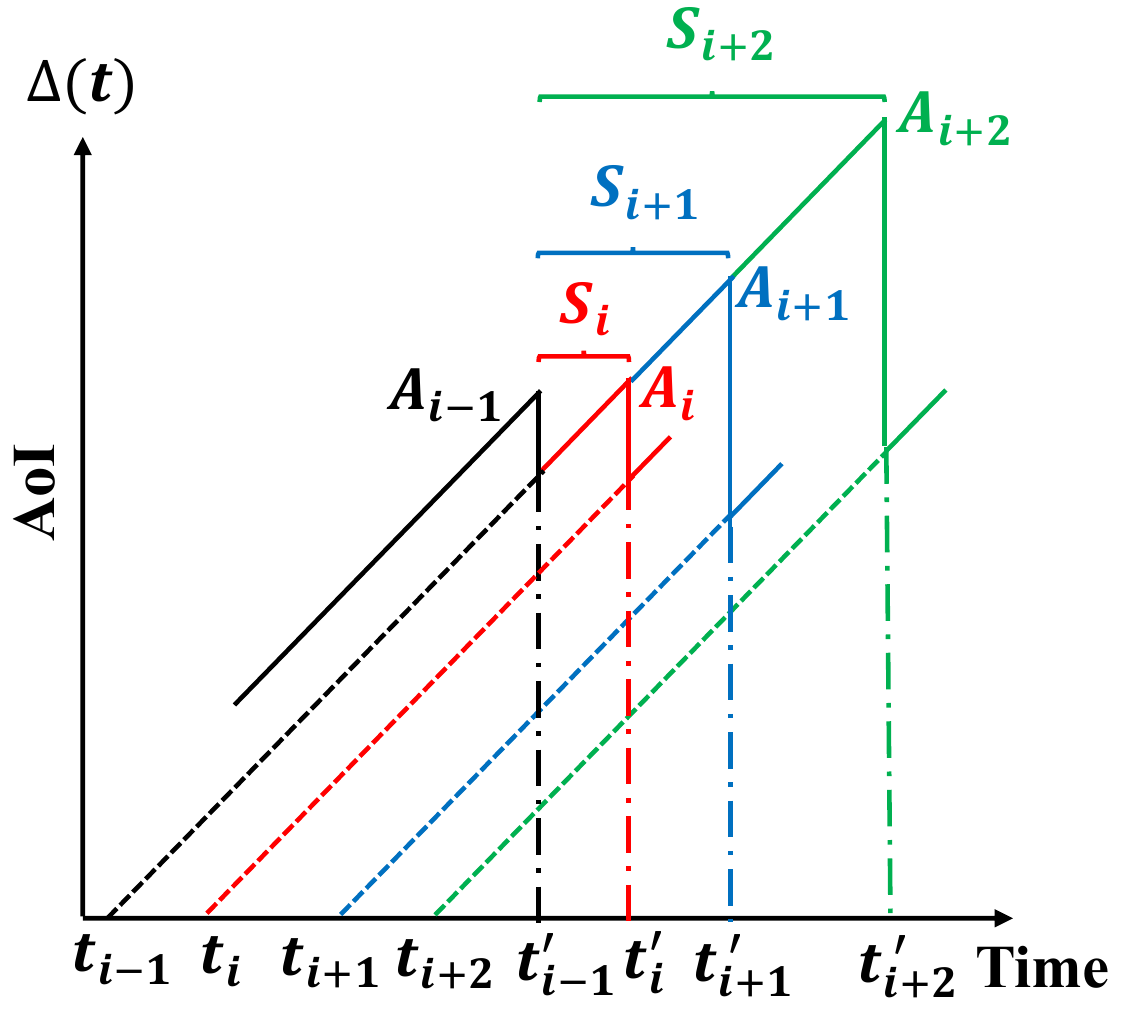}
	\caption{The AoI evolution under three AoI-based policies: ADE (red), ADS (blue), and ADM (green)}
	\label{fig:ADF-ADS-ADM}
\end{figure}

\section{\uppercase{AoI-based  policies}}
\label{sec:drop-to-point}

In Section \ref{sec:common-scheduling}, we have demonstrated that size-based policies achieve a better average AoI/PAoI performance than non-size-based policies. However, size-based policies do not utilize the arrival-time information, which also plays an important role in reducing the AoI. 
In this section, we propose three AoI-based scheduling policies, which leverage both the update-size and arrival-time information to reduce the AoI.
Our simulation results show that these AoI-based policies outperform non-AoI-based policies.

We begin with the definitions of three AoI-based policies that \emph{attempt to optimize the AoI at a specific future time instant} from three different perspectives: 

\begin{itemize}
\item \emph{AoI-Drop-Earliest (ADE)}: When the server frees up, it chooses to serve an update such that once it is delivered, the AoI drop as soon as possible.

\item \emph{AoI-Drop-to-Smallest (ADS)}: When the server frees up, it chooses to serve an update such that once it is delivered, the AoI drops to a value as small as possible.

\item \emph{AoI-Drop-Most (ADM)}: When the server frees up, it chooses to serve an update such that once it is delivered, the AoI drops as much as possible.
\end{itemize}
If all updates waiting in the queue are obsolete, then the above policies choose to serve an update with the smallest size.

Although all of these AoI-based policies are quite intuitive, they behave very differently.
In order to explain the differences of these AoI-based policies, we present an example in Fig.~\ref{fig:ADF-ADS-ADM} to show how the AoI evolves under these policies.
Suppose that when the ($i-1$)-st update is being served, three new updates (i.e., the $i$-th, ($i+1$)-st, and ($i+2$)-nd updates) arrive in sequence at times ${t_i}$, ${t_{i + 1}}$, and ${t_{i + 2}}$, respectively. The sizes of these updates satisfy ${S_i} < {S_{i + 1}} < {S_{i + 2}}$. When the server frees up after it finishes serving the ($i-1$)-st update at time $t^{\prime}_{i-1}$, ADE, ADS, and ADM choose to serve the $i$-th, ($i+1$)-st, and ($i+2$)-nd updates, respectively. This is because serving the $i$-th update leads to the earliest AoI drop at time $t^{\prime}_i$ (following the red curve), serving the ($i+1$)-st update leads to the AoI dropping to the smallest at time $t^{\prime}_{i+1}$ (following the blue curve), and serving the ($i+2$)-nd update leads to the largest AoI drop at time $t^{\prime}_{i+2}$ (following the green curve).
Clearly, ADE, ADS and ADM aim to optimize AoI at a specific future time instant (i.e., the future delivery time of chosen update) with  different myopic goals. Note that at first glance, ADS and ADM may look the same. Indeed, they would be equivalent if the events of AoI drop have happened at the same time instant. However, these two policies are different as the time instants at which the AoI drops are not necessarily the same (e.g., $t^{\prime}_{i+1}$ vs. $t^{\prime}_{i+2}$ in Fig.~\ref{fig:ADF-ADS-ADM}).
In addition, ADE and SJF may also look the same at first glance. Indeed, these two policies would make the same decision (i.e., choose the smallest update to serve) when the smallest update leads to an AoI drop. However, they make different decisions when the smallest update does not lead to an AoI drop. An example is provided in Fig.~\ref{fig:ADEvsSJF} to illustrate the key difference. In Fig.~\ref{fig:ADEvsSJF}, after the $(i-1)$-st update completes service at time $t'_{i-1}$, there are two updates waiting to be served: the $(n-2)$-nd update and the $i$-th update. Suppose that the update size and the arrival time of these two updates satisfy the following: ${S_{i - 2}} < {S_i}$ and ${t_{i - 2}} < {t_{i - 1}} < {t_i}$. Clearly, ADE chooses to serve the $i$-th update that leads to an earlier AoI drop (see Fig.~\ref{fig:ADEvsSJF_ADE}), while SJF chooses to serve the $(i-2)$-nd update that has a smaller size (see Fig.~\ref{fig:ADEvsSJF_SJF}).

Next we conduct extensive simulations to investigate the AoI performance of these AoI-based policies. 
In Fig.~\ref{fig:wei-fsm-area-variance}, we present the simulation results of the average AoI performance of the AoI-based policies compared to a representative arrival-time-based policy (i.e., LCFS) and a representative size-based-policy (i.e., SJF). All the policies considered here are non-preemptive; the preemptive cases will be discussed in Section~\ref{sec:pre-infor}.

In Fig.~\ref{fig:exp-fsm-area-ratio-aoi}, we observe that most AoI-based policies are slightly better than non-AoI-based policies, although their performances are very close. Among the AoI-based policies, ADE is the best, ADM is the worst, and ADS is in-between. This is not surprising that ADM is the worst: although ADM has the largest AoI drop, this is at the cost that it may have to wait until the AoI become large first. ADE being the best suggests that giving a higher priority to small updates (so that the AoI drops as soon as possible) is a good strategy.
In Figs.~\ref{fig:wei-fsm-area-ratio-aoi}~and~\ref{fig:fsd-area-variance-aoi}, similar observations can be made for update size following Weibull distributions.

\begin{figure}[!t]
    \centering
    \subfigure[ADE]{
    \begin{minipage}[t]{1\linewidth}
    \label{fig:ADEvsSJF_ADE}
    \centering
    \includegraphics[width=2in]{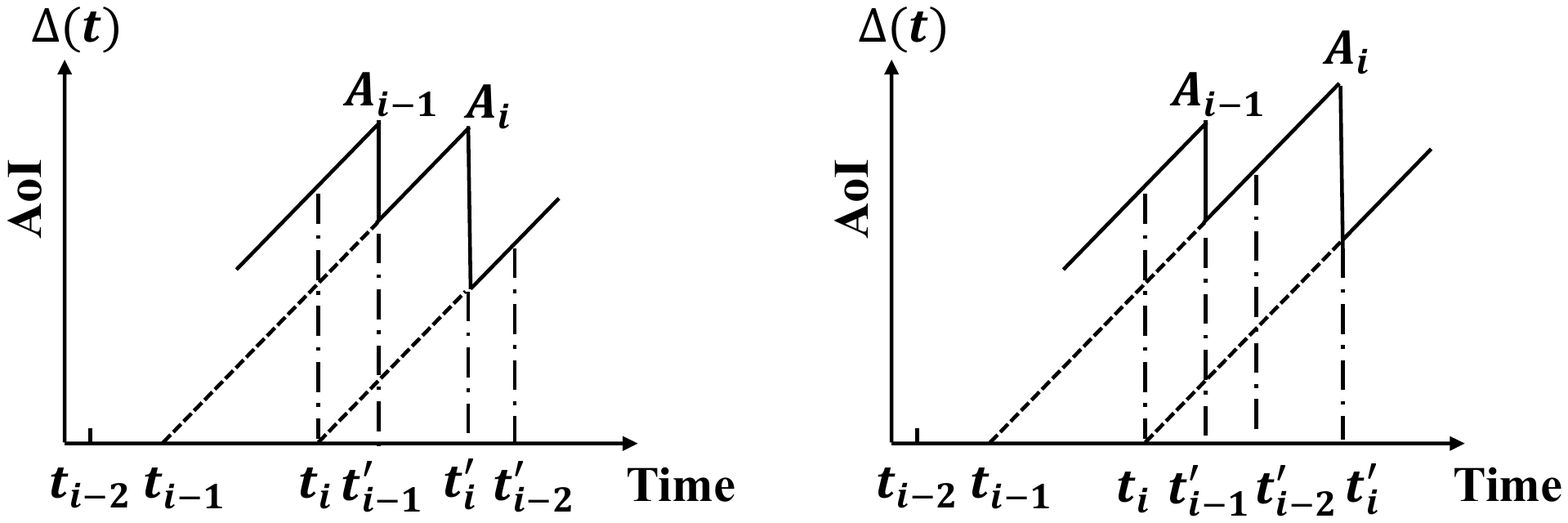}
    \end{minipage}}
    \subfigure[SJF]{
    \begin{minipage}[t]{1\linewidth}
    \label{fig:ADEvsSJF_SJF}
    \centering
    \includegraphics[width=2in]{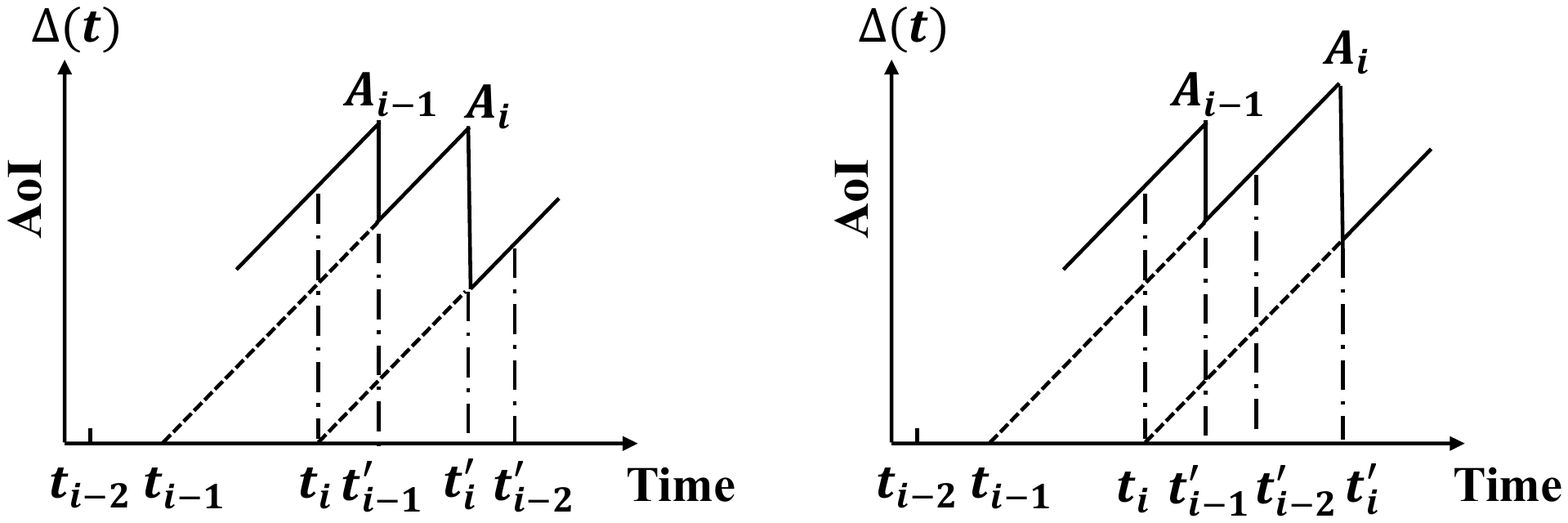}
    \end{minipage}}
    \caption{An example of the AoI evolution under ADE and SJF}
    \label{fig:ADEvsSJF}
\end{figure}

The above observations lead to the following guideline:
\begin{guideline}
	Leveraging both the update-size and arrival-time information can further improve the AoI performance. However, the benefit seems marginal.
	\label{guide:use-size-arrival-info}
\end{guideline}

\begin{figure*}[!t]
    \centering
    \subfigure[Exponential: $\mu=1$]{
		\label{fig:exp-fsm-area-ratio-aoi} 
		\includegraphics[width=0.322\textwidth]{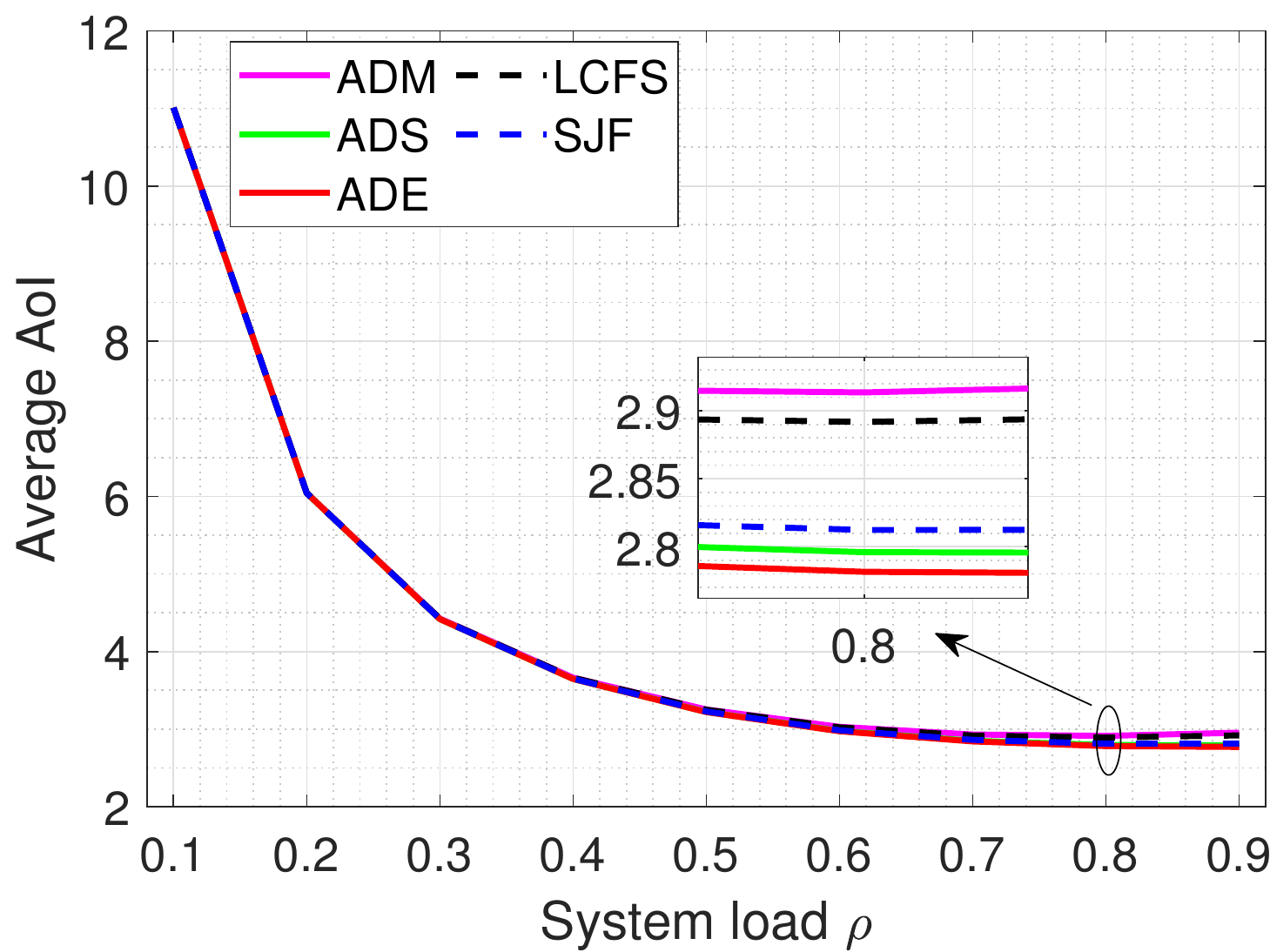}}
	\subfigure[Weibull: $\mu=1$ and ${C^{\rm{2}}}{\rm{ = 10}}$]{
		\label{fig:wei-fsm-area-ratio-aoi} 
		\includegraphics[width=0.322\textwidth]{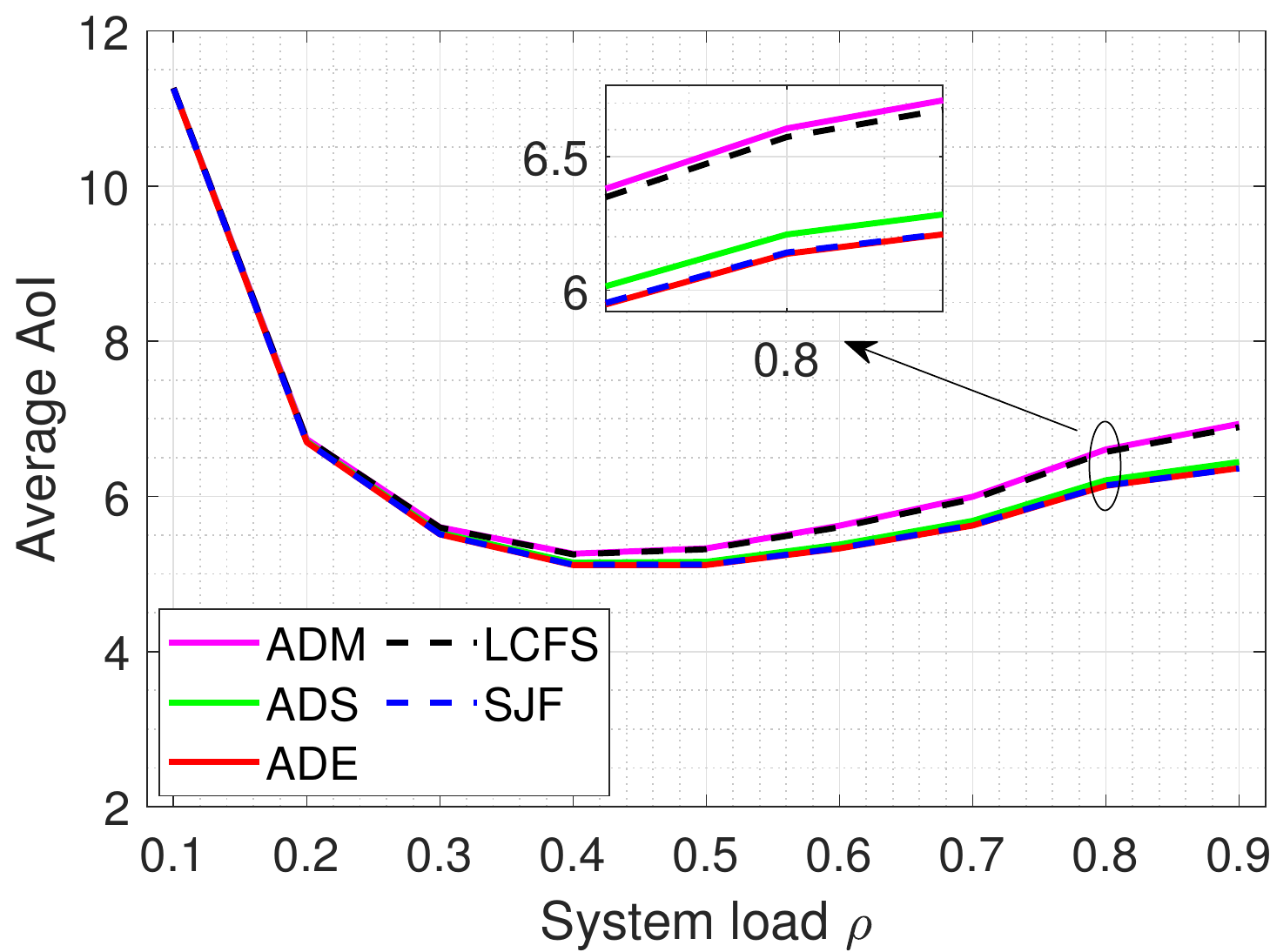}}
    \subfigure[Weibull: $\mu=1$ and $\rho {\rm{ = 0}}{\rm{.7}}$]{
		\label{fig:fsd-area-variance-aoi} 
		\includegraphics[width=0.322\textwidth]{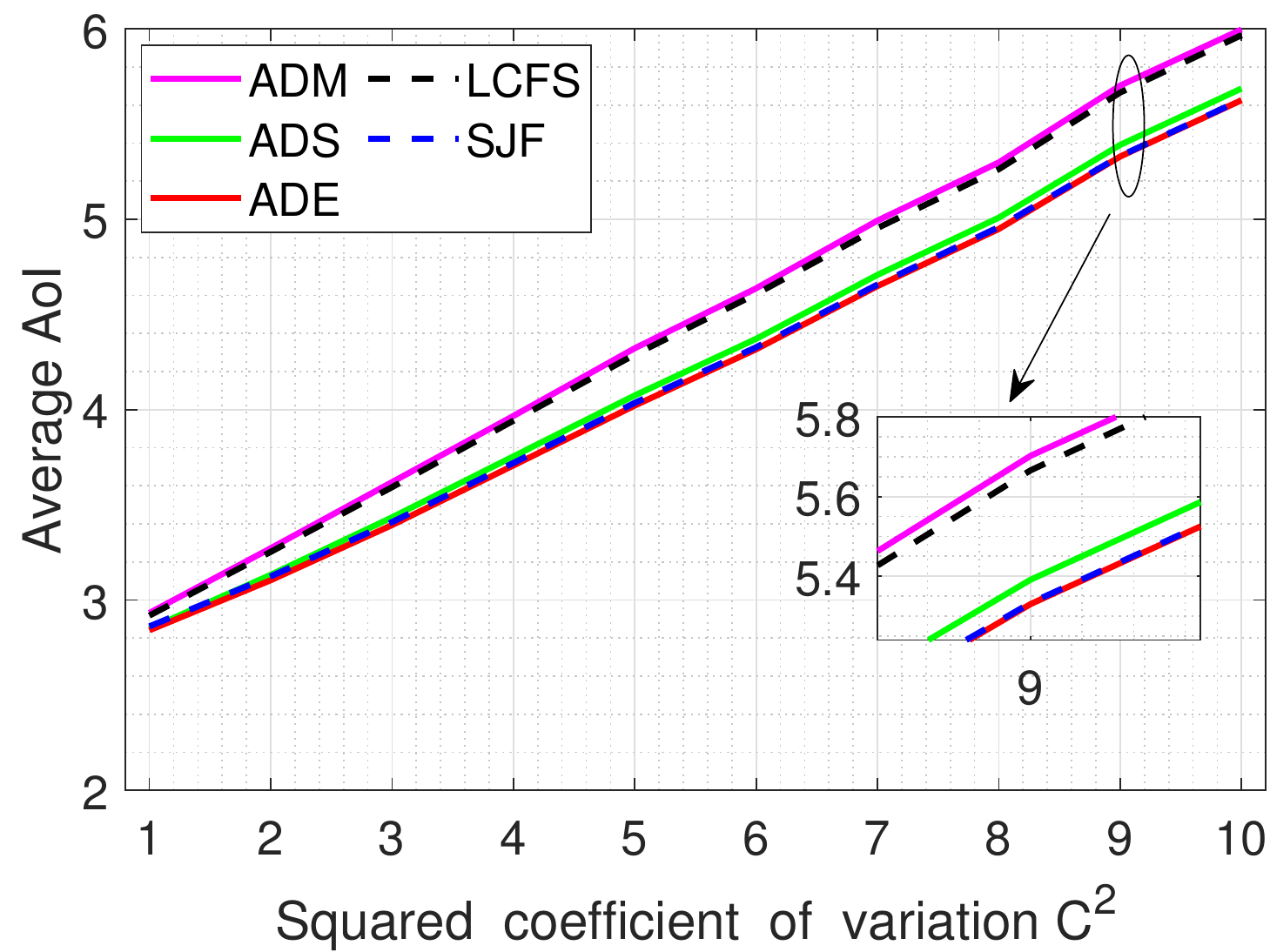}}
	\caption{Comparisons of the average AoI performance: AoI-based policies vs. non-AoI-based policies}
	\label{fig:wei-fsm-area-variance}
\end{figure*}

\begin{figure*}[!t]
    \centering
    \subfigure[Exponential: $\mu=1$]{
		\label{fig:exp-fsm-area-ratio-paoi} 
		\includegraphics[width=0.322\textwidth]{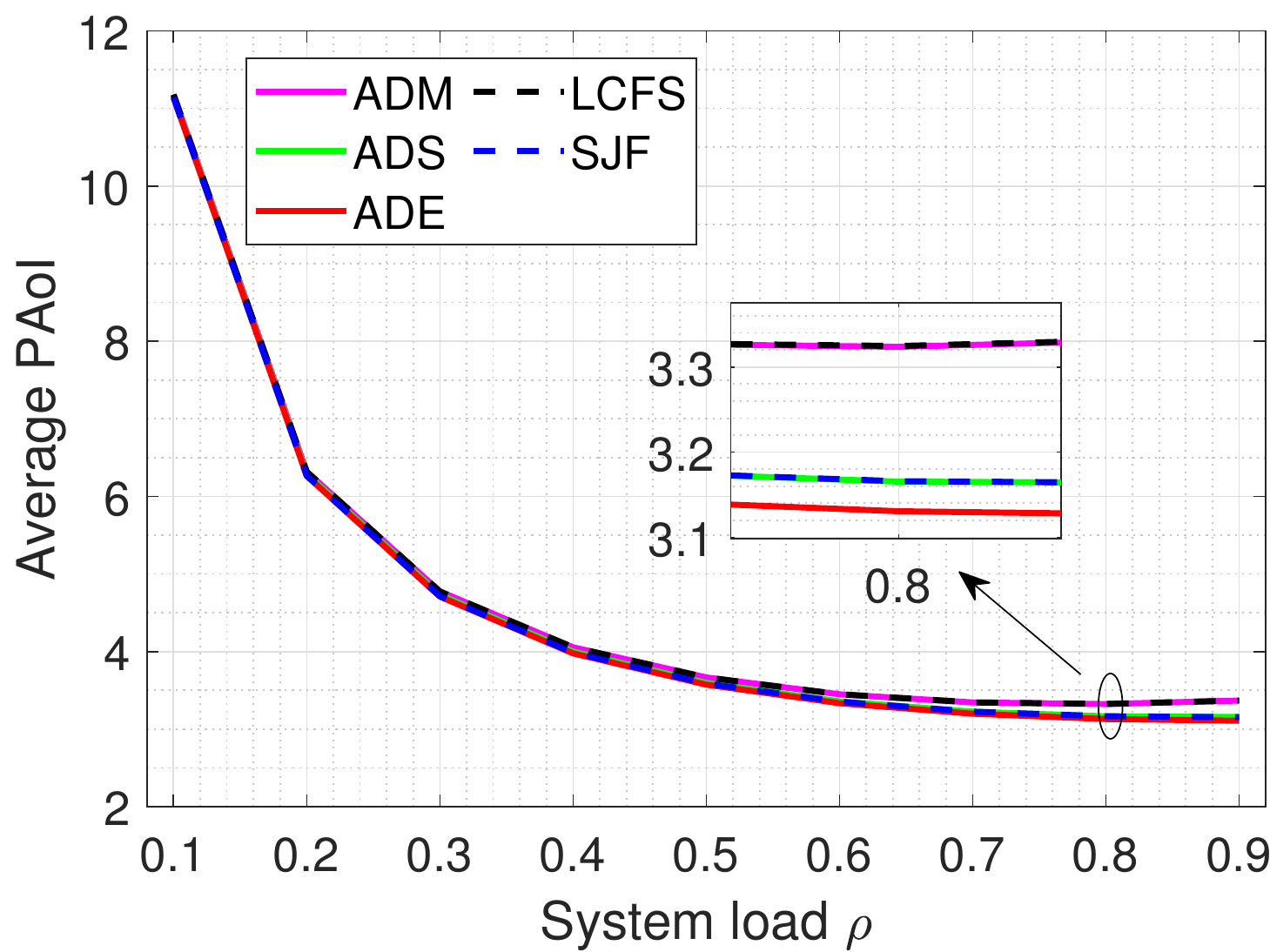}}
	\subfigure[Weibull: $\mu=1$ and ${C^{\rm{2}}}{\rm{ = 10}}$]{
		\label{fig:wei-fsm-area-ratio-paoi} 
		\includegraphics[width=0.322\textwidth]{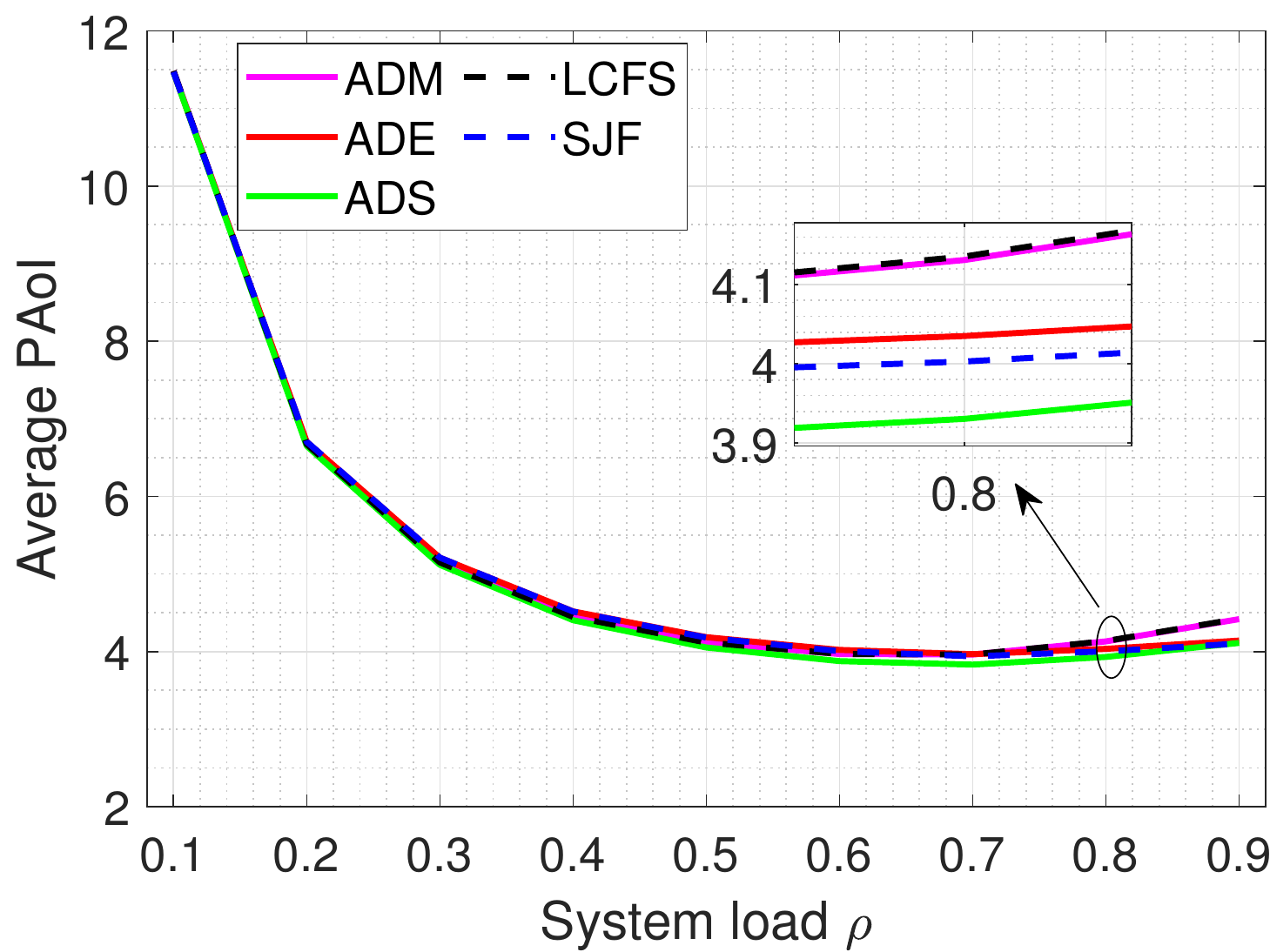}}
    \subfigure[Weibull: $\mu=1$ and $\rho {\rm{ = 0}}{\rm{.7}}$]{
		\label{fig:fsd-area-variance-paoi} 
		\includegraphics[width=0.322\textwidth]{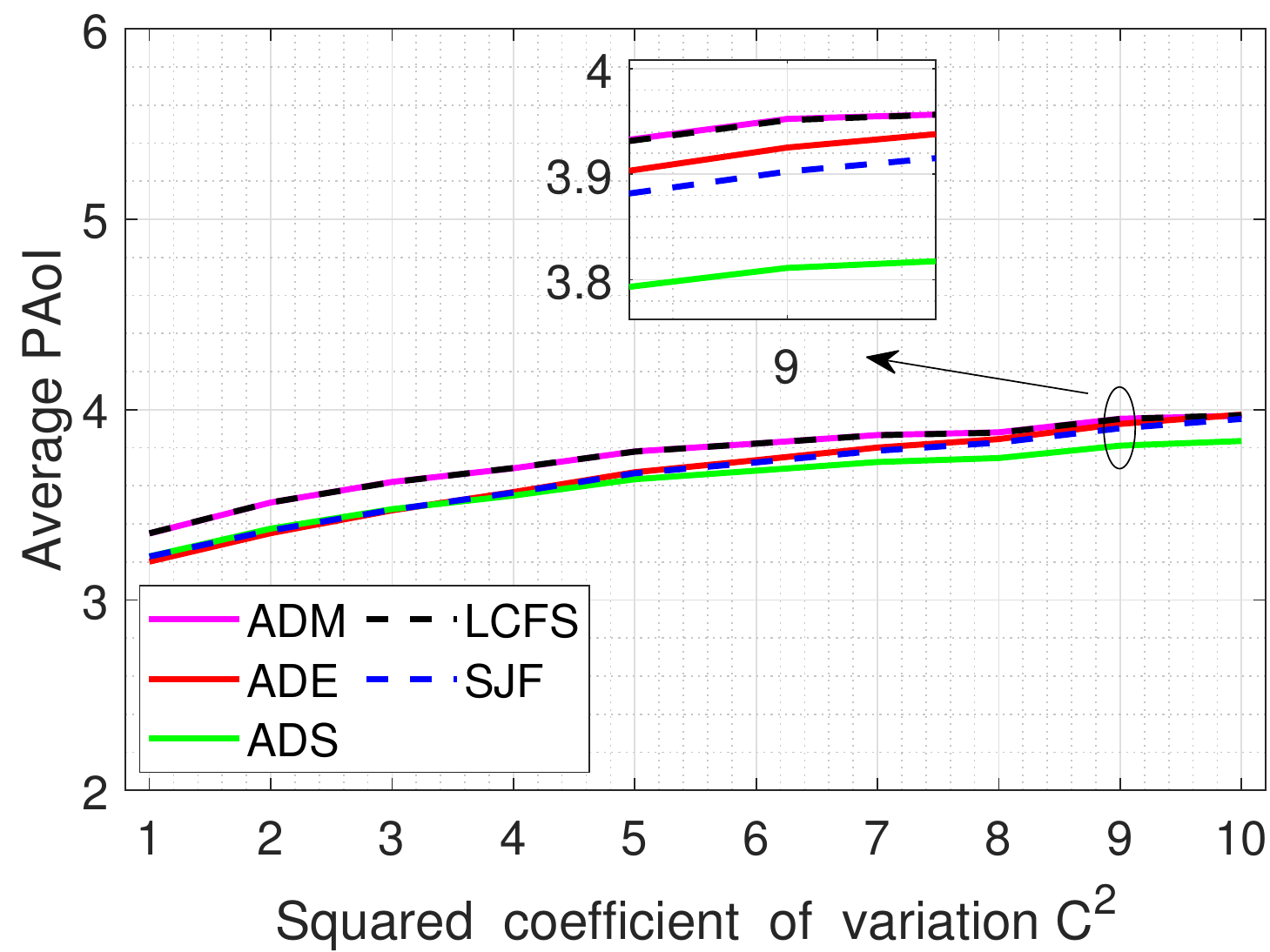}}
	\caption{Comparisons of the average PAoI performance: AoI-based policies vs. non-AoI-based policies}
	\label{fig:wei-fsm-area-variance-paoi}
\end{figure*}

\section{\uppercase{Preemptive, informative, AoI-based Policies}}
\label{sec:pre-infor}

In Section \ref{sec:common-scheduling}, we have observed that preemptive policies have several advantages and perform better than non-preemptive policies. In this section, we first demonstrate that policies that prioritize informative updates (i.e., those that can lead to AoI drops once delivered) perform  better than non-informative policies. Then, by integrating the guidelines we have, we consider preemptive, informative, AoI-based policies and evaluate their performances through simulations.

\subsection{Informative Policies}
\label{subsec:infor-policies}
As far as the AoI is concerned, there are two types of updates: informative updates and non-informative updates \cite{kam2013age}.
\emph{Informative updates lead to AoI drops once delivered while non-informative updates do not. }
In some applications, such as autonomous vehicles and stock quotes,  it is reasonable to discard non-informative updates (which do not help reduce the AoI but may block new updates).
In this subsection, we introduce the ``informative'' versions of various policies, which prioritize informative updates and discards non-informative updates. 
Then, we use simulation results to demonstrate that informative policies generally have a better average AoI/PAoI performance than the original (non-informative) ones. Furthermore, we rigorously prove that in a G/M/$1$ queue,  the informative version of LCFS is stochastically better than the original LCFS policy.

We use $\pi \_{\rm{I}}$ to denote the informative version\footnote{For simplicity, we omit the additional ``\_" in the policy name if policy $\pi$ is a preemptive policy ending with ``\_P". For example, we use LCFS\_PI to denote the informative version of LCFS\_P.} of policy $\pi$.
All the scheduling policies we consider have their informative versions. In some cases, the informative version is simply the same as the original policy (e.g., FCFS and LCFS\_P).

In Fig.~\ref{fig:aoi-ratio-com-info}, we show the simulation results of the average AoI performance of several informative policies compared to their non-informative counterparts. 
In order to evaluate the benefit of informative policies, we plot the informative AoI gain, which is the ratio of the difference between the average AoI of the non-informative version and the informative version to the average AoI of the non-informative version. Hence, a larger informative gain means a larger benefit of the informative version. One important observation from Fig.~\ref{fig:aoi-ratio-com-info} is as follows.

\begin{observation}
	\label{obs:informative}
	Informative policies achieve a better average AoI performance than their non-informative counterparts.
	The informative gain is larger for non-preemptive policies and increases as the system load increases.
\end{observation}
 
Intuitively, informative policies are expected to outperform their non-informative counterparts because serving non-informative updates cannot reduce the AoI but may block new updates.
The simulation results verify this intuition as the informative AoI gain is always non-negative. 
Second, we can see that most non-preemptive policies (e.g., RANDOM, LCFS, and SJF) benefit more from prioritizing informative updates.  
Third, as the system load $\rho$ increases, the informative AoI gain increases under most considered policies, especially those non-preemptive ones. This is because as the system load increases, the number of non-informative updates also increases, which has a larger negative impact on the AoI performance for non-preemptive, non-informative policies.

Observation~\ref{obs:informative} leads to the following guideline:
\begin{guideline}
	The server should prioritize informative updates and discard non-informative updates when it is allowed.
	\label{guide:informative}
\end{guideline}

Based on Observation~\ref{obs:informative}, we conjecture that an informative policy is as least as good as its non-informative counterpart. As a preliminary result, we prove that this conjecture is indeed true for LCFS in a G/M/1 queue. In the following, we introduce the stochastic ordering notion, which will be used in the statement of Proposition~\ref{pro:age-optimality}.

\begin{definition}
    Stochastic Ordering of Stochastic Processes \cite[Ch.6.B.7]{shaked2007stochastic}:  Let $\{ X(t),t \in [0,\infty )\} $ and $\{ Y(t),t \in [0,\infty )\} $ be two stochastic processes. Then, $\{ X(t),t \in [0,\infty )\} $ is said to be stochastically less than $\{ Y(t),t \in [0,\infty )\}$, denoted by $\{ X(t),t\in[0,\infty )\} {\le _{{\rm{st}}}}\{ Y(t),t \in [0,\infty )\}$, if, for all choices of integer $n$ and $t_{1}<t_{2}<\dots<t_{n}$ in $[0, \infty)$, the following holds for all upper sets\footnote{A set ${S^U} \subseteq {{\mathbb{R}}^n}$ is an upper set if $\vec y \in {S^U}$ whenever $\vec y \ge \vec x$ and $\vec x \in {S^U}$, where $\vec x = ({x_1}, \ldots ,{x_n})$ and $\vec y = ({y_1}, \ldots ,{y_n})$ are two vectors in ${\mathbb{R}^n}$ and $\vec y \ge \vec x$ if ${y_i} \ge {x_i}$ for all $i = 1,2, \ldots ,n$.} ${S^U} \subseteq {{\mathbb R}^n}$:
    \begin{equation}
        {\mathbb P}( {\vec X \in {S^U}} ) \le  {\mathbb P}( {\vec Y \in {S^U}}),
        \label{eq:stoc}
    \end{equation}
    where $\vec{X} \triangleq \left({X( {{t_1}}), X( {{t_2}}), \dots, X( {{t_n}})} \right)$ and $\vec{Y} \triangleq \left( {Y( {{t_1}} ) ,Y( {{t_2}}) , \dots, Y( {{t_n}})} \right)$.
    Stochastic equality can be defined in a similar manner and is denoted by $\{ X(t),t \in [0,\infty )\} {{\rm{ = }}_{{\rm{st}}}}\{ Y(t),t \in [0,\infty )\} $.
\end{definition} 

Roughly speaking, Eq.~(\ref{eq:stoc}) implies that $\vec X$ is less likely than $\vec Y$ to take on large values, where ``large'' means any value in an upper set ${S^U}$. We also use $\Delta_{\pi}(t)$ to denote the AoI process under policy $\pi$. 
Furthermore, we define a set of parameters ${\mathcal{I}} = \{ n,({t_i})_{i = 1}^n\}$, where $n$ is the number of updates and ${t_i}$ is the generation time of update $i$.
Having these definitions and notations, we are now ready to state Proposition~\ref{pro:age-optimality}.

\begin{proposition}
	\label{pro:age-optimality}
	In a G/M/1 queue, for all $\mathcal{I}$, the AoI under LCFS\_I is stochastically smaller than that under LCFS, i.e., 
	\begin{equation}
	[\left\{ {{\Delta _{{\rm{LCFS}}\_\rm{I}}}\left( t \right),t \in \left[ {0,\infty } \right)} \right\}|{\mathcal{I}}]{ \le _{{\rm{st}}}}[\left\{ {{\Delta _{{\rm{LCFS}}}}\left( t \right),t \in \left[ {0,\infty } \right)} \right\}|{\mathcal{I}}].
	\label{eq:aoi-st}
	\end{equation}
\end{proposition}

\begin{figure*}[!t]
	\centering
	\subfigure[Exponential: $\mu=1$]{
		\label{fig:info-tradition-exp-AoI}
		\includegraphics[width=0.322\textwidth]{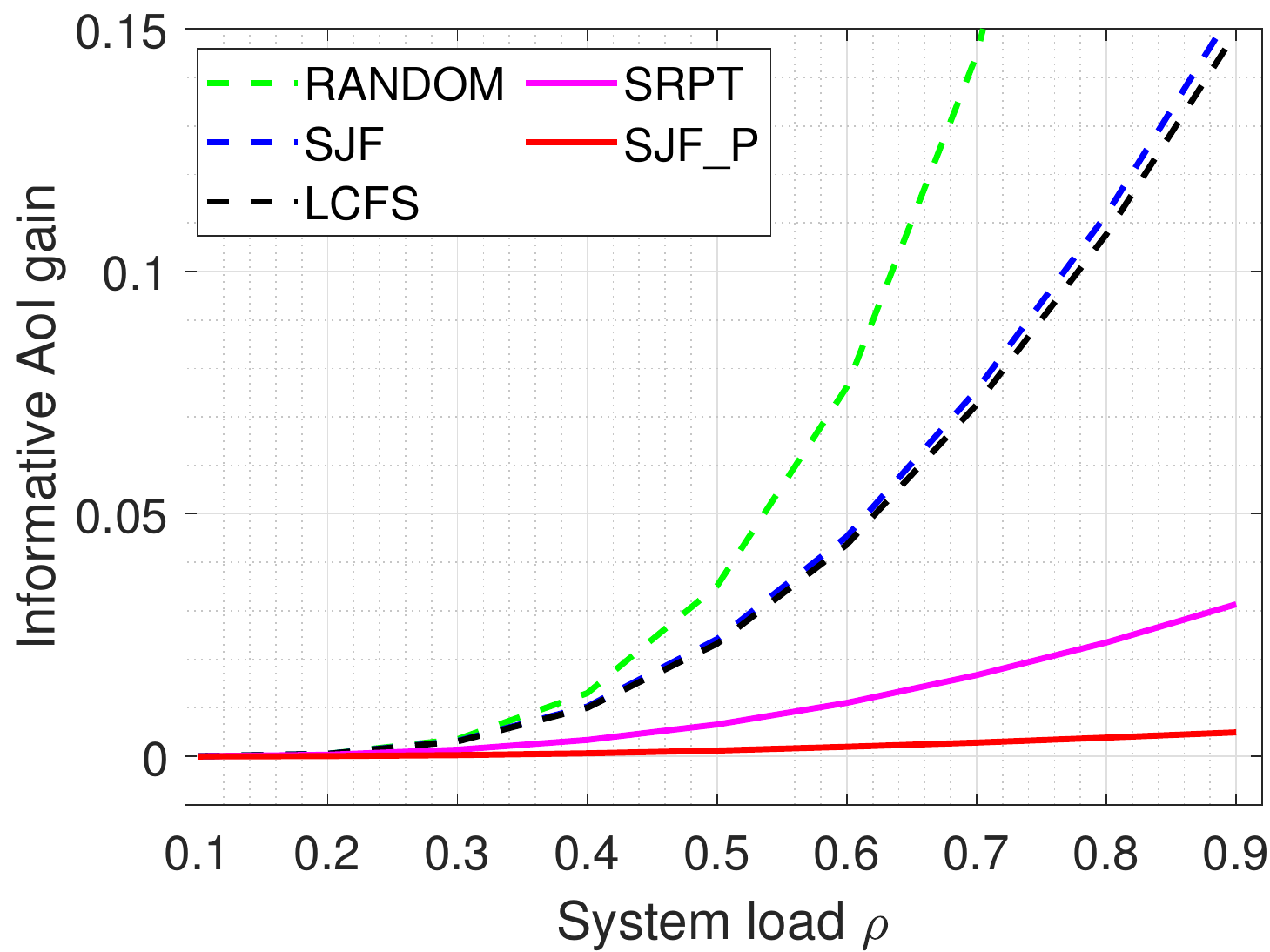}}
	\subfigure[Weibull: $\mu=1$ and ${C^{\rm{2}}}{\rm{ = 10}}$]{
		\label{fig:info-tradition-wei-AoI}
		\includegraphics[width=0.322\textwidth]{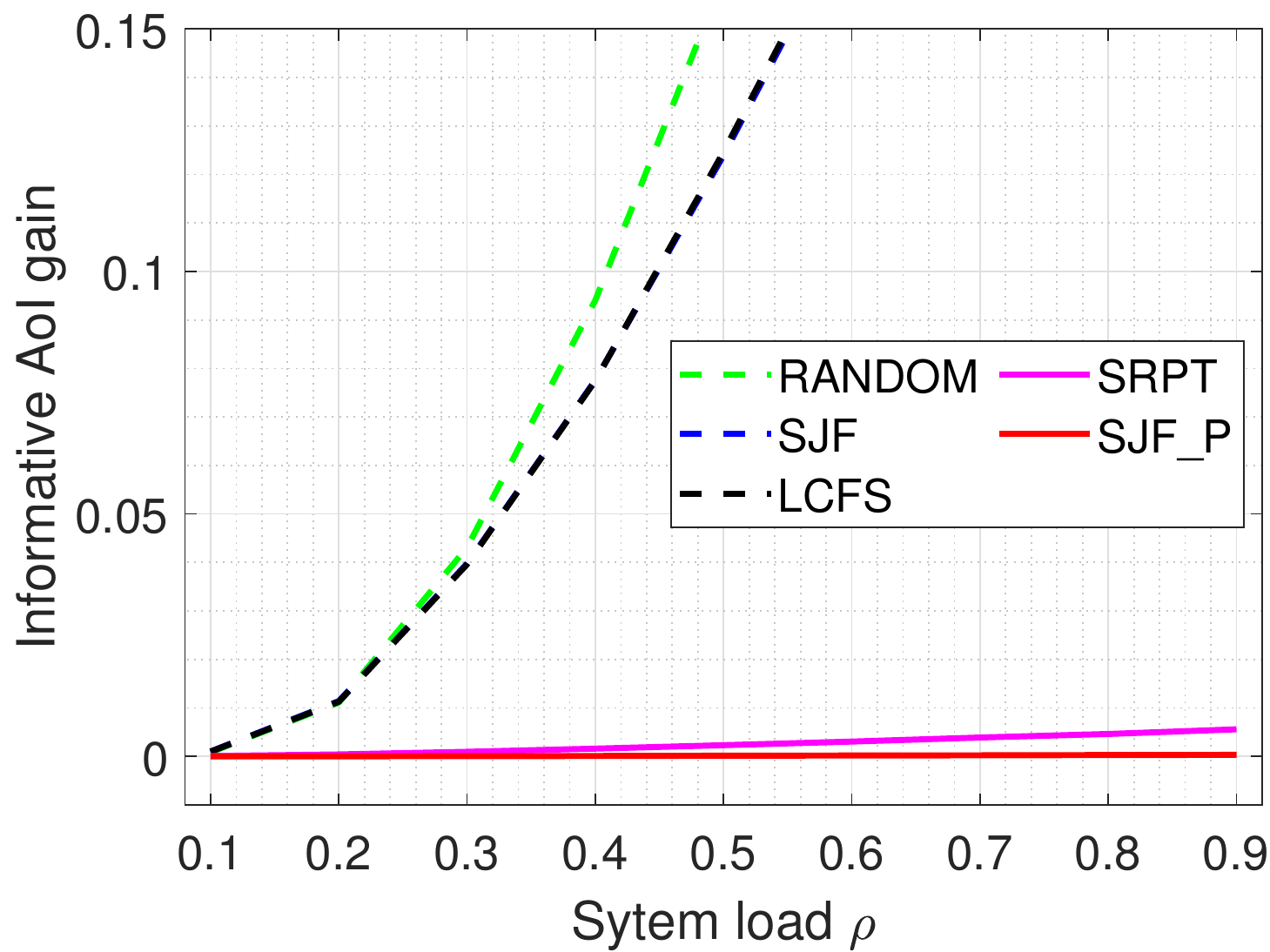}}
	\subfigure[Weibull: $\mu=1$ and $\rho {\rm{ = 0}}{\rm{.7}}$]{
		\label{fig:info-tradition-wei-variance-AoI} 
		\includegraphics[width=0.322\textwidth]{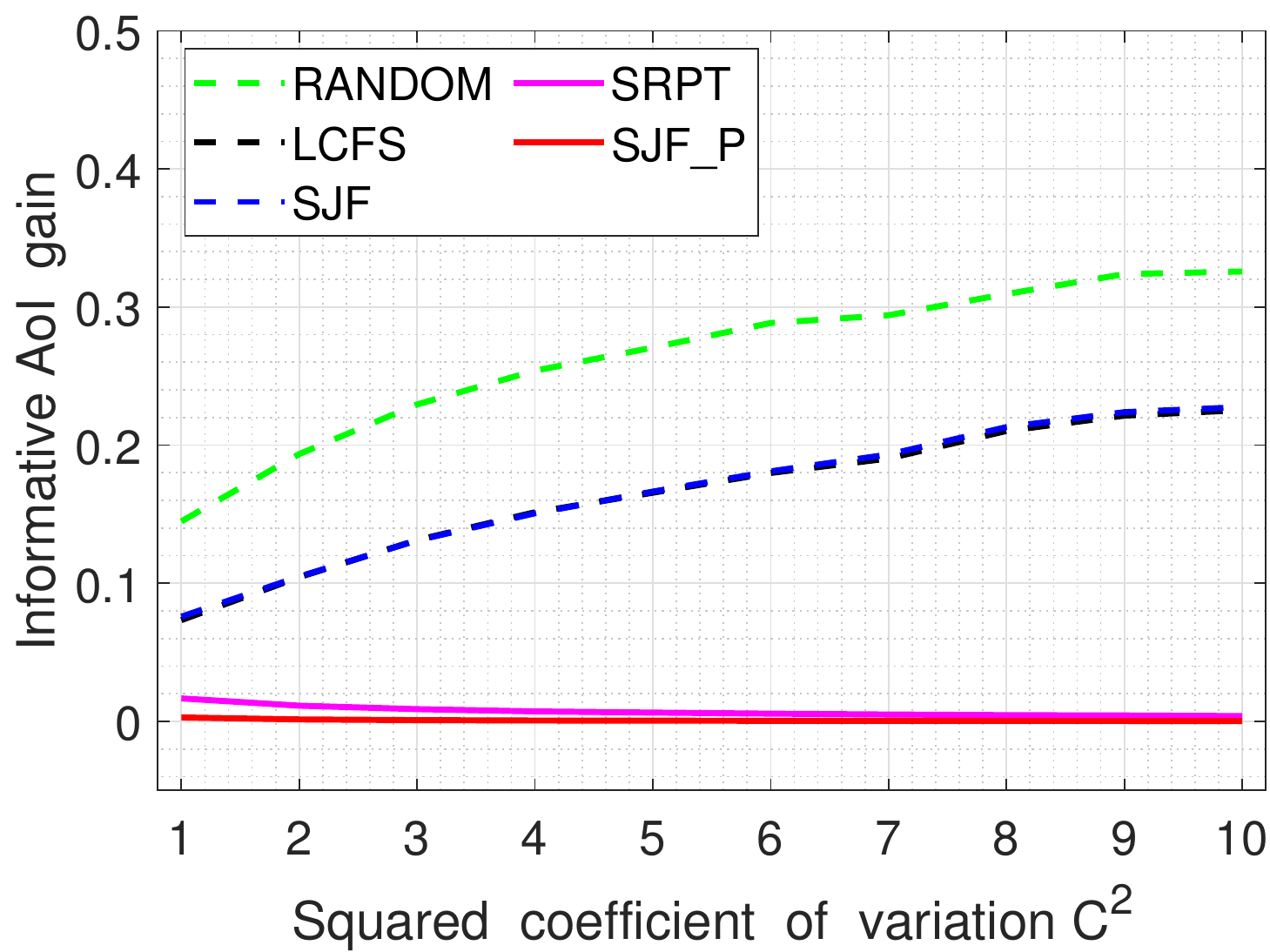}}
	\caption{Comparisons of the average AoI performance: informative policies vs. non-informative policies}	
	\label{fig:aoi-ratio-com-info}
\end{figure*}

\begin{figure*}[!t]
	\centering
	\subfigure[Exponential: $\mu=1$]{
		\label{fig:info-tradition-exp-PAoI}
		\includegraphics[width=0.322\textwidth]{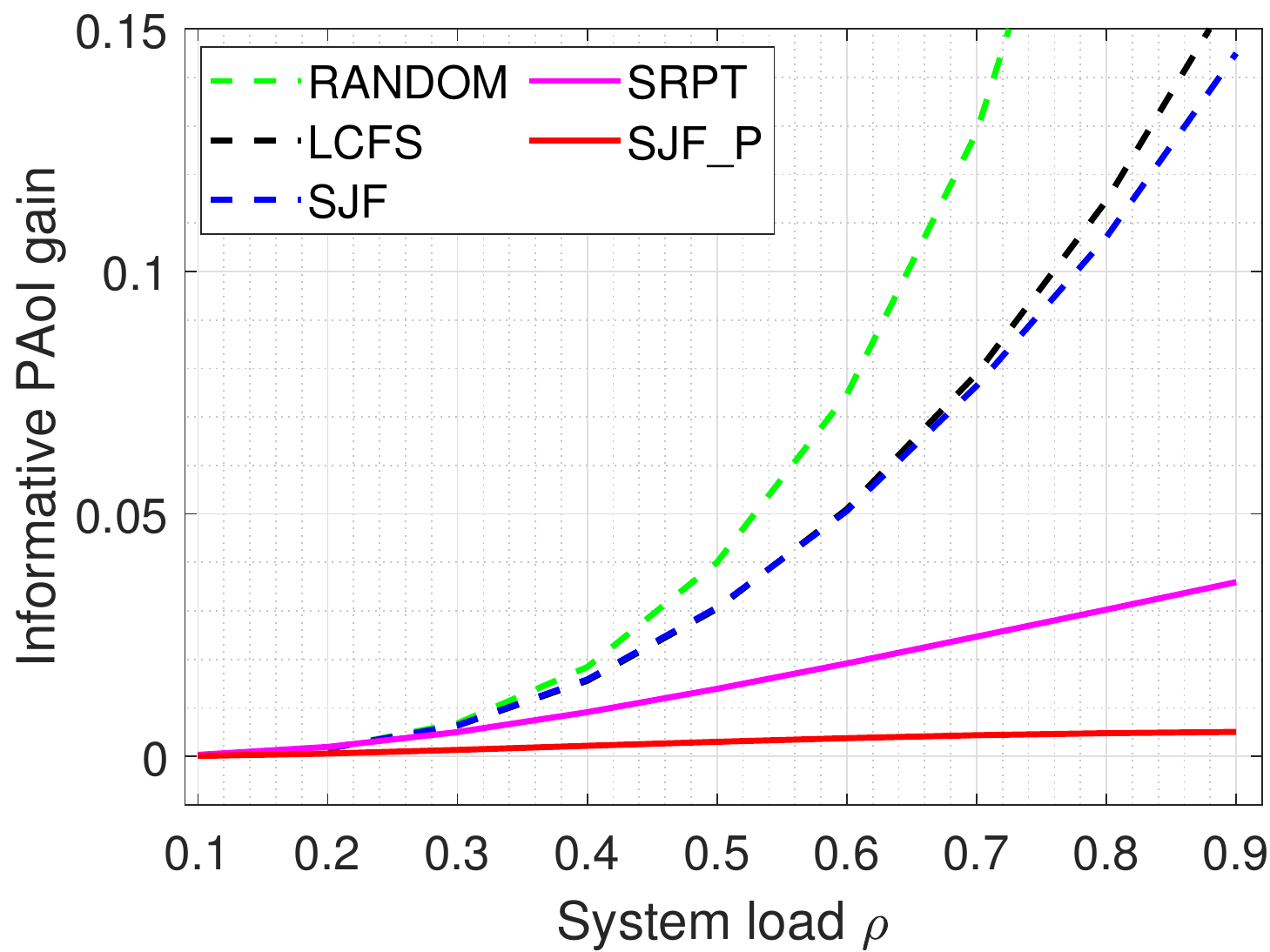}}
	\subfigure[Weibull: $\mu=1$ and ${C^{\rm{2}}}{\rm{ = 10}}$]{
		\label{fig:info-tradition-wei-PAoI}
		\includegraphics[width=0.322\textwidth]{exp-wei-aoi-informative-eps-converted-to.pdf}}
	\subfigure[Weibull: $\mu=1$ and $\rho {\rm{ = 0}}{\rm{.7}}$]{
		\label{fig:info-tradition-wei-variance-PAoI} 
		\includegraphics[width=0.322\textwidth]{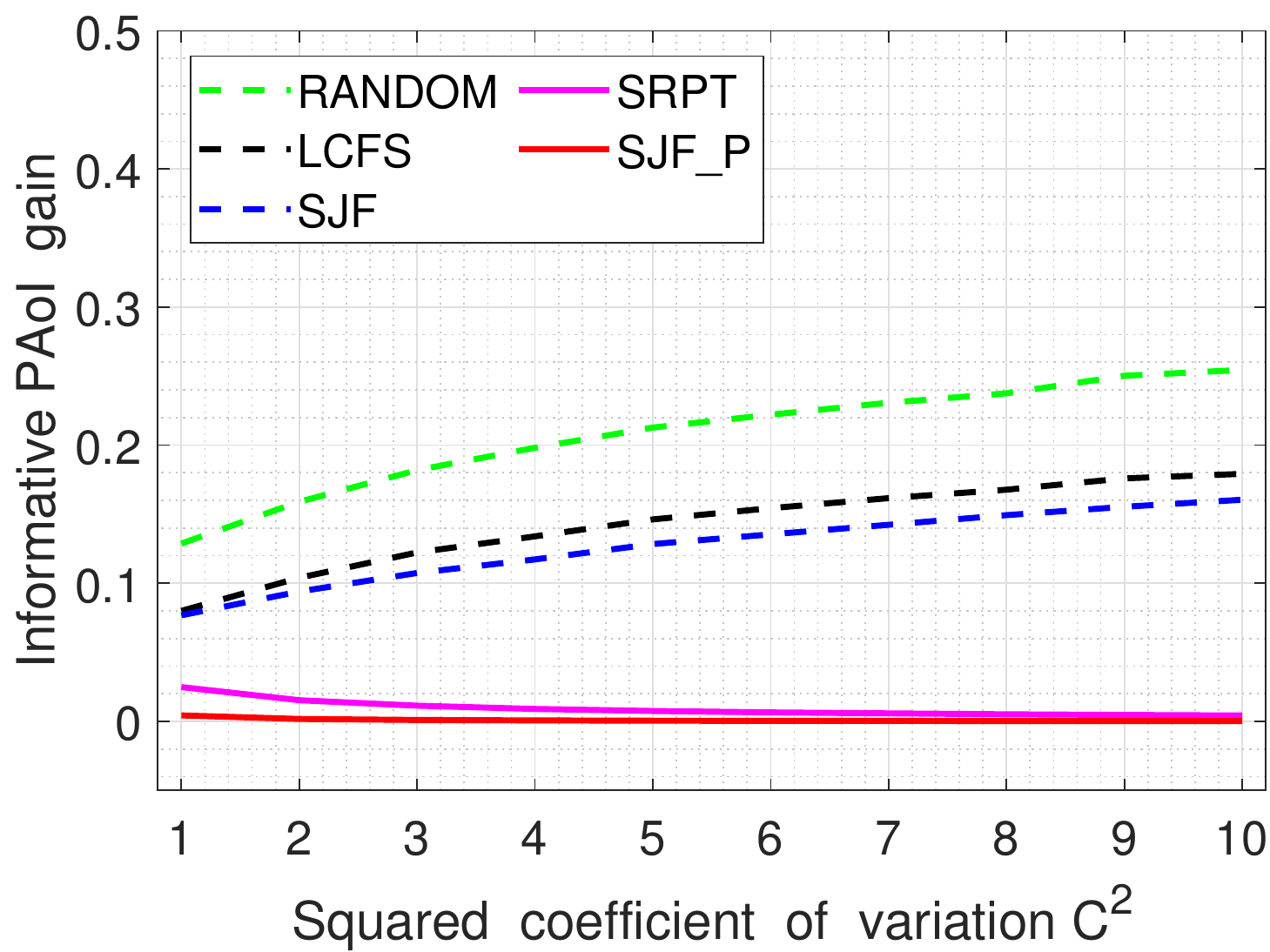}}
	\caption{Comparisons of the average PAoI performance: informative policies vs. non-informative policies}	
	\label{fig:paoi-ratio-com-info}
\end{figure*}

\begin{proof}
Recall that we use ${t_i}$ and ${t'_i}$ to denote the arrival time and the delivery time of the $i$-th update, respectively. 
In addition, we use $s_i$ to denote the service start time of the $i$-th update.

We define the system state at time $t$ under policy $\pi$ as ${S_\pi }(t) \triangleq  {U_\pi }(t)$, where  ${{U_\pi }\left( t \right)}$ is the largest arrival time of the updates that have been served under policy $\pi$ by time $t$. Let $\left\{{{S_\pi}(t),t \in \left[ {0,\infty } \right)} \right\}$ be the state process under policy $\pi$. By the definition of AoI, Eq.~(\ref{eq:aoi-st}) holds if the following holds:
\begin{equation}
        [\left\{ {{S_{{\rm{LCF}}{{\rm{S}}_-}{\rm{I}}}}(t),t \in [0,\infty )} \right\}|{\mathcal{I}}]{ \ge _{{\rm{st}}}}[\left\{ {{S_{{\rm{LCFS}}}}(t),t \in [0,\infty )} \right\}|{\mathcal{I}}].
        \label{eq:bridge}
\end{equation}

Next, we prove Eq.~(\ref{eq:bridge}) by contradiction through a coupling argument. Suppose that stochastic processes ${{\hat S}_{{\rm{LCFS\_I}}}}\left( t \right)$ and ${{\hat S}_{{\rm{LCFS}}}}\left( t \right)$ have the same stochastic laws as ${S_{{\rm{LCFS}}\_{\rm{I}}}}\left( t \right)$ and ${S_{{\rm{LCFS}}}}\left( t \right)$, respectively. We couple ${{\hat S}_{{\rm{LCFS\_I}}}}\left( t \right)$ and ${{\hat S}_{{\rm{LCFS}}}}\left( t \right)$ in the following manner: If an update $i$ is delivered at ${t'_i}$ in ${{\hat S}_{{\rm{LCFS}}}}(t)$, then the update $j$ being served at $t'_i$ (if any) in ${{\hat S}_{{\rm{LCFS\_I}}}}(t)$ is also delivered at the same time. This coupling is reasonable because: (i) the updates served in ${{\hat S}_{{\rm{LCFS\_I}}}}(t)$ are not chosen based on update size; (ii) the service time of an update in both ${{\hat S}_{{\rm{LCFS\_I}}}}\left( t \right)$ and ${{\hat S}_{{\rm{LCFS}}}}\left( t \right)$ is exponentially distributed and has the memoryless property.
By Theorem~6.B.30 in \cite{shaked2007stochastic}, Eq.~(\ref{eq:bridge}) holds if the following holds:
\begin{equation}
	{\mathbb P }({{\hat S}_{{\rm{LCFS}}\_{\rm{I}}}}\left( t \right) \ge {{\hat S}_{{\rm{LCFS}}}}\left( t \right),t \in \left[ {0,\infty } \right)|{\mathcal{I}}) = 1.
	\label{eq:prob-1}
\end{equation}

\begin{figure}
    \centering
    \subfigure[Case 1): The server in ${{\hat S}_{{\rm{LCFS\_I}}}}(t)$ is being idle at $t'_n$]{
    \begin{minipage}[]{\linewidth}
    \centering
    \includegraphics[scale=0.7]{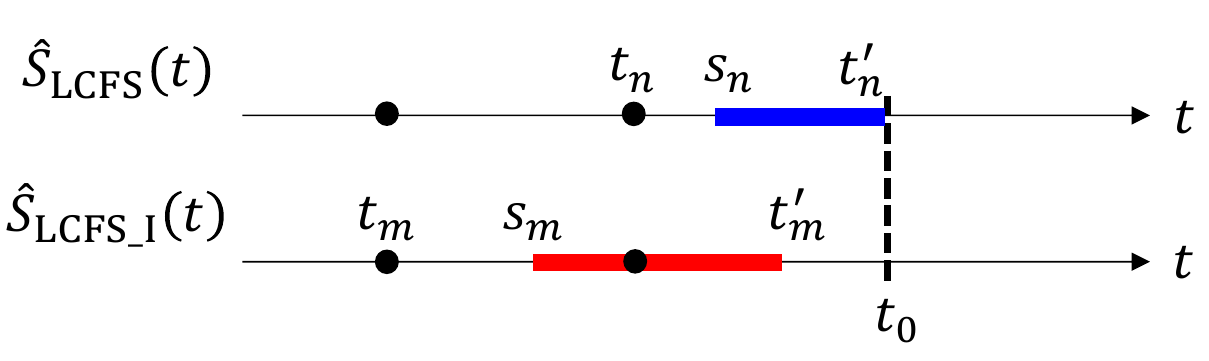}
    \label{fig:dominance-proof-1}
    \end{minipage}
    }
    
    \subfigure[Case 2a): The $m$-th update is delivered at $t'_n$ in ${{\hat S}_{{\rm{LCFS\_I}}}}(t)$, and the server in ${{\hat S}_{{\rm{LCFS}}}}(t)$ is idle at time $s_m$]{
    \begin{minipage}[]{\linewidth}
    \centering
    \includegraphics[scale=0.7]{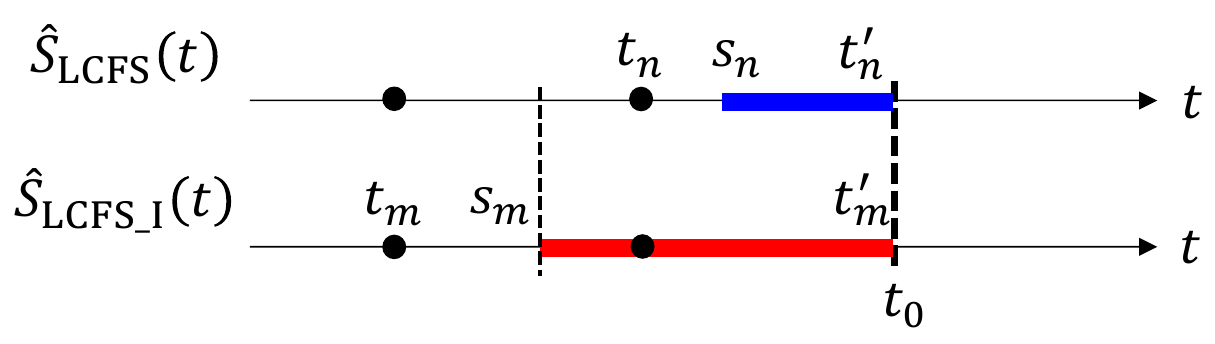}
    \label{fig:dominance-proof-2}
    \end{minipage}
    }
    
    \subfigure[Case 2b): The $m$-th update is delivered at $t'_n$ in ${{\hat S}_{{\rm{LCFS\_I}}}}(t)$, and the server in ${{\hat S}_{{\rm{LCFS}}}}(t)$ is busy at time $s_m$]{
    \begin{minipage}[]{\linewidth}
    \centering
    \includegraphics[scale=0.7]{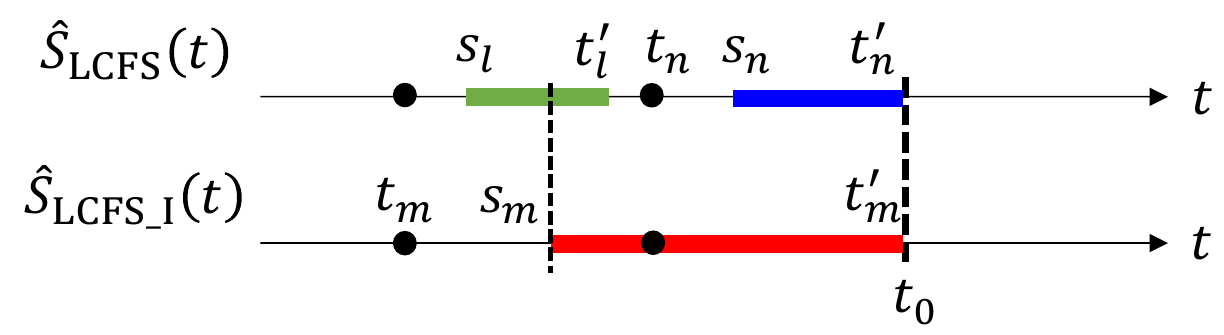}
    \label{fig:dominance-proof-3}
    \end{minipage}
    }
    \caption{Part of sample path of ${{\hat S}_{{\rm{LCFS}}\_{\rm{I}}}}(t)$ and ${{\hat S}_{{\rm{LCFS}}}}(t)$ in different cases}
    \label{fig:dominance-proof}
\end{figure}

In the following, we want to show that ${{{\hat S}_{{\rm{LCFS\_I}}}}\left( t \right) \ge {{\hat S}_{{\rm{LCFS}}}}\left( t \right)}$ holds conditionally on an arbitrary sample path $\mathcal{I}$, which trivially implies Eq.~(\ref{eq:prob-1}). We prove it by contradiction. 
For the sake of contradiction, suppose that ${\hat S_{{\rm{LCFS}}\_{\rm{I}}}}(t) < {\hat S_{{\rm{LCFS}}}}(t)$ does happen and that it happens for the first time at time $t_0$ (see Fig.~\ref{fig:dominance-proof} for illustration).
Let $m$ and $n$ be the index of the served updates with the largest arrival time by $t_0$ in ${{\hat S}_{{\rm{LCFS\_I}}}}(t)$ and ${{\hat S}_{{\rm{LCFS}}}}(t)$, respectively. Then, we have ${U_{{\rm{LCFS}}\_{\rm{I}}}}({t_0}) = {t_m}$ and ${U_{{\rm{LCFS}}}}({t_0}) = {t_n}$. Note that we also have ${t_m} < {t_n}$ due to ${\hat S_{{\rm{LCFS}}\_{\rm{I}}}}(t_0) < {\hat S_{{\rm{LCFS}}}}(t_0)$ (i.e., ${U_{{\rm{LCFS}}\_{\rm{I}}}}(t_0) < {U_{{\rm{LCFS}}}}(t_0)$). Since $t_0$ is the first time when ${\hat S_{{\rm{LCFS}}\_{\rm{I}}}}(t) < {\hat S_{{\rm{LCFS}}}}(t)$ happens, a crucial observation is that $t_0$ must be immediately after an update is delivered in ${{\hat S}_{{\rm{LCFS}}}}(t)$. Hence, we have ${t_0} = {({t'_n})^{+}}$, where ${({t'_n})^{+}}$ denotes the time immediately after $t'_n$.

Due to the coupling between ${{\hat S}_{{\rm{LCFS}}}}(t)$ and ${{\hat S}_{{\rm{LCFS\_I}}}}(t)$, there are two cases in ${{\hat S}_{{\rm{LCFS\_I}}}}(t)$: 1) the server is being idle at $t'_n$; 2) an update is delivered at $t'_n$ too. We discuss these two cases separately and show that there is a contradiction in both cases.

Case 1): The server in ${{\hat S}_{{\rm{LCFS\_I}}}}(t)$ is being idle at $t'_n$ (see Fig.~\ref{fig:dominance-proof-1}). Then, the most recently delivered update in ${{\hat S}_{{\rm{LCFS\_I}}}}(t)$ (i.e., the $m$-th update) must be delivered before $t'_n$. Hence, we have $t'_m < t'_n$ and that the server in ${{\hat S}_{{\rm{LCFS\_I}}}}(t)$ stays in the idle state during $(t'_m, t'_n]$. 
Then, the server in ${{\hat S}_{{\rm{LCFS\_I}}}}(t)$ could have started serving a newer update that arrives later than the $m$-th update immediately after $t'_m$. (Such a newer update must exist as the $n$-th update is a valid candidate due to ${t_m} < {t_n}$.)
This results in a contradiction with the server being idle during $(t'_m, t'_n]$.

Case 2): An update is delivered at $t'_n$ in ${{\hat S}_{{\rm{LCFS\_I}}}}(t)$. This delivered update is the $m$-th update. Note that we must have $s_m < t_n$. 
This is because if $s_m \ge t_n$, then the server in ${{\hat S}_{{\rm{LCFS\_I}}}}(t)$ would have chosen to serve the $n$-th update or a fresher update that arrives later than $t_n$ at time $s_m$ since this selected update is a newer update (due to $t_m < t_n$).
There are two subcases for the server in ${{\hat S}_{{\rm{LCFS}}}}(t)$ at time $s_m$: 2a) idle; 2b) busy. Again, we discuss these two subcases separately and show that there is a contradiction in both cases.

Case 2a): The server in ${{\hat S}_{{\rm{LCFS}}}}(t)$ is idle at time $s_m$ (see Fig.~\ref{fig:dominance-proof-2}). In this case, the $m$-th update must have already been delivered by time $s_m$ in ${{\hat S}_{{\rm{LCFS}}}}(t)$. Otherwise, the server in ${{\hat S}_{{\rm{LCFS}}}}(t)$ would have started serving the $m$-th update (or a newer update) at or before $s_m$. This implies that ${\hat S_{{\rm{LCFS}}\_{\rm{I}}}}(t) < {\hat S_{{\rm{LCFS}}}}(t)$ happens before $s_m$, which results in a contradiction with that $t_0$ is the first time at which ${\hat S_{{\rm{LCFS}}\_{\rm{I}}}}(t) < {\hat S_{{\rm{LCFS}}}}(t)$ happens.

Case 2b): The server in ${{\hat S}_{{\rm{LCFS}}}}(t)$ is busy at time $s_m$ (see Fig.~\ref{fig:dominance-proof-3}). Assume that the $l$-th update is being served at $s_m$ in ${{\hat S}_{{\rm{LCFS}}}}(t)$. In this case, the $l$-th update must be delivered by time $s_n$ in ${{\hat S}_{{\rm{LCFS}}}}(t)$. This is because the $n$-th update starts service at $s_n$ in ${{\hat S}_{{\rm{LCFS}}}}(t)$. Then, the $m$-th update must also be delivered by time $s_n$ in ${{\hat S}_{{\rm{LCFS\_I}}}}(t)$, due to the coupling between ${{\hat S}_{{\rm{LCFS}}}}(t)$ and ${{\hat S}_{{\rm{LCFS\_I}}}}(t)$. This results in a contradiction that the $m$-th update is delivered at $t'_n$.

Combining all the cases, we show that ${{{\hat S}_{{\rm{LCFS\_I}}}}\left( t \right) \ge {{\hat S}_{{\rm{LCFS}}}}\left( t \right)}$ holds conditionally on an arbitrary sample path $\mathcal{I}$. This trivially implies Eq.~(\ref{eq:prob-1}), which further implies Eq.~(\ref{eq:bridge}) by Theorem~6.B.30 in \cite{shaked2007stochastic}. This completes the proof. 
\end{proof}

\subsection{Preemptive, Informative, AoI-based Policies}
\label{sec:piaoi}
So far, we have demonstrated the advantages of preemptive policies, AoI-based policies, and informative policies. In this subsection, we want to integrate all of these three ideas and propose preemptive, informative, AoI-based policies.

We first consider preemptive, informative version of three AoI-based policies: ADE\_PI, ADS\_PI, and ADM\_PI. Interestingly, we can show equivalence between ADE\_PI and SRPT\_I (i.e., the informative version of SRPT) and between ADE\_I and SJF\_I (i.e., the informative version of ADE and SJF, respectively) in the sample-path sense. These results are stated in Propositions~\ref{pro:padf=srpti} and \ref{pro:padf=sjfi}.

\begin{proposition}
	\label{pro:padf=srpti}
    ADE\_PI and SRPT\_I are equivalent in every sample path. 
\end{proposition}

\begin{proof}
We use strong induction to prove that under the same sample path, ADE\_PI  and SRPT\_I  always choose the same update to serve at the same time. In the following, we only consider informative updates since non-informative updates are discarded under both ADE\_PI and SRPT\_I.

Suppose that when ADE\_PI needs to choose the $n$-th update to serve at time ${t_{{\rm{ADE\_PI}}}}\left( n \right)$, it chooses the update with index ${d_{{\rm{ADE\_PI}}}}\left( n \right)$. Similarly, SRPT\_I chooses the update with index ${d_{{\rm{SRPT\_I}}}}\left( n \right)$ as its $n$-th update to serve at ${t_{{\rm{SRPT\_I}}}}\left( n \right)$.

Claim: ADE\_PI  and SRPT\_I always serve the same update at the same time, i.e., $(d_{\mathrm{ADE}_{-} \mathrm{PI}}(n),$ $t_{\mathrm{ADE}_{-} \mathrm{PI}}(n))=(d_{\mathrm{SRPT}_{-} \mathrm{I}}(n), t_{\mathrm{SRPT}_{-} \mathrm{I}}(n))$ for all $n$. 

Base case: When $n=1$, both ADE\_PI and SRPT\_I serve the first update when it arrives. Hence, we have $\left( {{d_{{\rm{ADE\_PI}}}}\left( 1 \right),{t_{{\rm{ADE\_PI}}}}\left( 1 \right)} \right) = \left( {{d_{{\rm{SRPT}}\_{\rm{I}}}}\left( 1 \right),{t_{{\rm{SRPT}}\_{\rm{I}}}}\left( 1 \right)} \right)$.

Induction step: Suppose that for $n=k$ $(k \ge 1)$, we have $\left( {{d_{{\rm{ADE\_PI}}}}\left( m \right),{t_{{\rm{ADE\_PI}}}}\left( m \right)} \right) = \left( {{d_{{\rm{SRPT}}\_{\rm{I}}}}\left( m \right),{t_{{\rm{SRPT}}\_{\rm{I}}}}\left( m \right)} \right)$ for the $m$-th update for all $1 \le m \le k$. We want to show that $\left( {{d_{{\rm{ADE\_PI}}}}\left( n \right),{t_{{\rm{ADE\_PI}}}}\left( n \right)} \right) = \left( {{d_{{\rm{SRPT}}\_{\rm{I}}}}\left( n \right),{t_{{\rm{SRPT}}\_{\rm{I}}}}\left( n \right)} \right)$ still holds for $n=k+1$. 
Note that there are two cases for the $(k+1)$-st  update: 1) the $(k+1)$-st update preempts the $k$-th update; 2) the $(k+1)$-st update does not preempt the $k$-th update, i.e., the $(k+1)$-st update starts service from the idle state or immediately after the $k$-th update is delivered. We discuss these two cases separately and show that $\left( {{d_{{\rm{ADE\_PI}}}}\left( k+1 \right),{t_{{\rm{ADE\_PI}}}}\left( k+1 \right)} \right) = \left( {{d_{{\rm{SRPT}}\_{\rm{I}}}}\left( k+1 \right),{t_{{\rm{SRPT}}\_{\rm{I}}}}\left( k+1 \right)} \right)$ holds in both cases.

Case 1): The $(k+1)$-st update preempts the $k$-th update. During the service of the $k$-th update,  the $(k+1)$-st update arrives. Under ADE\_PI, in order to make AoI drop as early as possible, the server compares the remaining service time of the $k$-th update with the original service time of the $(k+1)$-st update and chooses to serve the update with a smaller remaining service time. This is exactly the same as what SRPT\_I does. Therefore, we have $\left( {{d_{{\rm{ADE\_PI}}}}\left( k+1 \right),{t_{{\rm{ADE\_PI}}}}\left( k+1 \right)} \right) = \left( {{d_{{\rm{SRPT}}\_{\rm{I}}}}\left( k+1 \right),{t_{{\rm{SRPT}}\_{\rm{I}}}}\left( k+1 \right)} \right)$. 

Case 2): The $(k+1)$-st update does not preempt the $k$-th update.
On the one hand, if the  $(k+1)$-st update starts service from the idle state, then by the induction hypothesis, both ADE\_PI and SRPT\_I finish serving the $k$-th update at the same time and then go through a period of being idle. Therefore, ADE\_PI and SRPT\_I will also serve the same $(k+1)$-st update at the same time, i.e., $\left( {{d_{{\rm{ADE\_PI}}}}\left( k+1 \right),{t_{{\rm{ADE\_PI}}}}\left( k+1 \right)} \right) = \left( {{d_{{\rm{SRPT}}\_{\rm{I}}}}\left( k+1 \right),{t_{{\rm{SRPT}}\_{\rm{I}}}}\left( k+1 \right)} \right)$. 
On the other hand, if the $(k+1)$-st update starts service immediately after the service of $k$-th update, then by the induction hypothesis, ADE\_PI and SRPT\_I will start service at the same time, i.e., ${t_{{\rm{ADE}}\_{\rm{PI}}}}\left( {k + 1} \right){\rm{ = }}{t_{{\rm{SRPT}}\_{\rm{I}}}}\left( {k+1} \right)$.  SRPT\_I will select the $(k+1)$-st update with the shortest remaining size. However, this selected $(k+1)$-st update must have  not been served before. Otherwise, this update is no longer informative it was preempted by other update. Thus, SRPT\_I ends up choosing an update with the shortest original size, which  will also  be selected by ADE\_PI. This implies ${d_{{\rm{ADE}}\_{\rm{PI}}}}\left( {k + 1} \right) = {d_{{\rm{SRPT}}\_{\rm{I}}}}\left( {k + 1} \right)$.  Therefore, we have $\left( {{d_{{\rm{ADE\_PI}}}}\left( k+1 \right),{t_{{\rm{ADE\_PI}}}}\left( k+1 \right)} \right) = \left( {{d_{{\rm{SRPT}}\_{\rm{I}}}}\left( k+1 \right),{t_{{\rm{SRPT}}\_{\rm{I}}}}\left( k+1 \right)} \right)$. 
\end{proof}

\begin{proposition}
\label{pro:padf=sjfi}
    ADE\_I and SJF\_I are equivalent in every sample path. 
\end{proposition}
\begin{proof}
Similar to the proof of Proposition \ref{pro:padf=srpti}, we use strong induction to show that under the same sample path, ADE\_I and SJF\_I always choose the same update to serve at the same time. Here, we also only consider the informative updates.

Suppose that when ADE\_I needs to choose the $n$-th update to serve at time ${t_{{\rm{ADE\_I}}}}\left( n \right)$, it chooses the update with index ${d_{{\rm{ADE\_I}}}}\left( n \right)$. Similarly, SJF\_I chooses the update with index ${d_{{\rm{SJF\_I}}}}\left( n \right)$ as its $n$-th update to serve at ${t_{{\rm{SJF\_I}}}}\left( n \right)$.

Claim: ADE\_I  and SJF\_I  always serve the same update at the same time, i.e., $(d_{\mathrm{ADE}_{-} \mathrm{I}}(n),$ $t_{\mathrm{ADE}_{-} \mathrm{I}}(n))=(d_{\mathrm{SJF}_{-} \mathrm{I}}(n), t_{\mathrm{SJF}_{-} \mathrm{I}}(n))$ for all $n$.

Base case: When $n=1$, both ADE\_I and SJF\_I  serve the first update when it arrives. Hence, we have $\left( {{d_{{\rm{ADE\_I}}}}\left( 1 \right),{t_{{\rm{ADE\_I}}}}\left( 1 \right)} \right) = \left( {{d_{{\rm{SJF}}\_{\rm{I}}}}\left( 1 \right),{t_{{\rm{SJF}}\_{\rm{I}}}}\left( 1 \right)} \right)$.

Induction step: Suppose that for $n=k$ $(k \ge 1)$, we have $\left( {{d_{{\rm{ADE\_I}}}}\left( m \right),{t_{{\rm{ADE\_I}}}}\left( m \right)} \right) = \left( {{d_{{\rm{SJF}}\_{\rm{I}}}}\left( m \right),{t_{{\rm{SJF}}\_{\rm{I}}}}\left( m \right)} \right)$ for the $m$-th update for $1 \le m \le k$.
We want to show that $(d_{\mathrm{ADE}_{-} \mathrm{I}}(n),$ $t_{\mathrm{ADE}_{-} \mathrm{I}}(n))=(d_{\mathrm{SJF}_{-} \mathrm{I}}(n), t_{\mathrm{SJF}_{-} \mathrm{I}}(n))$ still holds for $n=k+1$. 
Note that there are two cases for the $(k+1)$-st update: 1) the $(k+1)$-st update starts service from the idle state; 2) the $(k+1)$-st update starts service immediately after the $k$-th update is delivered. We discuss these two cases separately and show that $\left( {{d_{{\rm{ADE\_I}}}}\left( k+1 \right),{t_{{\rm{ADE\_I}}}}\left( k+1 \right)} \right) = \left( {{d_{{\rm{SJF}}\_{\rm{I}}}}\left( k+1 \right),{t_{{\rm{SJF}}\_{\rm{I}}}}\left( k+1 \right)} \right)$ holds in both cases.

\begin{figure*}[!t]
	\centering
	\subfigure[Exponential: $\mu=1$]{
		\label{fig:exp-area-ratio-aoi}
		\includegraphics[width=0.322\textwidth]{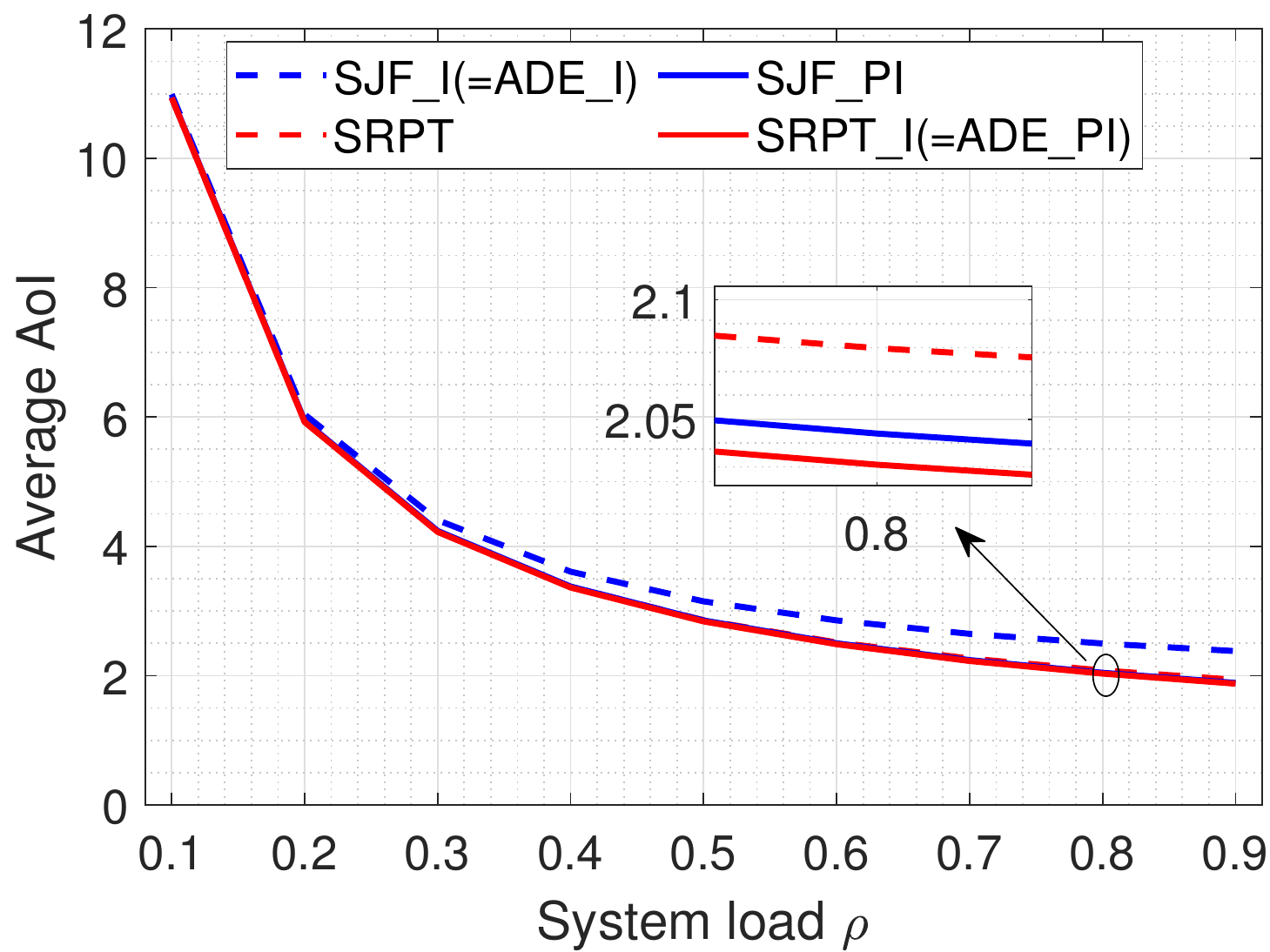}}
	\subfigure[Weibull: $\mu=1$ and ${C^{\rm{2}}}{\rm{ = 10}}$]{
		\label{fig:wei-area-ratio-aoi}
		\includegraphics[width=0.322\textwidth]{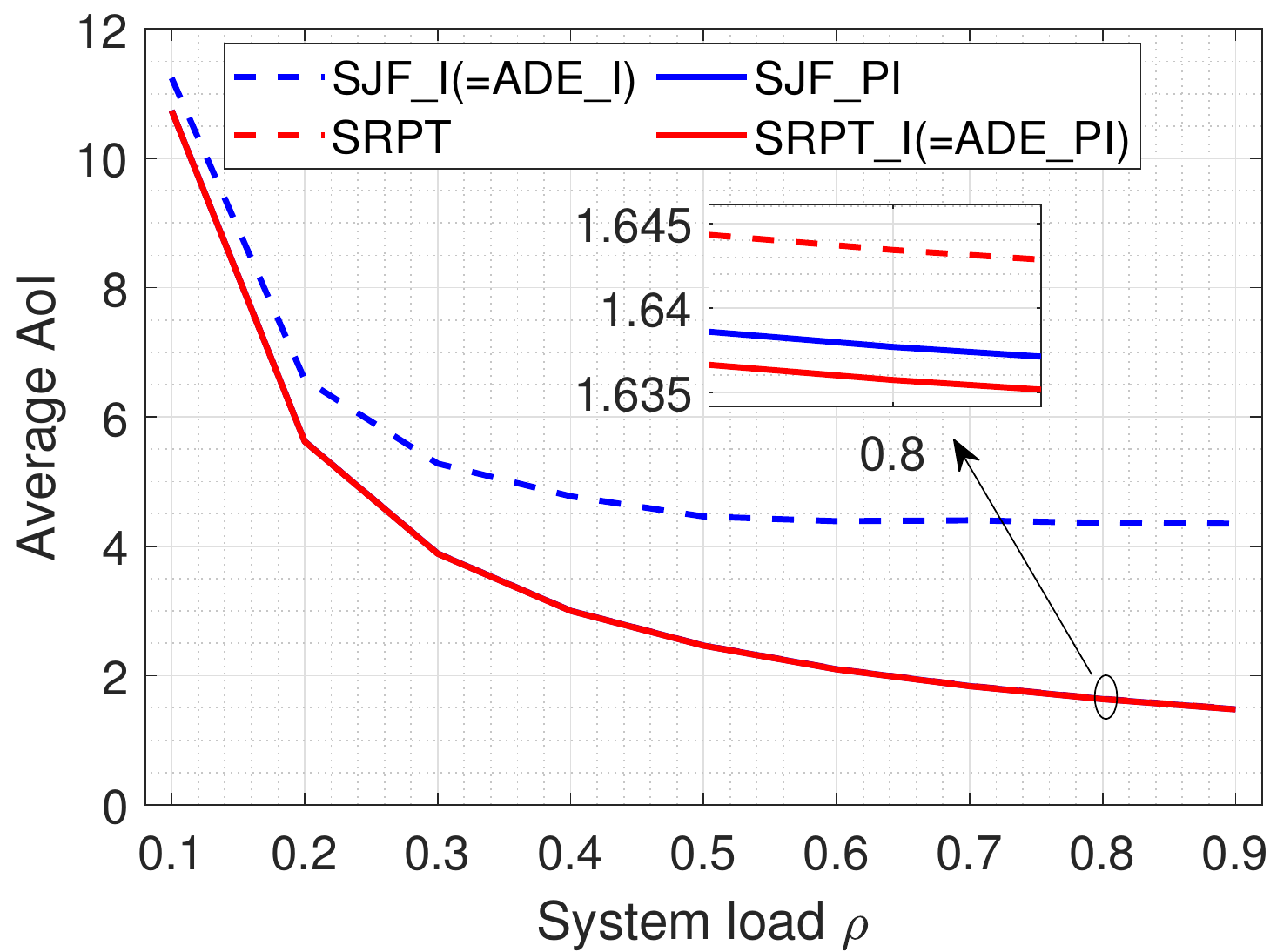}}
	\subfigure[Weibull: $\mu=1$ and $\rho {\rm{ = 0}}{\rm{.7}}$]{
		\label{fig:wei-area-info-aoi} 
		\includegraphics[width=0.322\textwidth]{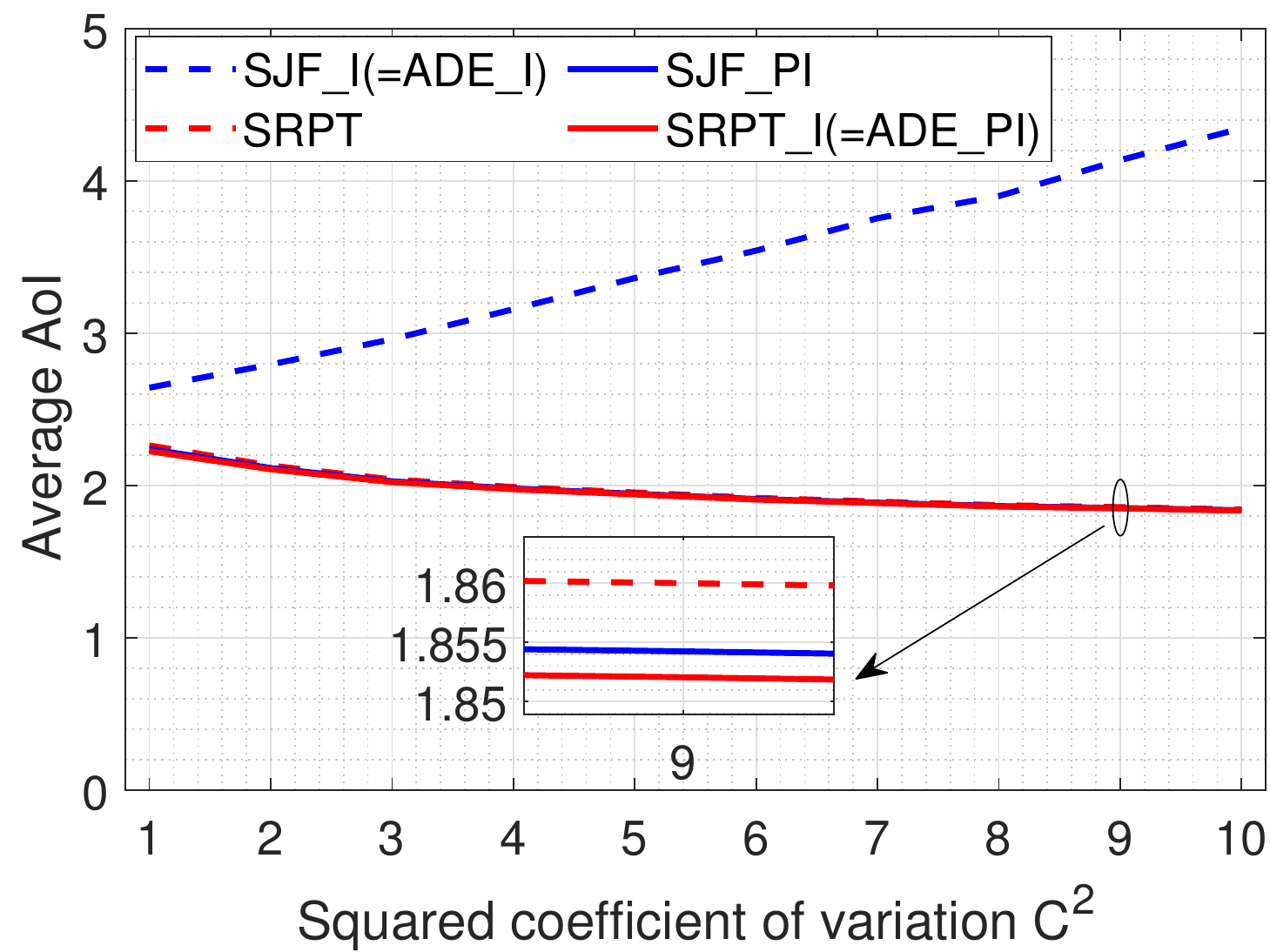}}
	\caption{Comparisons of the average AoI performance: preemptive, informative, AoI-based policies vs. others}
	\label{fig:aoi-ratio-pre-info}
\end{figure*}

\begin{figure*}[!t]
	\centering
	\subfigure[Exponential: $\mu=1$]{
		\label{fig:exp-area-ratio-paoi}
		\includegraphics[width=0.322\textwidth]{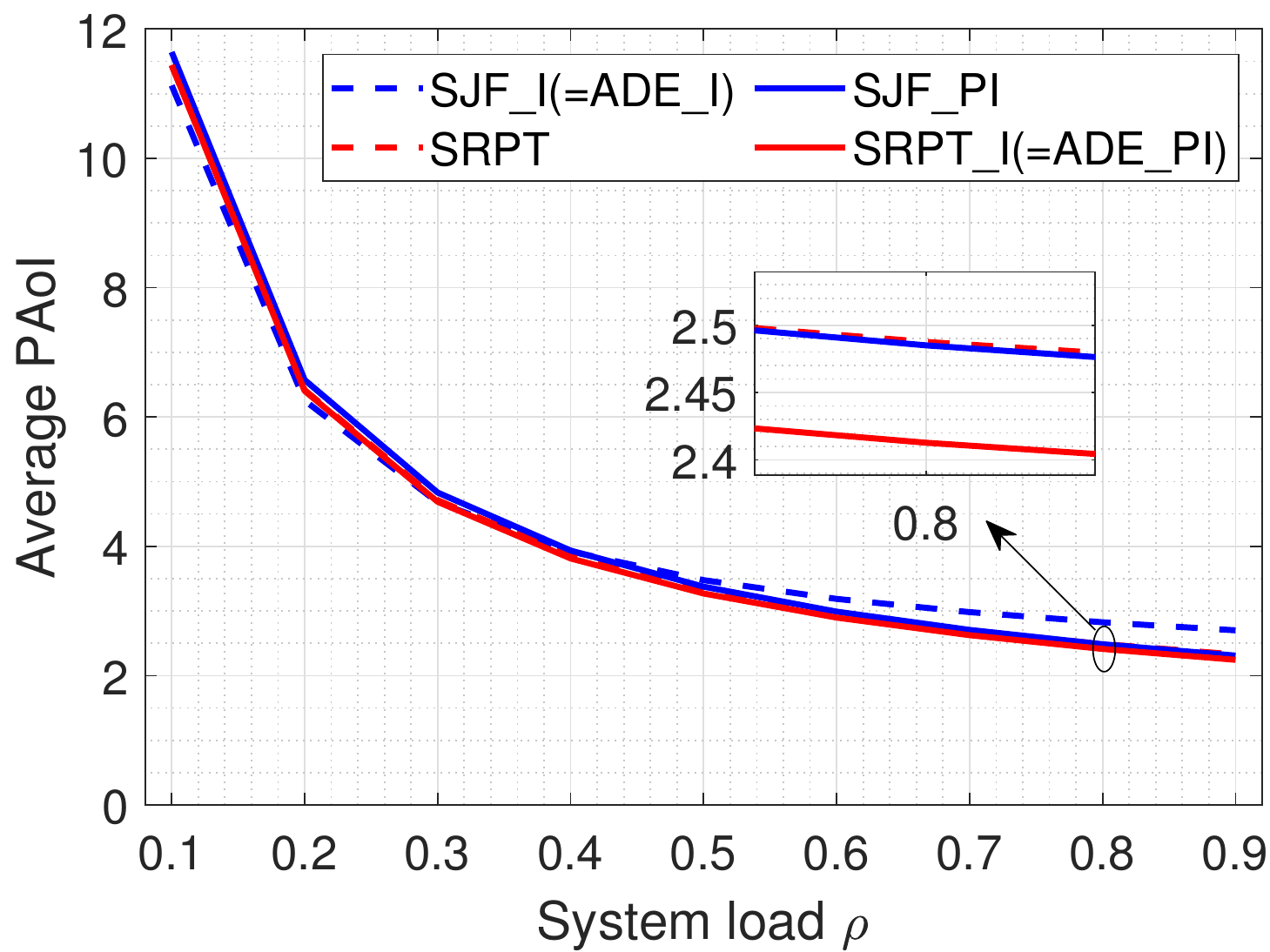}}
	\subfigure[Weibull: $\mu=1$ and ${C^{\rm{2}}}{\rm{ = 10}}$]{
		\label{fig:wei-area-ratio-paoi}
		\includegraphics[width=0.322\textwidth]{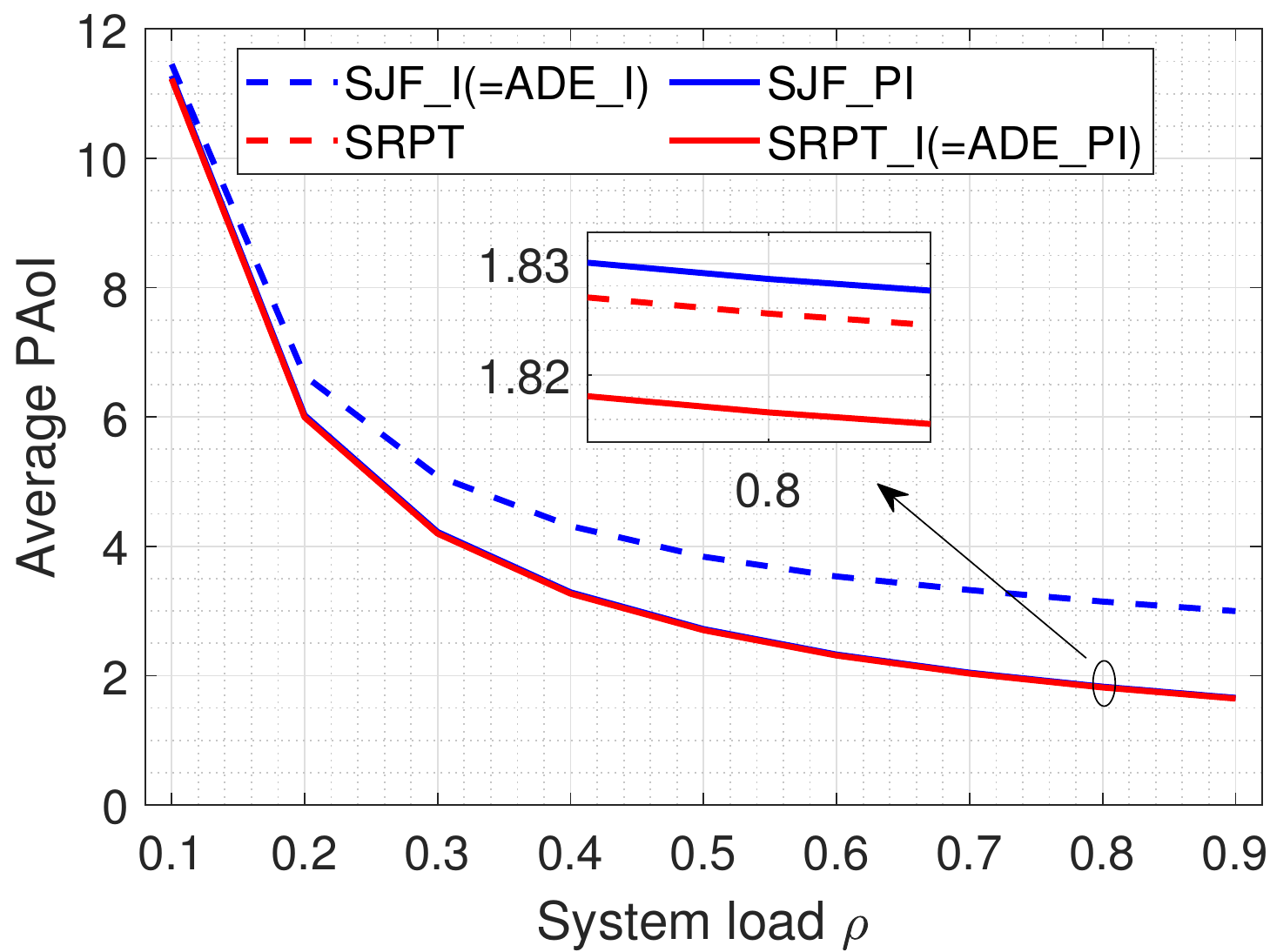}}
	\subfigure[Weibull: $\mu=1$ and $\rho {\rm{ = 0}}{\rm{.7}}$]{
		\label{fig:wei-area-info-paoi} 
		\includegraphics[width=0.322\textwidth]{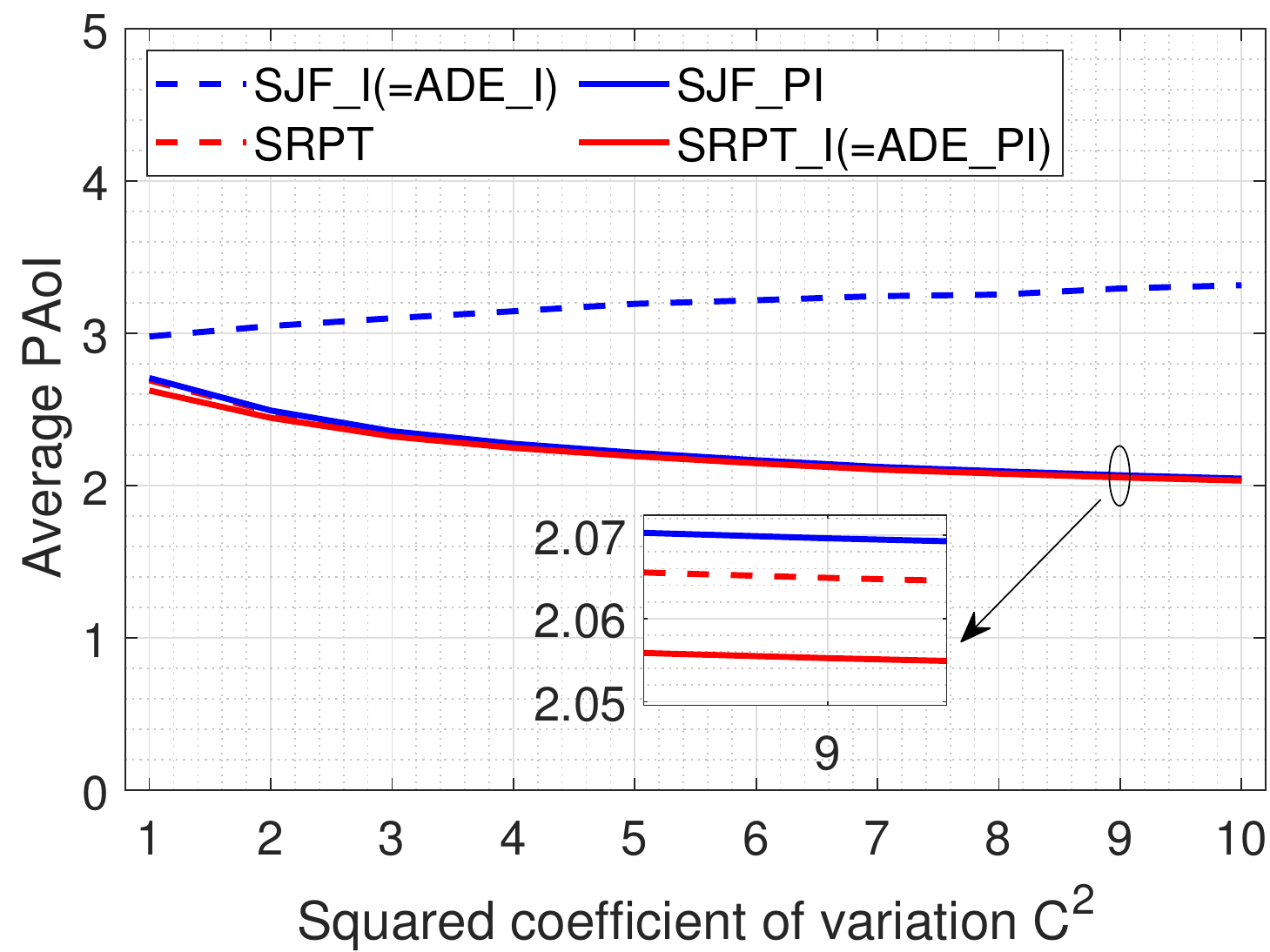}}
	\caption{Comparisons of the average PAoI performance: preemptive, informative, AoI-based policies vs. others}
	\label{fig:paoi-ratio-pre-info}
\end{figure*}

Case 1): The $(k+1)$-st update starts service from the idle state. By the induction hypothesis, both ADE\_I and SJF\_I finish serving the $k$-th update at the same time and then go through a period of being idle. Therefore, ADE\_I and SJF\_I will also serve the same $(k+1)$-st update at the same time, i.e., $\left( {{d_{{\rm{ADE\_I}}}}\left( k+1 \right),{t_{{\rm{ADE\_I}}}}\left( k+1 \right)} \right) = \left( {{d_{{\rm{SJF}}\_{\rm{I}}}}\left( k+1 \right),{t_{{\rm{SJF}}\_{\rm{I}}}}\left( k+1 \right)} \right)$. 

Case 2): The $(k+1)$-st update starts service immediately after the $k$-th update is delivered. By the induction hypothesis, ADE\_I and SJF\_I will start service at the same time, i.e., ${t_{{\rm{ADE}}\_{\rm{I}}}}\left( {k + 1} \right){\rm{ = }}{t_{{\rm{SJF}}\_{\rm{I}}}}\left( {k + 1} \right)$.  SJF\_I will choose the $(k+1)$-st update that has the smallest update size,  which  will also  be selected by ADE\_I since this update can make AoI drop earliest. This implies ${d_{{\rm{ADE}}\_{\rm{PI}}}}\left( {k + 1} \right) = {d_{{\rm{SJF}}\_{\rm{I}}}}\left( {k + 1} \right)$. Therefore, we have $\left( {{d_{{\rm{ADE\_I}}}}\left( k+1 \right),{t_{{\rm{ADE\_I}}}}\left( k+1 \right)} \right) = \left( {{d_{{\rm{SJF}}\_{\rm{I}}}}\left( k+1 \right),{t_{{\rm{SJF}}\_{\rm{I}}}}\left( k+1 \right)} \right)$. 
\end{proof}

Propositions~\ref{pro:padf=srpti} and \ref{pro:padf=sjfi} imply that although SRPT\_I and SJF\_I do not explicitly follow an AoI-based design, they are essentially AoI-based policies.
This provides an intuitive explanation for why size-based policies, such as variants of SRPT and SJF, have a good empirical AoI performance.

In Fig.~\ref{fig:aoi-ratio-pre-info}, we present the simulation results for the average AoI performance of the preemptive, informative, AoI-based policies (ADE\_PI) compared to several other policies. 
We observe that in various settings we consider, ADE\_PI  achieves the best AoI performance.
However, compared to the best delay-efficient policies (such as SRPT), the AoI improvement of the preemptive, informative, and AoI-based policies is rather marginal in the settings with exogenous arrivals. 

\section{\uppercase{Conclusion}}
\label{sec:conclusion}
In this paper, we systematically studied the impact of various aspects of scheduling policies on the AoI performance and provided several useful guidelines for the design of AoI-efficient scheduling policies.
Our study reveals that among various aspects of scheduling policies we investigated, prioritizing small updates, allowing service preemption, and prioritizing informative updates play the most important role in the design of AoI-efficient scheduling policies.
It turns out that common scheduling policies like SRPT and SJF\_P and their informative variants can achieve a very good AoI performance, although they do not explicitly make scheduling decisions based on the AoI. This can be partially explained by the equivalence between such size-based policies and some AoI-based policies.
Moreover, when the AoI requirement is not stringent or the update-size information is not available, some simple delay-efficient policies (such as LCFS\_P) are also good candidates for AoI-efficient policies.

Our findings also raise several interesting questions that are worth investigating as future work. One important direction is to pursue more theoretical results beyond the simulation results we provided in this paper. For example, it would be interesting to see whether one can rigorously prove that any informative policy always outperforms its non-informative counterpart, which is consistently observed in the simulation results.

\bibliographystyle{IEEEtran}
\bibliography{reference}

\appendix
\subsection{The $95\%$ confidence interval of Figs.~\ref{fig:tradition-AoI} and \ref{fig:tradition-PAoI}}
\label{appendix:confidence-interval}
The $95\% $ confidence intervals of Figs.~\ref{fig:tradition-AoI} and \ref{fig:tradition-PAoI} are provided in Table~\ref{table:tradition-exp-AoI}-\ref{table:tradition-wei-variance-PAoI}, in which we observe that the margin of error is only a very small portion of the average (about $1\%$). Note that the $95\% $ confidence intervals of the rest of figures are omitted, since the portion of margin of error is small in all figures.

\begin{table*}[b]
\centering
\begin{tabular}{|c|c|c|c|c|c|c|c|c|}
\hline
 & FCFS & RANDOM & LCFS & SJF & PS & LCFS\_P & SRPT & SJF\_P \\ \hline
0.1 & $10.99 \pm 0.03$ & $10.99 \pm 0.03$ & $10.99 \pm 0.03$ & $10.99 \pm 0.03$ & $10.99 \pm 0.03$ & $10.99 \pm 0.03$ & $10.94 \pm 0.03$ & $10.95 \pm 0.03$ \\ \hline
0.2 & $6.05 \pm 0.01$ & $6.04 \pm 0.01$ & $6.04 \pm 0.01$ & $6.04 \pm 0.01$ & $6.02 \pm 0.01$ & $5.99 \pm 0.01$ & $ 5.92\pm 0.01$ & $5.93 \pm 0.01$ \\ \hline
0.3 & $4.46 \pm 0.01$ & $4.44 \pm 0.01$ & $4.43 \pm 0.01$ & $4.42 \pm 0.01$ & $4.39 \pm 0.01$ & $4.34 \pm 0.01$ & $4.23 \pm 0.01$ & $4.24 \pm 0.01$ \\ \hline
0.4 & $3.77 \pm 0.01$ & $3.71 \pm 0.01$ & $3.66 \pm 0.01$ & $3.65 \pm 0.01$ & $3.62 \pm 0.01$ & $3.50 \pm 0.01$ & $3.38 \pm 0.01$ & $3.39 \pm 0.01$ \\ \hline
0.5 & $3.50 \pm 0.01$ & $3.35 \pm 0.01$ & $3.25 \pm 0.01$ & $3.23 \pm 0.01$ & $3.20 \pm 0.01$ & $3.00 \pm 0.01$ & $2.86 \pm 0.01$ & $2.86 \pm 0.01$ \\ \hline
0.6 & $3.56 \pm 0.01$ & $3.23 \pm 0.01$ & $3.02 \pm 0.01$ & $2.99 \pm 0.01$ & $3.00 \pm 0.01$ & $2.67 \pm 0.01$ & $2.51 \pm 0.01$ & $2.50 \pm 0.01$ \\ \hline
0.7 & $4.04 \pm 0.02$ & $3.31 \pm 0.01$ & $2.92 \pm 0.01$ & $2.86 \pm 0.01$ & $2.99 \pm 0.01$ & $2.43 \pm 0.01$ & $2.27 \pm 0.01$ & $2.25 \pm 0.01$ \\ \hline
0.8 & $5.44 \pm 0.05$ & $3.67 \pm 0.01$ & $2.89 \pm 0.01$ & $2.81 \pm 0.01$ & $3.22 \pm 0.01$ & $2.25 \pm 0.01$ & $2.08 \pm 0.01$ & $2.05 \pm 0.01$ \\ \hline
0.9 & $10.12 \pm 0.24$ & $4.65 \pm 0.04$ & $2.91 \pm 0.01$ & $2.80 \pm 0.01$ & $4.05 \pm 0.04$ & $2.11 \pm 0.01$ & $1.93 \pm 0.01$ & $1.90 \pm 0.01$ \\ \hline
\end{tabular}
\caption{The $95\% $ confidence interval of Fig.~\ref{fig:tradition-exp-AoI}}
\label{table:tradition-exp-AoI}
\vspace{0pt}
\end{table*}
\begin{table*}[b]
\centering
\begin{tabular}{|c|c|c|c|c|c|c|c|c|}
\hline
 & FCFS & RANDOM & LCFS & SJF & PS & LCFS\_P & SRPT & SJF\_P \\ \hline
0.1 & $11.29 \pm 0.02$ & $11.27 \pm 0.02$ & $11.26 \pm 0.02$ & $11.25 \pm 0.02$ & $10.77 \pm 0.02$ & $10.76 \pm 0.02$ & $10.74 \pm 0.02$ &  $10.74 \pm 0.02$ \\ \hline
0.2 & $6.88 \pm 0.03$ & $6.76 \pm 0.03$ & $6.70 \pm 0.02$ & $6.66 \pm 0.02$ & $5.69 \pm 0.01$ & $5.66 \pm 0.01$ & $5.63 \pm 0.01$ &  $5.63 \pm 0.01$   \\ \hline
0.3 & $6.14 \pm 0.07$ & $5.74 \pm 0.05$ & $5.59 \pm 0.05$ & $5.50 \pm 0.04$ & $3.98 \pm 0.01$ & $3.92 \pm 0.01$ & $3.89 \pm 0.01$ &  $3.89 \pm 0.01$   \\ \hline
0.4 & $6.5 \pm 0.09$ & $5.59 \pm 0.06$ & $5.29 \pm 0.05$ & $5.16 \pm 0.04$ & $3.13 \pm 0.01$ & $3.04 \pm 0.01$ & $3.00 \pm 0.01$ &   $3.00 \pm 0.01$  \\ \hline
0.5 & $7.83 \pm 0.17$ & $5.93 \pm 0.08$ & $5.38 \pm 0.05$ & $5.19 \pm 0.06$ & $2.63 \pm 0.01$ & $2.50 \pm 0.01$ & $2.46 \pm 0.01$ &   $2.46 \pm 0.01$  \\ \hline
0.6 & $9.97 \pm 0.23$ & $6.50 \pm 0.09$ & $5.60 \pm 0.07$ & $5.33 \pm 0.07$ & $2.31 \pm 0.01$ & $2.14 \pm 0.01$ & $2.11 \pm 0.01$ &   $2.10 \pm 0.01$  \\ \hline
0.7 & $14.58 \pm 0.44$ & $7.54 \pm 0.13$ & $6.04 \pm 0.10$ & $5.70 \pm 0.10$ & $2.12 \pm 0.01$ & $1.88 \pm 0.01$ & $1.84 \pm 0.01$ &   $1.84 \pm 0.01$  \\ \hline
0.8 & $23.92 \pm 0.98$ & $9.13 \pm 0.21$ & $6.52 \pm 0.12$ & $6.09 \pm 0.12$ & $2.02 \pm 0.01$ & $1.68 \pm 0.01$ & $1.64 \pm 0.01$ &   $1.63 \pm 0.01$  \\ \hline
0.9 & $46.96 \pm 2.62$ & $11.68 \pm 0.22$ & $6.92 \pm 0.08$ & $6.40 \pm 0.08$ & $2.06 \pm 0.02$ & $1.54 \pm 0.04$ & $1.49 \pm 0.01$ &  $1.48 \pm 0.01$   \\ \hline
\end{tabular}
\label{table:tradition-wei-AoI}
\caption{The $95\% $ confidence interval of Fig.~\ref{fig:tradition-wei-AoI}}
\vspace{0pt}
\end{table*}

\begin{table*}[b]
\centering
\begin{tabular}{|c|c|c|c|c|c|c|c|c|}
\hline
 & FCFS & RANDOM & LCFS & SJF & PS & LCFS\_P & SRPT & SJF\_P \\ \hline
1 & $4.06 \pm 0.02$ & $3.32 \pm 0.01$ & $2.92 \pm 0.01$ & $2.86 \pm 0.01$ & $2.99 \pm 0.01$ & $2.43 \pm 0.01$ & $2.27 \pm 0.01$ &  $2.25 \pm 0.01$  \\ \hline
2 & $5.17 \pm 0.04$ & $3.83 \pm 0.01$ & $3.25 \pm 0.01$ & $3.12 \pm 0.01$ & $2.69 \pm 0.01$ & $2.24 \pm 0.01$ & $2.13 \pm 0.01$ &  $2.11 \pm 0.01$  \\ \hline
3 & $6.26 \pm 0.06$ & $4.32 \pm 0.03$ & $3.58 \pm 0.01$ & $3.40 \pm 0.01$ & $2.52 \pm 0.01$ & $2.13 \pm 0.01$ & $2.05 \pm 0.01$ &  $2.04 \pm 0.01$  \\ \hline
4 & $7.45 \pm 0.08$ & $4.82 \pm 0.03$ & $3.93 \pm 0.02$ & $3.71 \pm 0.01$ & $2.40 \pm 0.01$ & $2.06 \pm 0.01$ & $1.99 \pm 0.01$ & $1.98 \pm 0.01$   \\ \hline
5 & $8.59 \pm 0.17$ & $5.28 \pm 0.05$ & $4.27 \pm 0.03$ & $4.02 \pm 0.03$ & $2.33 \pm 0.01$ & $2.01 \pm 0.01$ & $1.96 \pm 0.01$ & $1.95 \pm 0.01$   \\ \hline
6 & $9.67 \pm 0.26$ & $5.72 \pm 0.08$ & $4.62 \pm 0.05$ & $4.34 \pm 0.04$ & $2.26 \pm 0.01$ & $1.97 \pm 0.01$ & $1.92 \pm 0.01$ &  $1.91 \pm 0.01$  \\ \hline
7 & $10.72 \pm 0.27$& $6.14 \pm 0.08$ & $4.94 \pm 0.05$ & $4.65 \pm 0.05$ & $2.21 \pm 0.01$ & $1.94 \pm 0.01$ & $1.90 \pm 0.01$ &  $1.89 \pm 0.01$  \\ \hline
8 & $12.12 \pm 0.30$ & $6.68 \pm 0.09$ & $5.35 \pm 0.06$ & $5.04 \pm 0.06$ & $2.17 \pm 0.01$ & $1.92 \pm 0.01$ & $1.87 \pm 0.01$ &  $1.87 \pm 0.01$  \\ \hline
9 & $13.78 \pm 0.57$& $7.21 \pm 0.13$ & $5.74 \pm 0.09$ & $5.41 \pm 0.08$ & $2.15 \pm 0.01$ & $1.89 \pm 0.01$ & $1.85 \pm 0.01$ &  $1.85 \pm 0.01$  \\ \hline
10 & $14.30 \pm 0.35$& $7.50 \pm 0.09$ & $6.02 \pm 0.07$ & $5.68 \pm 0.06$ & $2.11 \pm 0.01$ & $1.88 \pm 0.01$ & $1.84 \pm 0.01$ & $1.84 \pm 0.01$   \\ \hline
\end{tabular}
\label{table:tradition-wei-variance-AoI}
\caption{The $95\% $ confidence interval of Fig.~\ref{fig:tradition-wei-variance-AoI}}
\vspace{0pt}
\end{table*}
\begin{table*}[b]
\centering
\begin{tabular}{|c|c|c|c|c|c|c|c|c|}
\hline
 & FCFS & RANDOM & LCFS & SJF & PS & LCFS\_P & SRPT & SJF\_P \\ \hline
0.1 & $11.10\pm 0.01$ & $11.14 \pm 0.01$ & $11.17 \pm 0.01$  & $11.13 \pm 0.01$ & $11.51 \pm 0.01$ & $11.90 \pm 0.01$ & $11.45 \pm 0.01$ & $11.65 \pm 0.01$   \\ \hline
0.2 & $6.25 \pm 0.01$ & $6.29 \pm 0.01$ & $6.32 \pm 0.01$  & $6.27 \pm 0.01$ & $6.56 \pm 0.01$ & $6.83 \pm 0.01$ & $6.42 \pm 0.01$ &  $6.58 \pm 0.01$  \\ \hline
0.3 & $4.76 \pm 0.01$ & $4.77 \pm 0.01$ & $4.78 \pm 0.01$  & $4.71 \pm 0.01$ & $4.94 \pm 0.01$ & $5.11 \pm 0.01$ & $4.71 \pm 0.01$ &  $4.84 \pm 0.01$  \\ \hline
0.4 & $4.17 \pm 0.01$ & $4.09 \pm 0.01$ & $4.06 \pm 0.01$  & $3.98 \pm 0.01$ & $4.18 \pm 0.01$ & $4.22 \pm 0.01$ & $3.85 \pm 0.01$ &  $3.95 \pm 0.01$  \\ \hline
0.5 & $4.00 \pm 0.01$ & $3.76 \pm 0.01$ & $3.67 \pm 0.01$ & $3.58 \pm 0.01$ & $3.76 \pm 0.01$ & $3.67 \pm 0.01$ & $3.32 \pm 0.01$ &  $3.39 \pm 0.01$  \\ \hline
0.6 & $4.16 \pm 0.01$ & $3.64 \pm 0.01$ & $3.45 \pm 0.01$  & $3.37 \pm 0.01$ & $3.57 \pm 0.01$ & $3.29 \pm 0.01$ & $2.96 \pm 0.01$ &  $3.00 \pm 0.01$  \\ \hline
0.7 & $4.74 \pm 0.01$ & $3.68 \pm 0.01$ & $3.35 \pm 0.01$  & $3.23 \pm 0.01$ & $3.54 \pm 0.01$ & $3.02 \pm 0.01$ & $2.69 \pm 0.01$ &  $2.72 \pm 0.01$  \\ \hline
0.8 & $6.24 \pm 0.01$ & $3.95 \pm 0.01$ & $3.32 \pm 0.01$  & $3.17 \pm 0.01$ & $3.73 \pm 0.01$ & $2.81 \pm 0.01$ & $2.49 \pm 0.01$ & $2.50 \pm 0.01$   \\ \hline
0.9 & $11.02 \pm 0.01$ & $4.74 \pm 0.01$ & $3.37 \pm 0.01$  & $3.15 \pm 0.01$ & $4.42 \pm 0.01$ & $2.64 \pm 0.01$ & $2.33 \pm 0.01$ &  $2.32 \pm 0.01$  \\ \hline
\end{tabular}
\label{table:tradition-exp-PAoI}
\caption{The $95\% $ confidence interval of Fig.~\ref{fig:tradition-exp-PAoI}}
\vspace{0pt}
\end{table*}

\begin{table*}[]
\centering
\begin{tabular}{|c|c|c|c|c|c|c|c|c|}
\hline
 & FCFS & RANDOM & LCFS & SJF & PS & LCFS\_P & SRPT & SJF\_P \\ \hline
0.1 & $11.59 \pm 0.01$ & $11.48 \pm 0.01$ & $11.47 \pm 0.01$ & $11.45 \pm 0.01$ & $11.28 \pm 0.01$ & $11.34 \pm 0.01$ & $11.22 \pm 0.01$ & $11.26 \pm 0.01$ \\ \hline
0.2 & $7.35 \pm 0.01$ & $6.77 \pm 0.01$ & $6.66 \pm 0.01$ & $6.69 \pm 0.01$ & $6.09 \pm 0.01$ & $6.09 \pm 0.01$ & $5.99 \pm 0.01$ & $6.02 \pm 0.01$ \\ \hline
0.3 & $6.73 \pm 0.02$ & $5.35 \pm 0.01$ & $5.15 \pm 0.01$ & $5.21 \pm 0.01$ & $4.33 \pm 0.01$ & $4.28 \pm 0.01$ & $4.20 \pm 0.01$ & $4.22 \pm 0.01$ \\ \hline
0.4 & $7.18 \pm 0.02$ & $4.77 \pm 0.01$ & $4.46 \pm 0.01$ & $4.53 \pm 0.01$ & $3.44 \pm 0.01$ & $3.35 \pm 0.01$ & $3.27 \pm 0.01$ & $3.29 \pm 0.01$ \\ \hline
0.5 & $8.60 \pm 0.04$ & $4.57 \pm 0.01$ & $4.12 \pm 0.01$ & $4.19 \pm 0.01$ & $2.91 \pm 0.01$ & $2.77 \pm 0.01$ & $2.71 \pm 0.01$ & $2.72 \pm 0.01$ \\ \hline
0.6 & $10.79 \pm 0.05$ & $4.60 \pm 0.01$ & $3.98 \pm 0.01$ & $4.01 \pm 0.01$ & $2.58 \pm 0.01$ & $2.38 \pm 0.01$ & $2.32 \pm 0.01$ & $2.33 \pm 0.01$ \\ \hline
0.7 & $15.42 \pm 0.10$ & $4.87 \pm 0.01$ & $3.97 \pm 0.01$ & $3.95 \pm 0.01$ & $2.37 \pm 0.01$ & $2.10 \pm 0.01$ & $2.04 \pm 0.01$ & $2.04 \pm 0.01$ \\ \hline
0.8 & $24.81 \pm 0.22$ & $5.63 \pm 0.02$ & $4.13 \pm 0.01$ & $4.00 \pm 0.01$ & $2.27 \pm 0.01$ & $1.88 \pm 0.01$ & $1.82 \pm 0.01$ & $1.83 \pm 0.01$ \\ \hline
0.9 & $47.94 \pm 0.59$ & $7.29 \pm 0.03$ & $4.41 \pm 0.01$ & $4.10 \pm 0.01$ & $2.31 \pm 0.01$ & $1.71 \pm 0.01$ & $1.65 \pm 0.01$ & $1.66 \pm 0.01$ \\ \hline
\end{tabular}
\label{table:tradition-wei-PAoI}
\caption{The $95\% $ confidence interval of Fig.~\ref{fig:tradition-wei-PAoI}}
\vspace{0pt}
\end{table*}
\begin{table*}[]
\centering
\begin{tabular}{|c|c|c|c|c|c|c|c|c|}
\hline
 & FCFS & RANDOM & LCFS & SJF & PS & LCFS\_P & SRPT & SJF\_P \\ \hline
1 & $4.76 \pm 0.01$ & $3.69 \pm 0.01$ & $3.35 \pm 0.01$ & $3.23 \pm 0.01$ & $3.54 \pm 0.01$ & $3.02 \pm 0.01$ & $2.69 \pm 0.01$ & $2.72 \pm 0.01$ \\ \hline
2 & $5.92 \pm 0.01$ & $3.96 \pm 0.01$ & $3.51 \pm 0.01$ & $3.36 \pm 0.01$ & $3.14 \pm 0.01$ & $2.68 \pm 0.01$ & $2.48 \pm 0.01$ & $2.49 \pm 0.01$ \\ \hline
3 & $7.04 \pm 0.02$ & $4.15 \pm 0.01$ & $3.61 \pm 0.01$ & $3.47 \pm 0.01$ & $2.91 \pm 0.01$ & $2.51 \pm 0.01$ & $2.36 \pm 0.01$ & $2.37 \pm 0.01$ \\ \hline
4 & $8.26 \pm 0.02$ & $4.31 \pm 0.01$ & $3.69 \pm 0.01$ & $3.56 \pm 0.01$ & $2.76 \pm 0.01$ & $2.39 \pm 0.01$ & $2.27 \pm 0.01$ & $2.28 \pm 0.01$ \\ \hline
5 & $9.43 \pm 0.04$ & $4.44 \pm 0.01$ & $3.76 \pm 0.01$ & $3.65 \pm 0.01$ & $2.66 \pm 0.01$ & $2.31 \pm 0.01$ & $2.21 \pm 0.01$ & $2.22 \pm 0.01$ \\ \hline
6 & $10.53 \pm 0.06$ & $4.54 \pm 0.01$ & $3.81 \pm 0.01$ & $3.72 \pm 0.01$ & $2.57 \pm 0.01$ & $2.25 \pm 0.01$ & $2.16 \pm 0.01$ & $2.17 \pm 0.01$ \\ \hline
7 & $11.56 \pm 0.06$ & $4.63 \pm 0.01$ & $3.85 \pm 0.01$ & $3.78 \pm 0.01$ & $2.51 \pm 0.01$ & $2.20 \pm 0.01$ & $2.12 \pm 0.01$ & $2.13 \pm 0.01$ \\ \hline
8 & $12.99 \pm 0.07$ & $4.73 \pm 0.01$ & $3.90 \pm 0.01$ & $3.85 \pm 0.01$ & $2.45 \pm 0.01$ & $2.16 \pm 0.01$ & $2.09 \pm 0.01$ & $2.09 \pm 0.01$ \\ \hline
9 & $14.63 \pm 0.13$ & $4.87 \pm 0.01$ & $3.96 \pm 0.01$ & $3.93 \pm 0.01$ & $2.41 \pm 0.01$ & $2.12 \pm 0.01$ & $2.06 \pm 0.01$ & $2.07 \pm 0.01$ \\ \hline
10 & $15.15 \pm 0.08$ & $4.89 \pm 0.01$ & $3.99 \pm 0.01$ & $3.96 \pm 0.01$ & $2.37 \pm 0.01$ & $2.10 \pm 0.01$ & $2.04 \pm 0.01$ & $2.05 \pm 0.01$ \\ \hline
\end{tabular}
\caption{The $95\% $ confidence interval of Fig.~\ref{fig:tradition-wei-variance-PAoI}}
\label{table:tradition-wei-variance-PAoI}
\end{table*}

\newpage
\subsection{Additional Simulation Results for the G/G/$1$ Queue}
\label{appendix:more_simu}
We present additional simulation results for the G/G/$1$ queue in Figs.~\ref{fig:tradition-AoI-2}-\ref{fig:all-combined-paoi}. For the simulations in subfigures (a)-(c), we assume that the interarrival time follows a Weibull distribution with ${C^2} = 10$. 
In subfigure (a), we assume that the update size follows an Exponential distribution with mean $1/\mu  = 1$; in subfigures (b) and (c), we assume that the update size  follows a Weibull distribution with mean $1/\mu  = 1$. Note that in subfigures (a) and (b), we change the value of the system load $\rho$; in subfigure (c), we change the value of ${C^2}$ for the update size while fixing the system load at $\rho=0.7$.  
For the simulations in subfigures (d)-(f), we assume that both interarrival time and service time follow the same distribution (Gamma, Log-normal, and Pareto, respectively.). We fix the mean update size as $1/\mu  = 1$ and change the value of the system load $\rho$. 
Observations \ref{obs:size_better}-\ref{obs:informative} can also be made for the setting of G/G/$1$ queue.

\newpage
\begin{figure*}[!t]
    \centering
    \captionsetup{justification=centering}
    \subfigure[Interarrival time: Weibull (${C^{\rm{2}}}{\rm{ = 10}}$);
    Update size: Exponential ($\mu=1$)]{
        \begin{minipage}[t]{0.312\linewidth}
        \label{fig:tradition-wei-exp-AoI} 
		\includegraphics[width=1\textwidth]{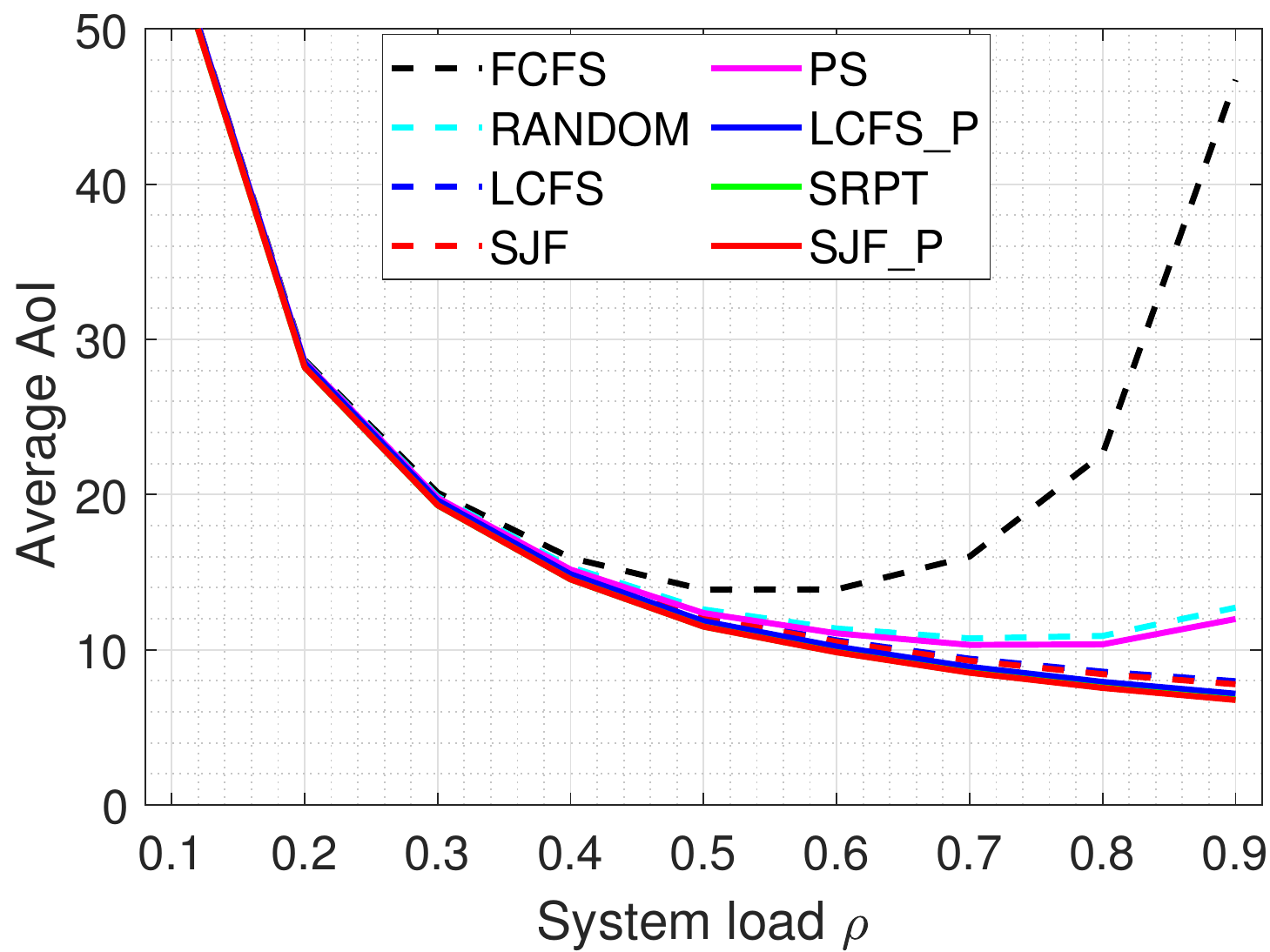}
 	    \end{minipage}}
 	\hspace{0.5em}    
	\subfigure[Interarrival time: Weibull (${C^{\rm{2}}}{\rm{ = 10}}$);   
	Update size: Weibull ($\mu=1$ and ${C^{\rm{2}}}{\rm{ = 10}}$)]{
	    \begin{minipage}[t]{0.312\linewidth}
	    \label{fig:tradition-wei-wei-AoI} 
		\includegraphics[width=1\textwidth]{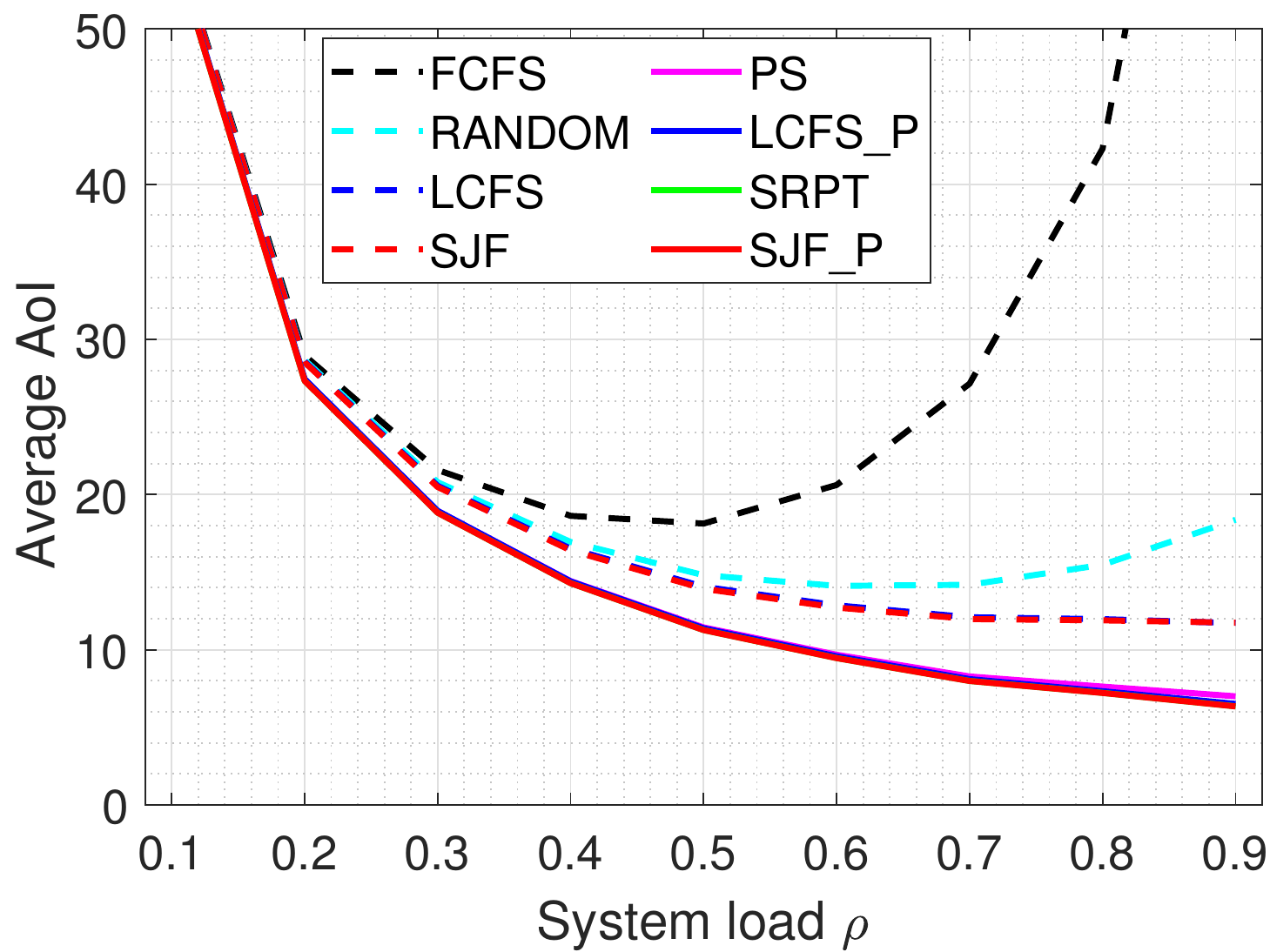}
 	\end{minipage}}
 	\hspace{0.5em}
	\subfigure[Interarrival time: Weibull (${C^{\rm{2}}}{\rm{ = 10}}$);  
	Update size: Weibull ($\mu=1$ and $\rho=0.7$)]{
	\begin{minipage}[t]{0.312\linewidth}
	  \label{fig:tradition-wei-wei-variance-AoI} 
		\includegraphics[width=1\textwidth]{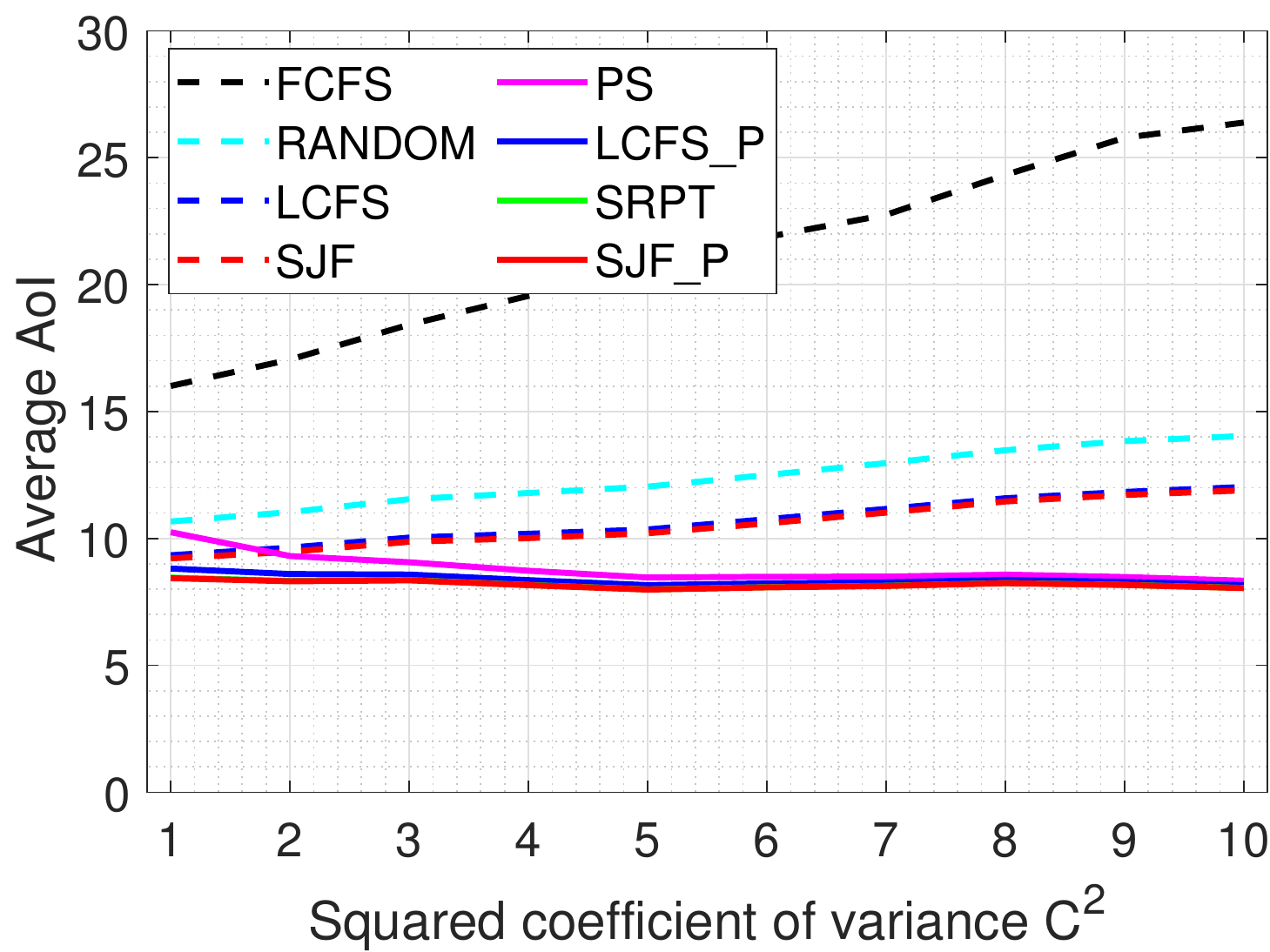}
 	\end{minipage}}
 	
	\subfigure[Interarrival time: Gamma;  
     Update size: Gamma ($\mu=1$)]{
		\label{fig:tradition-gam-gam-AoI} 
		\includegraphics[width=0.312\textwidth]{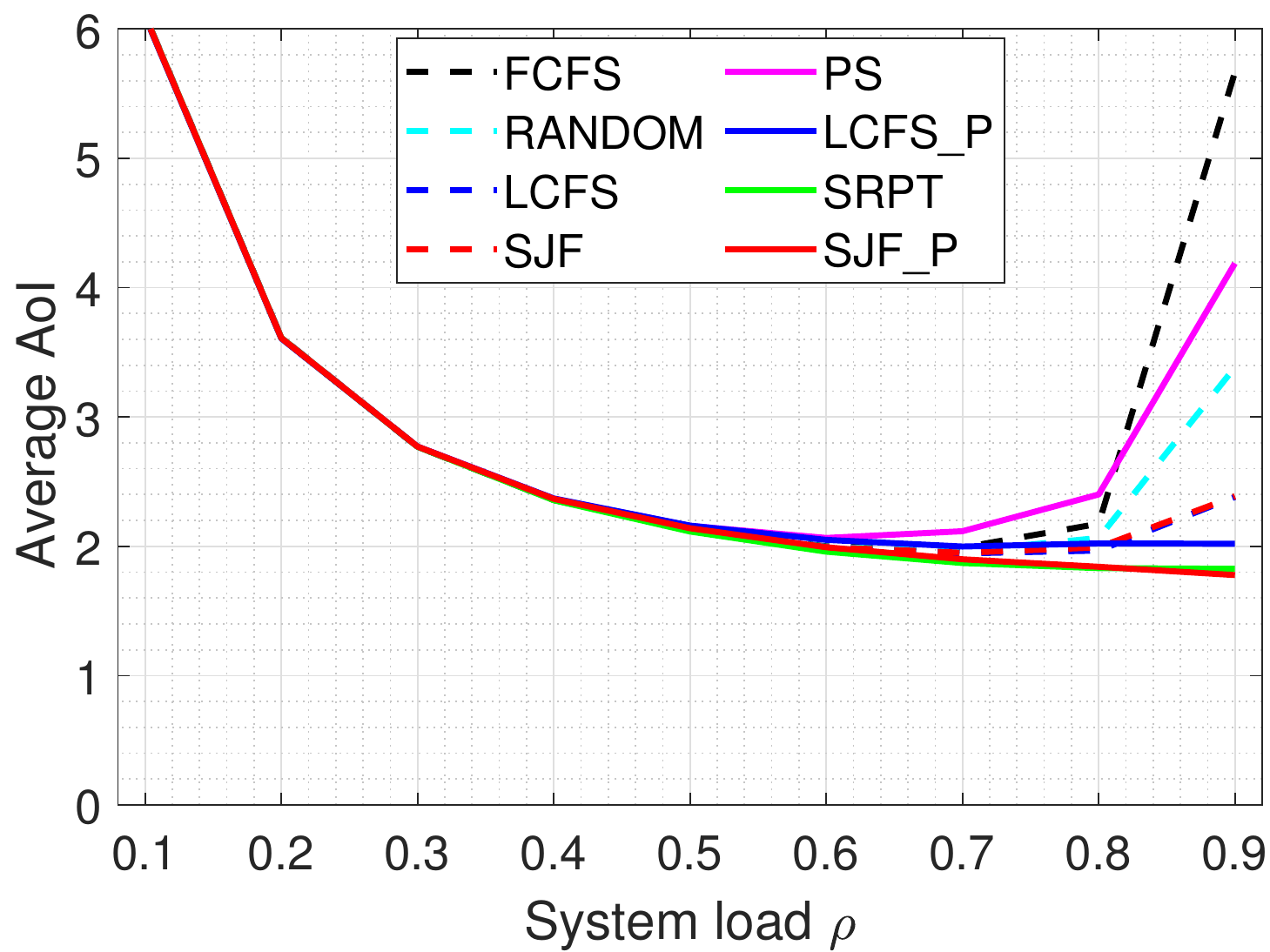}}	
	\hspace{0.5em}	
	\subfigure[Interarrival time: Log-normal;  
     Update size: Log-normal ($\mu=1$)]{
		\label{fig:tradition-log-log-AoI} 
		\includegraphics[width=0.312\textwidth]{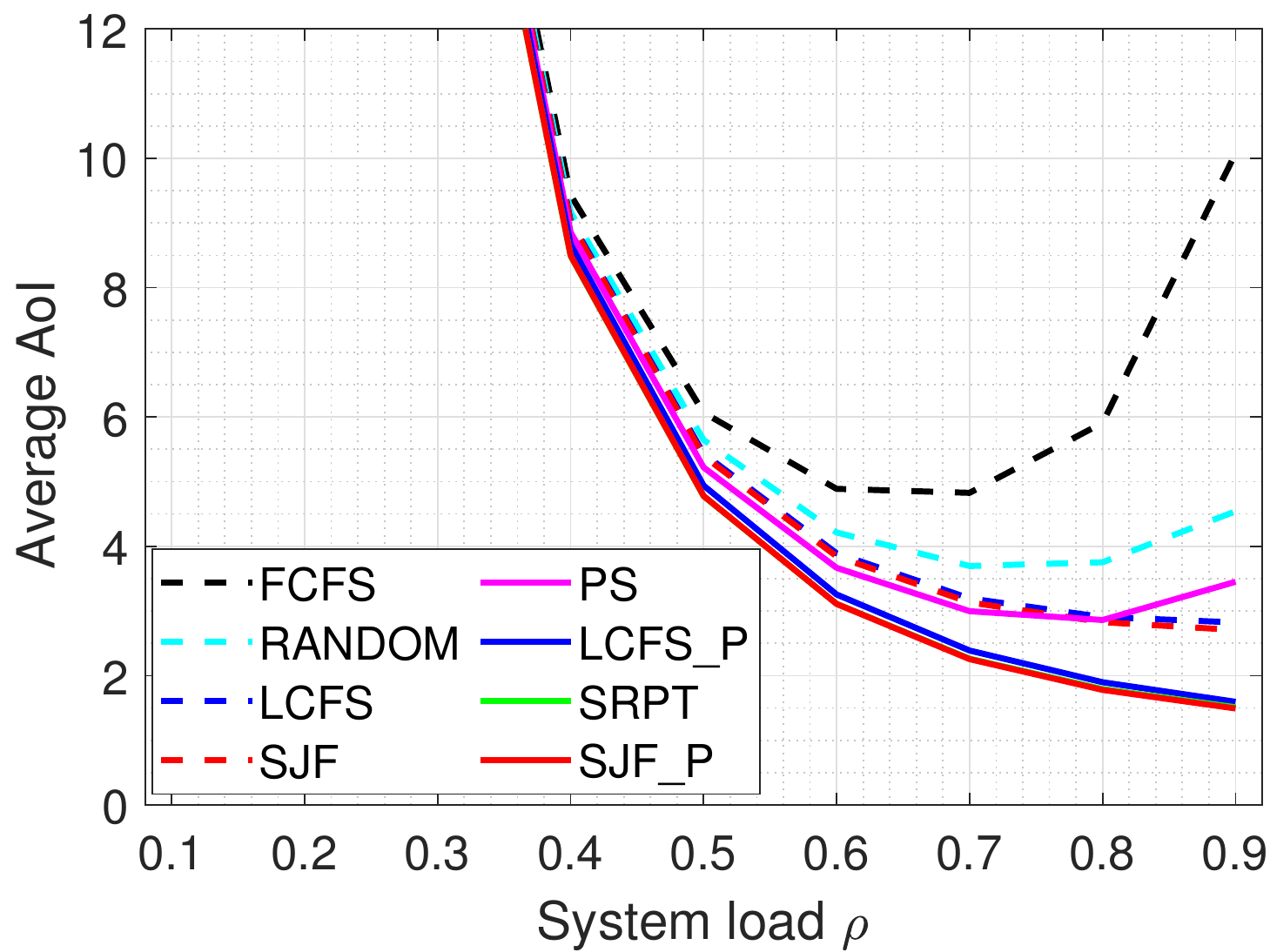}}	
	\hspace{0.5em}
	\subfigure[Interarrival time: Pareto;  
     Update size: Pareto ($\mu=1$)]{
		\label{fig:tradition-par-par-AoI} 
		\includegraphics[width=0.312\textwidth]{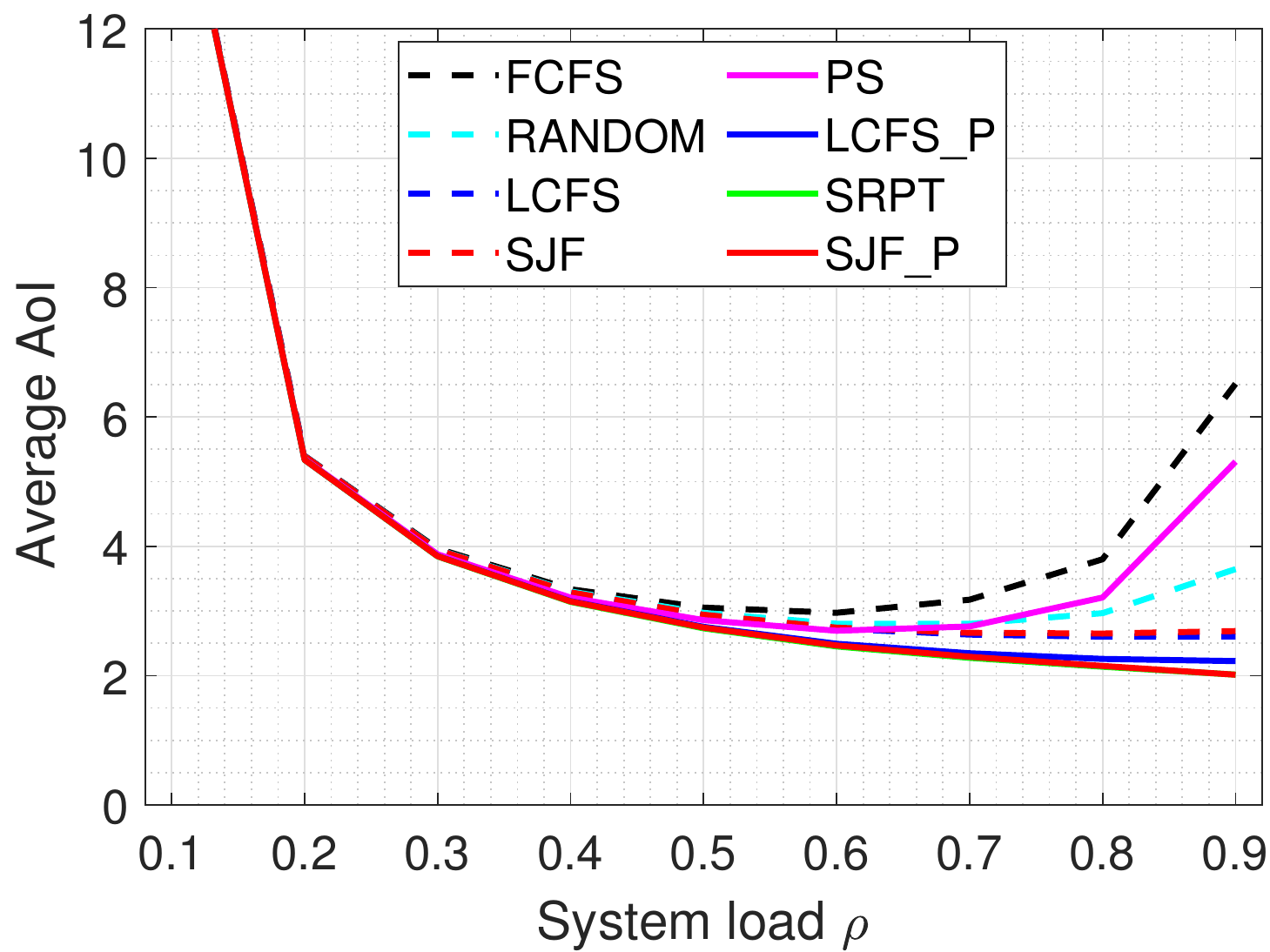}}	
	\caption{Comparisons of the average AoI performances of several common scheduling policies under different distributions}
	\label{fig:tradition-AoI-2}
	\vspace{-5pt}
\end{figure*}

\begin{figure*}[!t]
    \centering
    \setlength{\abovecaptionskip}{-1pt}
    \subfigcapskip=-1pt
     \subfigure[Interarrival time: Weibull (${C^{\rm{2}}}{\rm{ = 10}}$); 
     Update size: Exponential ($\mu=1$)]{
		\label{fig:tradition-wei-exp-PAoI} 
		\includegraphics[width=0.312\textwidth]{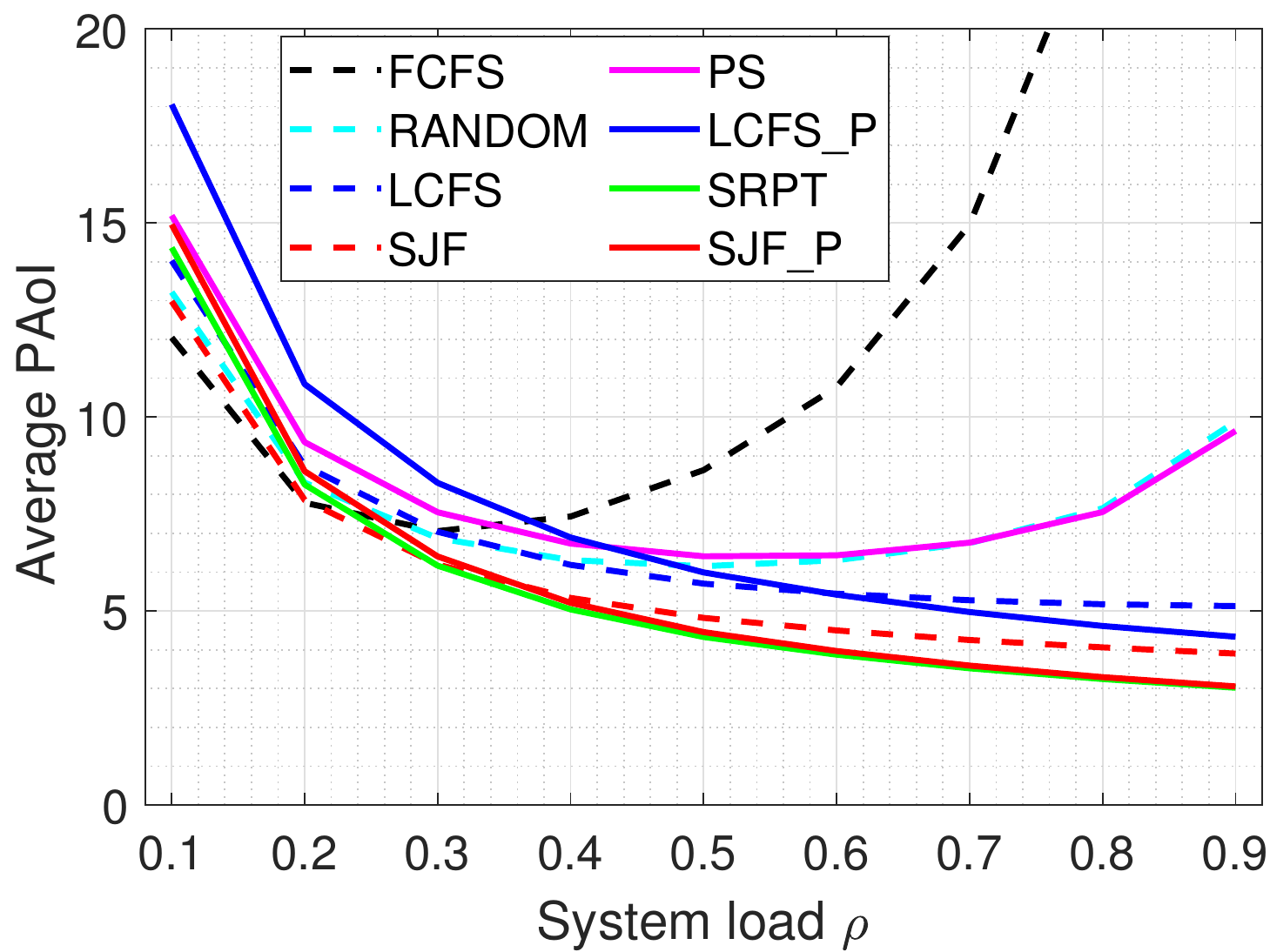}}
	\hspace{0.45em}	
	\subfigure[Interarrival time: Weibull (${C^{\rm{2}}}{\rm{ = 10}}$);
	Update size: Weibull ($\mu=1$ and ${C^{\rm{2}}}{\rm{ = 10}}$)]{
		\label{fig:tradition-wei-wei-PAoI} 
		\includegraphics[width=0.312\textwidth]{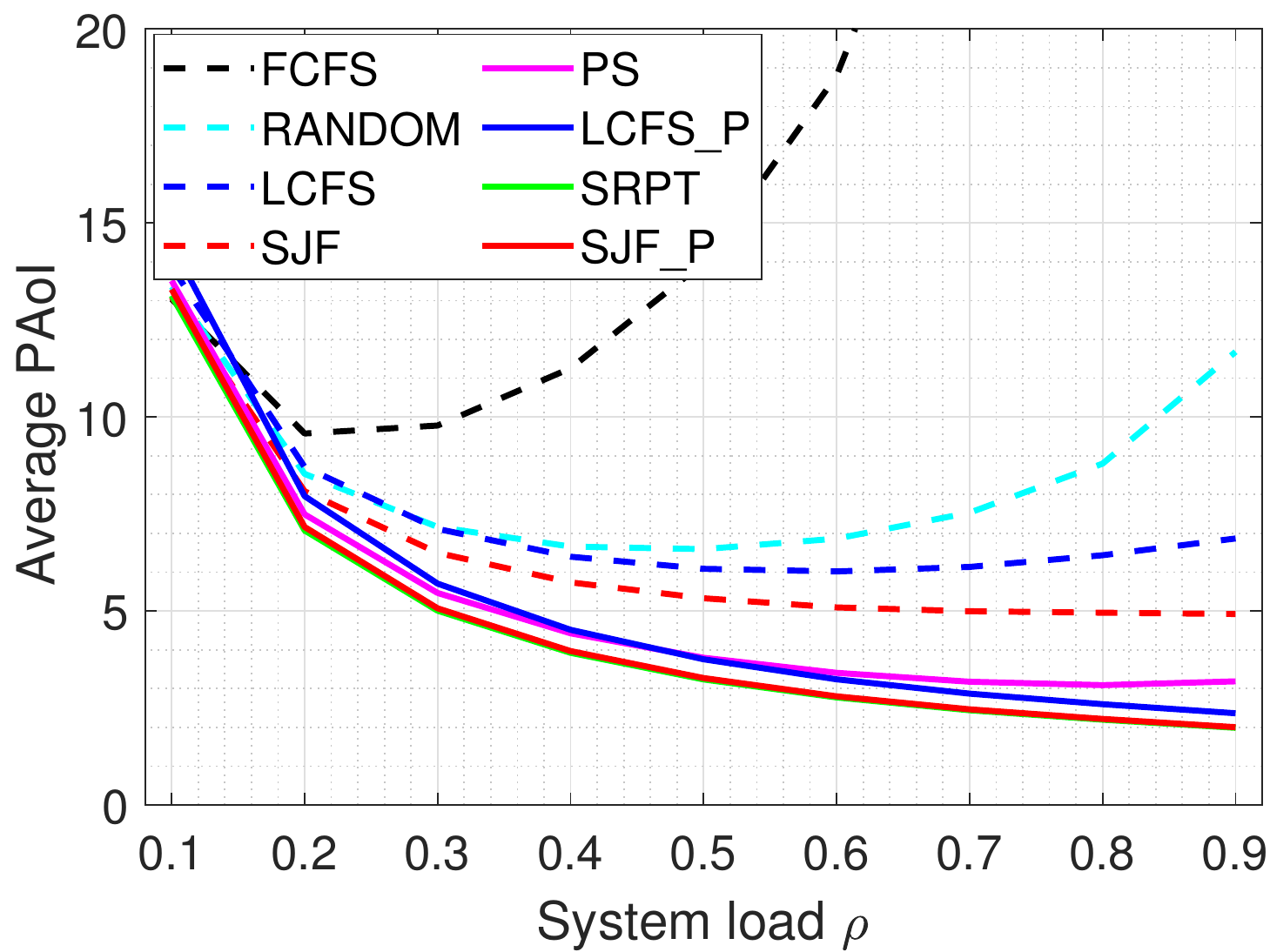}}
	\hspace{0.5em}	
	\subfigure[Interarrival time: Weibull (${C^{\rm{2}}}{\rm{ = 10}}$);  
	Update size: Weibull ($\mu=1$ and $\rho=0.7$)]{
		\label{fig:tradition-wei-wei-variance-PAoI} 
		\includegraphics[width=0.312\textwidth]{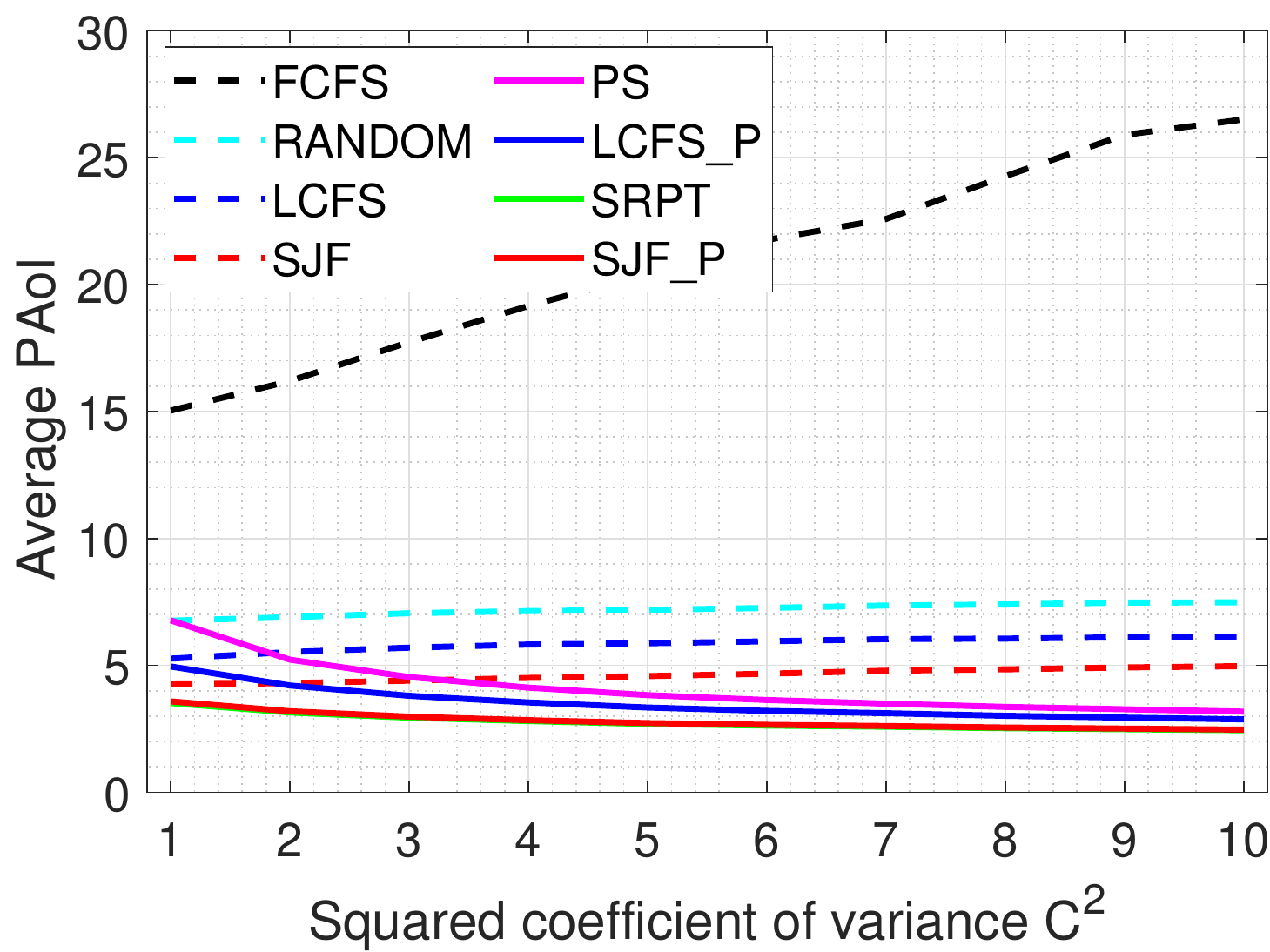}}
	\subfigure[Interarrival time: Gamma; 
     Update size: Gamma ($\mu=1$)]{
		\label{fig:tradition-gam-gam-PAoI} 
		\includegraphics[width=0.312\textwidth]{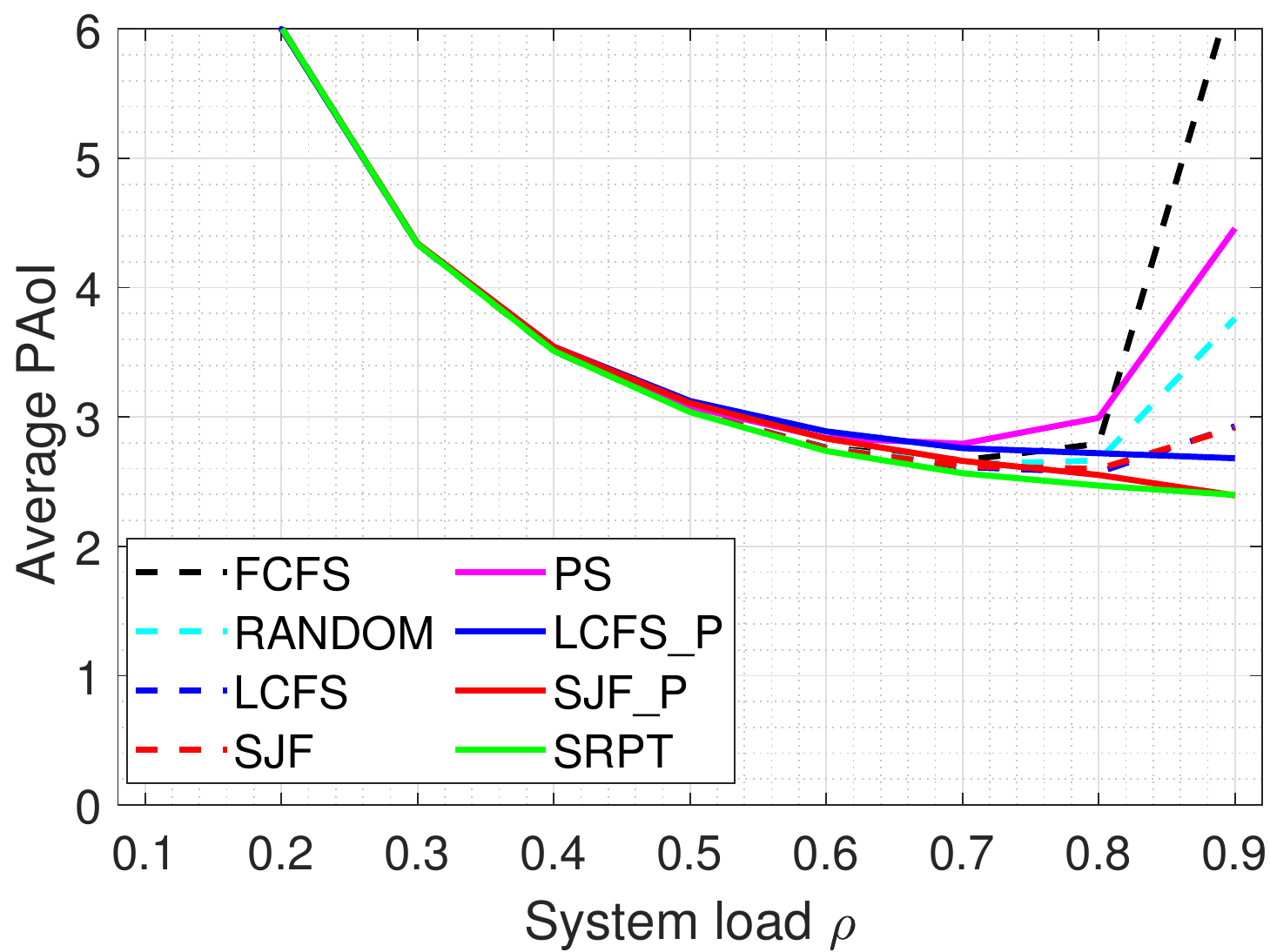}}
	\hspace{0.5em}	
	\subfigure[Interarrival time: Log-normal;  
     Update size: Log-normal ($\mu=1$)]{
		\label{fig:tradition-log-log-PAoI} 
		\includegraphics[width=0.312\textwidth]{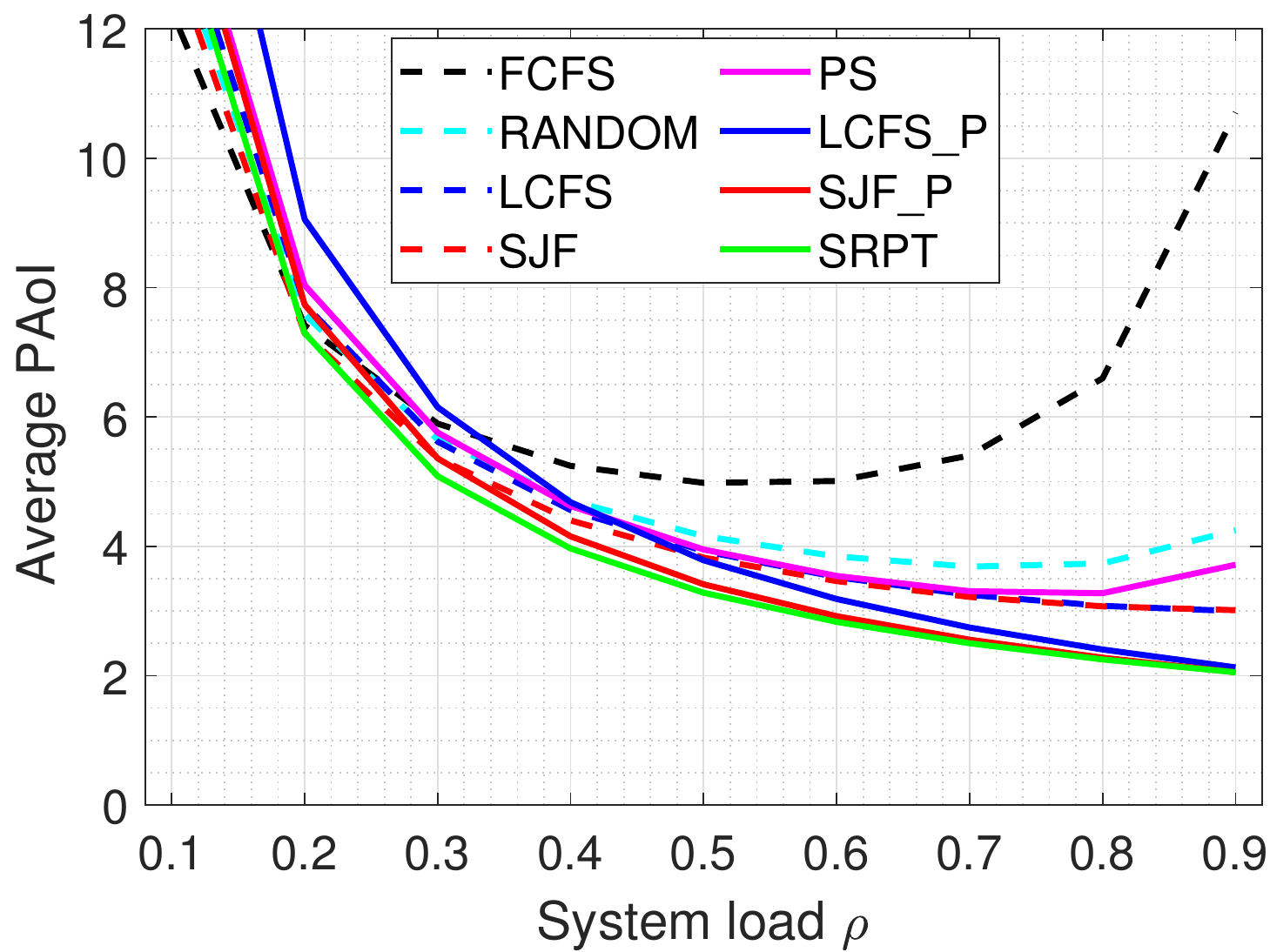}}	
	\hspace{0.5em}	
	\subfigure[Interarrival time: Pareto; 
     Update size: Pareto ($\mu=1$)]{
		\label{fig:tradition-par-par-PAoI} 
		\includegraphics[width=0.312\textwidth]{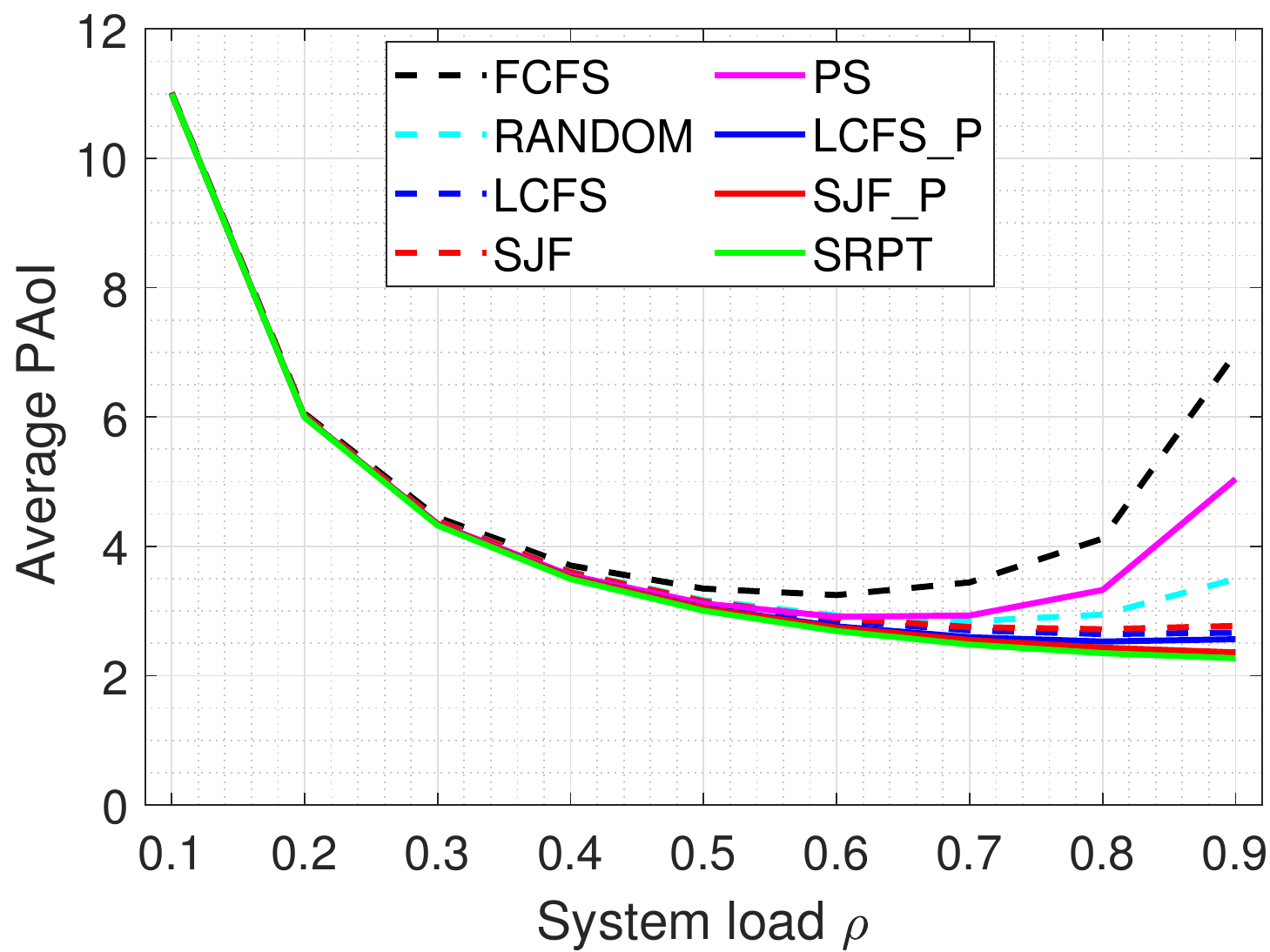}}	
	\caption{Comparisons of the average PAoI performances of several common scheduling policies under different distributions}
	\label{fig:tradition-PAoI-2}
	\vspace{-5pt}
\end{figure*}

\begin{figure*}[!t]
    \centering
    \setlength{\abovecaptionskip}{-1pt}
    \subfigcapskip=-1pt
    \subfigure[Interarrival time: Weibull (${C^{\rm{2}}}{\rm{ = 10}}$);  
    Update size: Exponential ($\mu=1$)]{
		\label{fig:aoi-based-wei-exp-AoI} 
		\includegraphics[width=0.31\textwidth]{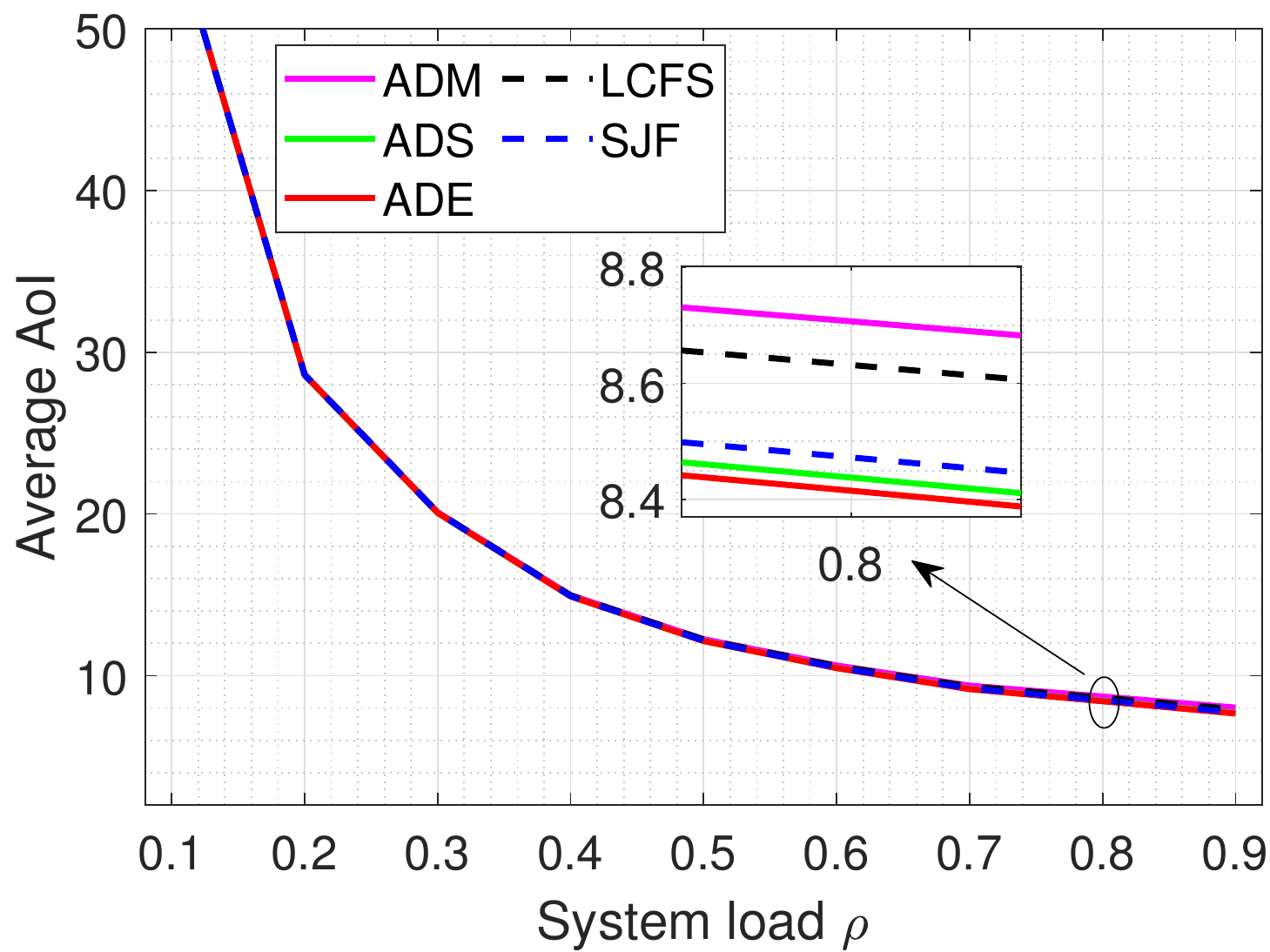}}
	\hspace{0.5em}	
	\subfigure[Interarrival time: Weibull (${C^{\rm{2}}}{\rm{ = 10}}$);  
	Update size: Weibull ($\mu=1$ and ${C^{\rm{2}}}{\rm{ = 10}}$)]{
		\label{fig:aoi-based-wei-wei-AoI} 
		\includegraphics[width=0.31\textwidth]{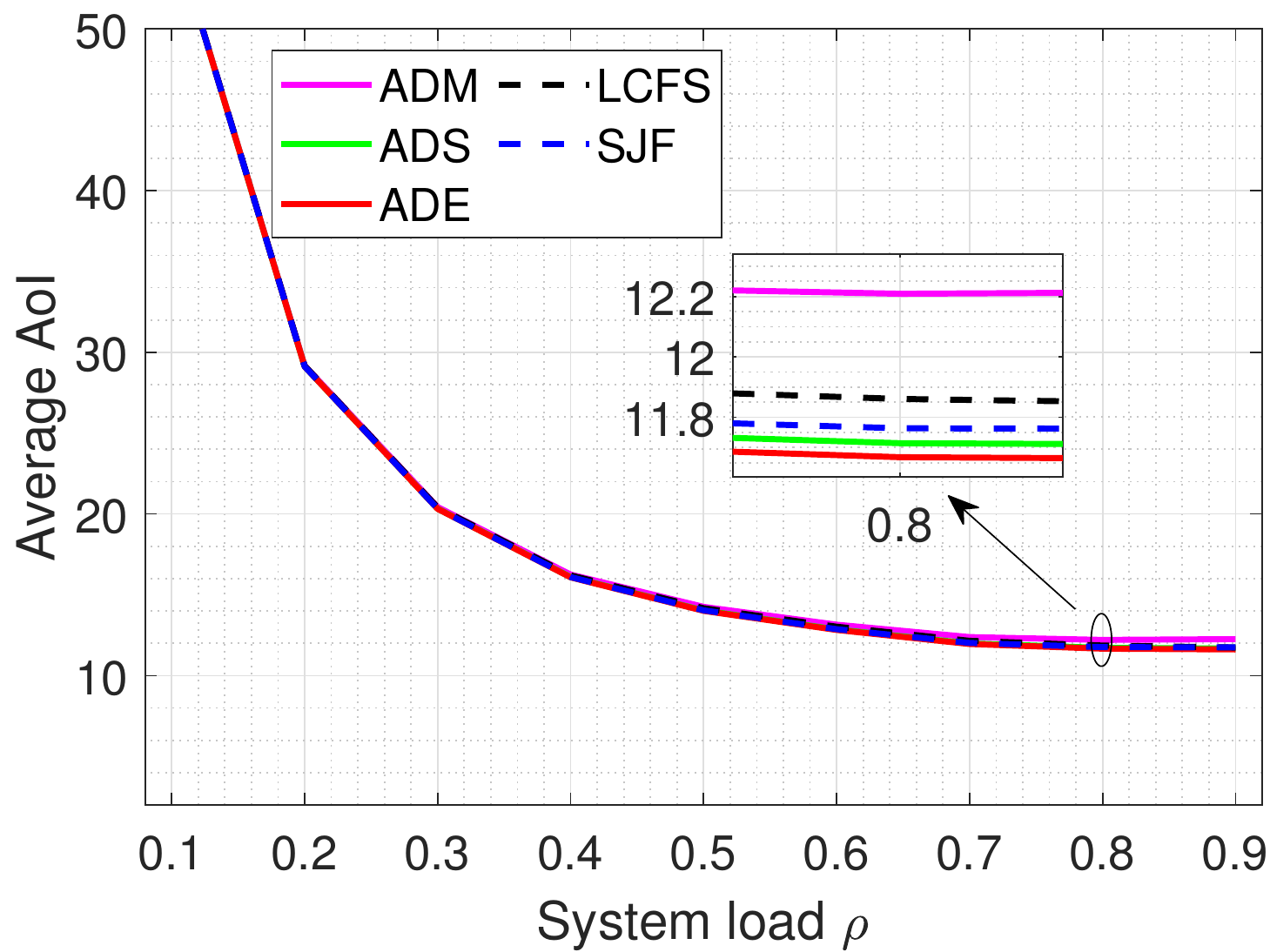}}
	\hspace{0.5em}	
	\subfigure[Interarrival time: Weibull (${C^{\rm{2}}}{\rm{ = 10}}$);  
	Update size: Weibull ($\mu=1$ and $\rho=0.7$)]{
		\label{fig:area-wei-wei-variance-aoi} 
		\includegraphics[width=0.31\textwidth]{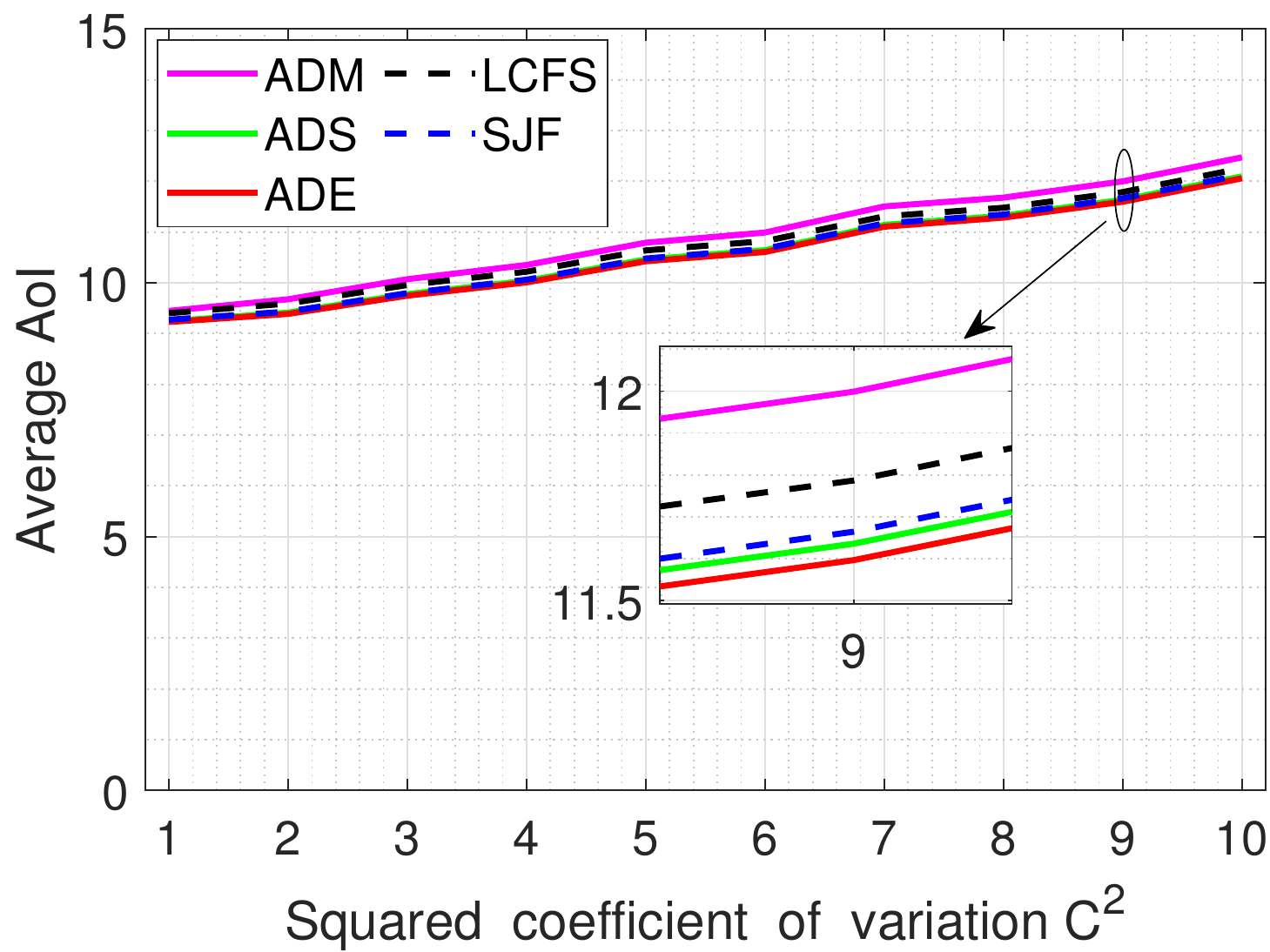}}
	\subfigure[Interarrival time: Gamma;  
     Update size: Gamma ($\mu=1$)]{
		\label{fig:gam-gam-AoI-age} 
		\includegraphics[width=0.312\textwidth]{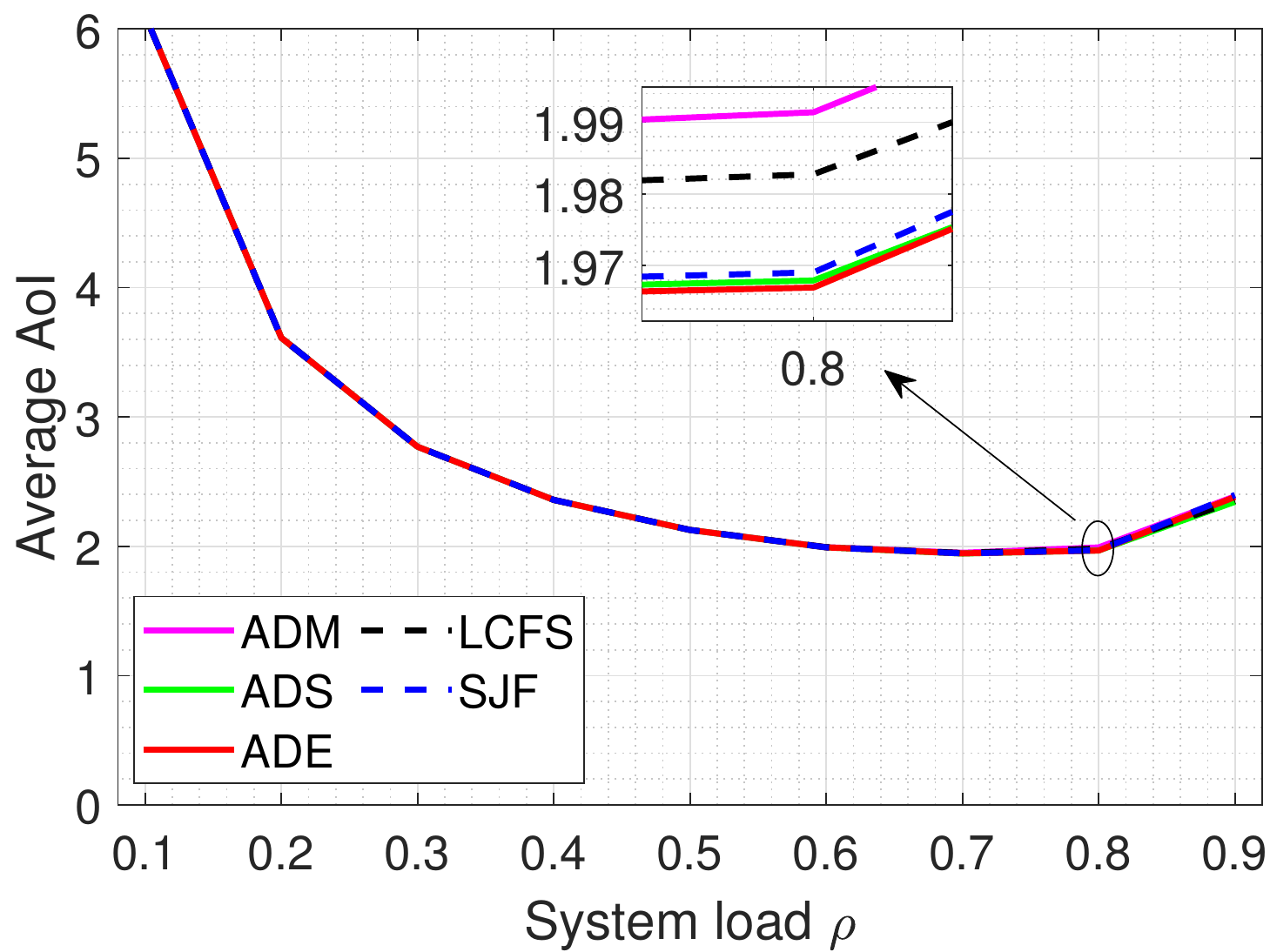}}	
	\hspace{0.5em}	
	\subfigure[Interarrival time: Log-normal;  
     Update size: Log-normal ($\mu=1$)]{
		\label{fig:log-log-AoI-age} 
		\includegraphics[width=0.312\textwidth]{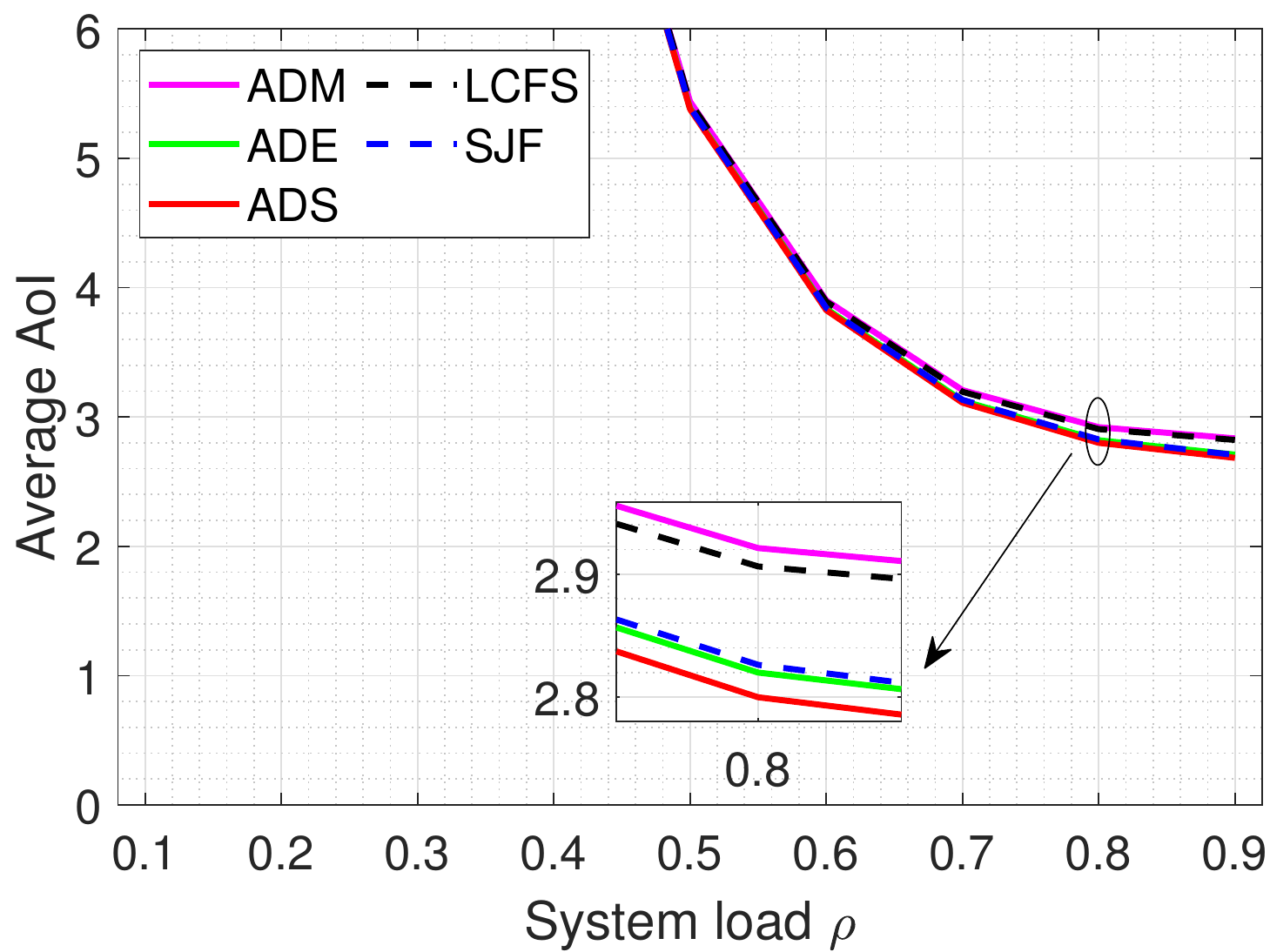}}	
	\hspace{0.5em}	
	\subfigure[Interarrival time: Pareto;  
     Update size: Pareto ($\mu=1$)]{
		\label{fig:par-par-aoi-agebased} 
		\includegraphics[width=0.312\textwidth]{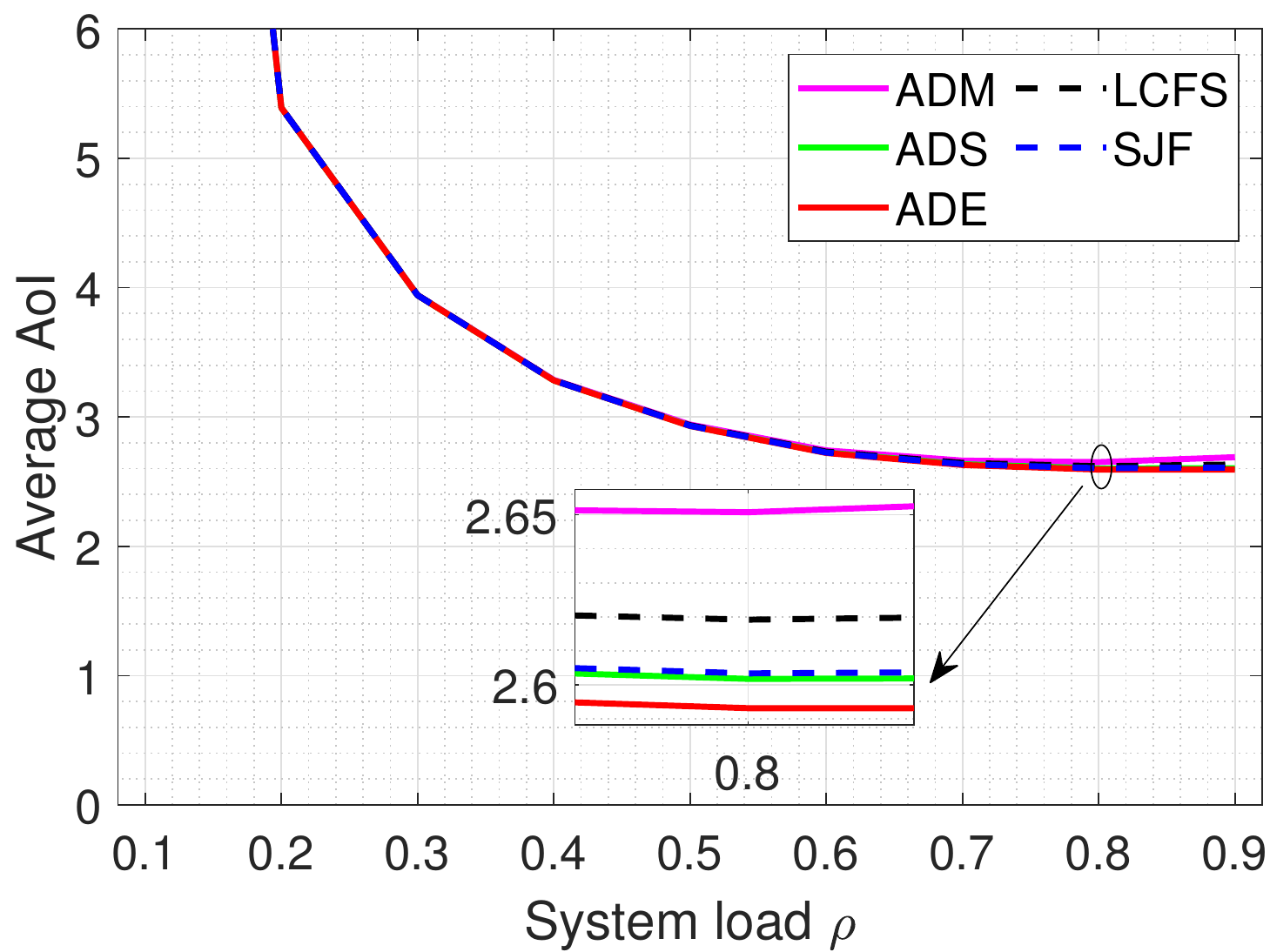}}		
	\caption{Comparisons of the average AoI performance under different distributions: AoI-based policies vs. non-AoI-based policies}
	\label{fig:aoi-based-AoI}
	\vspace{-5pt}
\end{figure*}

\begin{figure*}[!t]
    \centering
    \setlength{\abovecaptionskip}{-1pt}
    \subfigcapskip=-1pt
    \subfigure[Interarrival time: Weibull (${C^{\rm{2}}}{\rm{ = 10}}$);  
    Update size: Exponential ($\mu=1$) ]{
		\label{fig:aoi-based-wei-exp-PAoI} 
		\includegraphics[width=0.312\textwidth]{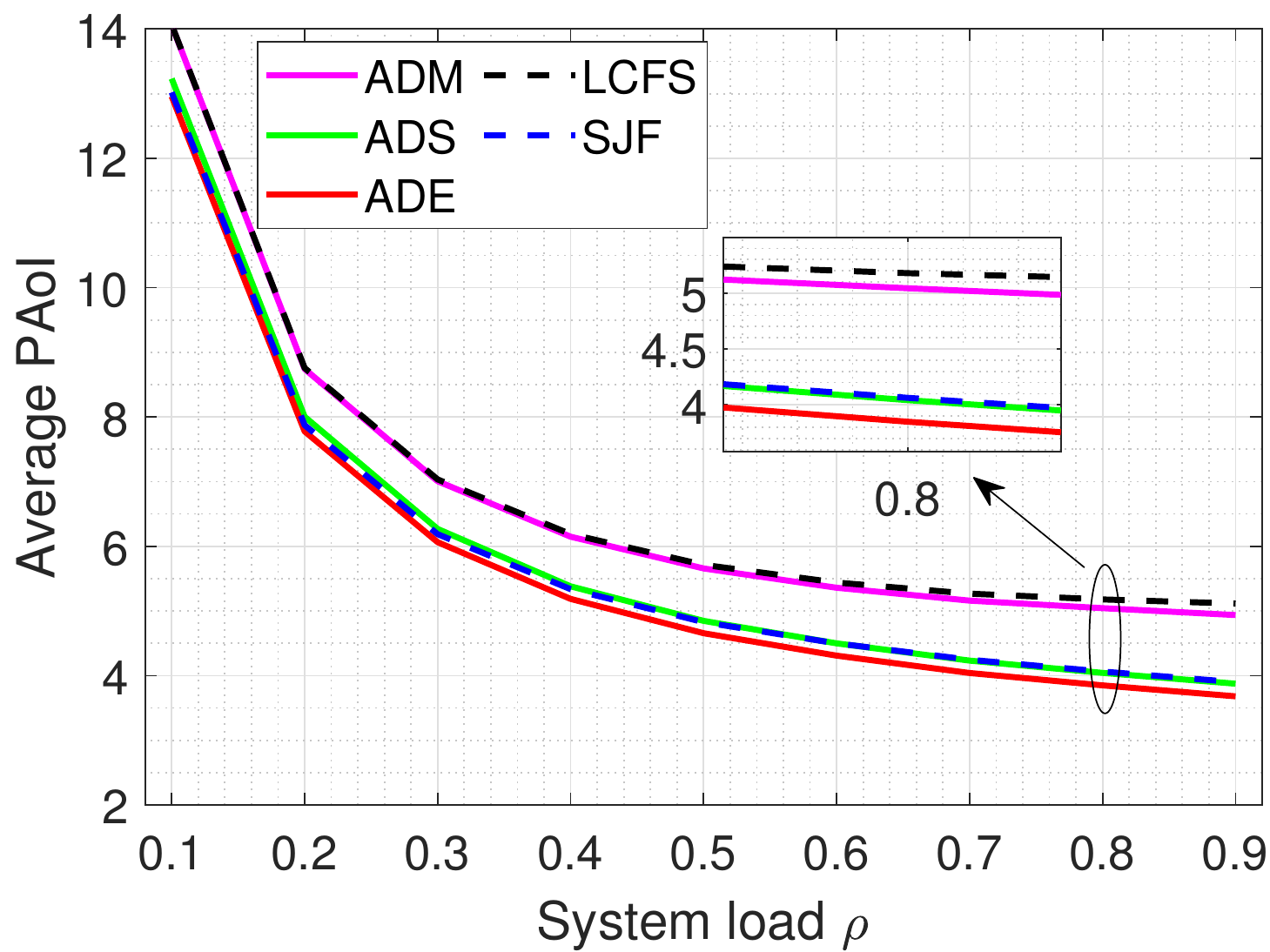}}
	\hspace{0.45em}	
	\subfigure[Interarrival time: Weibull (${C^{\rm{2}}}{\rm{ = 10}}$);  
	Update size: Weibull ($\mu=1$ and ${C^{\rm{2}}}{\rm{ = 10}}$)]{
		\label{fig:aoi-based-wei-wei-PAoI} 
		\includegraphics[width=0.312\textwidth]{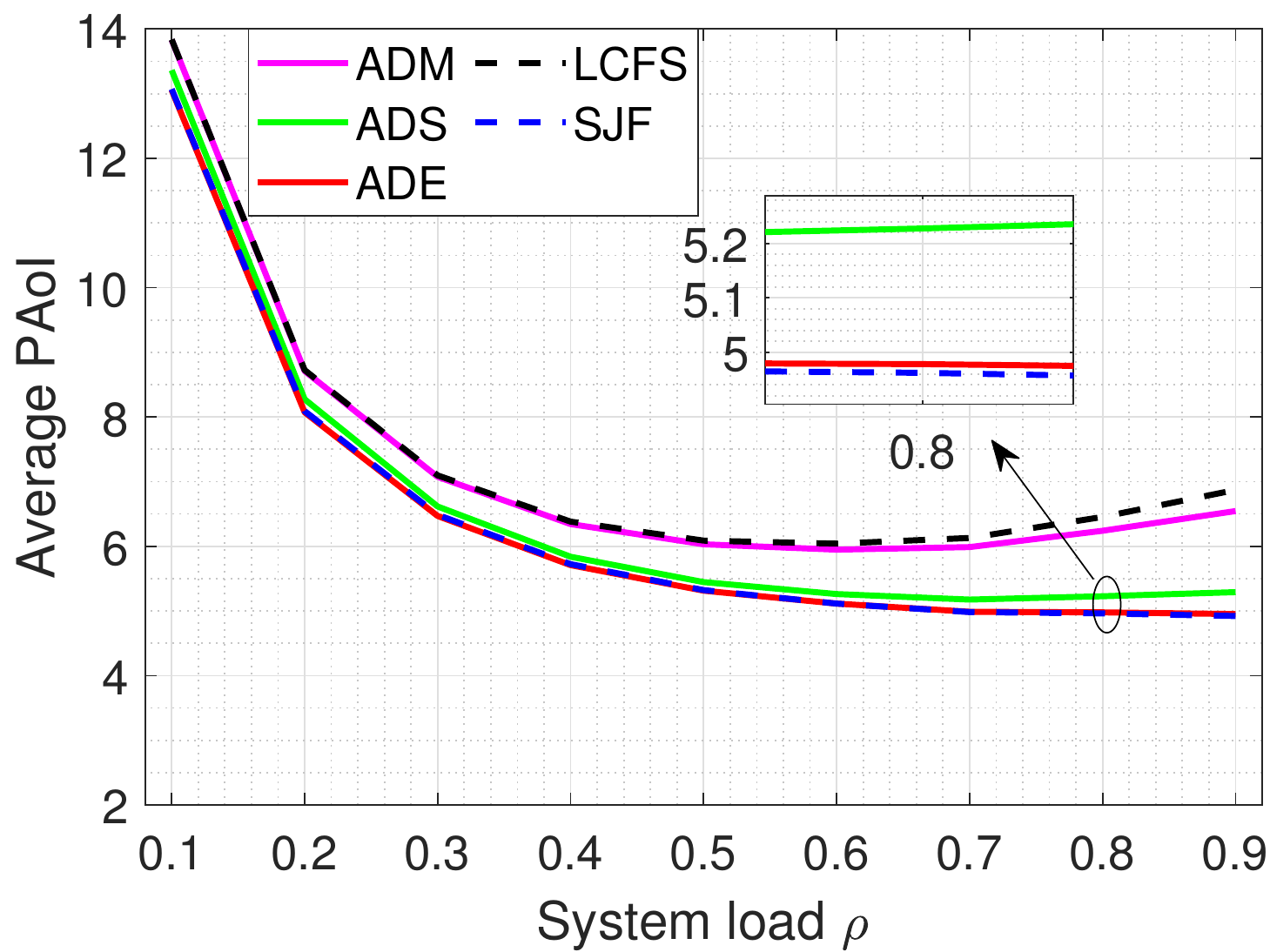}}
	\hspace{0.5em}	
	\subfigure[Interarrival time: Weibull (${C^{\rm{2}}}{\rm{ = 10}}$); 
	Update size: Weibull ($\mu=1$ and $\rho=0.7$)]{
		\label{fig:area-wei-wei-variance-paoi} 
		\includegraphics[width=0.312\textwidth]{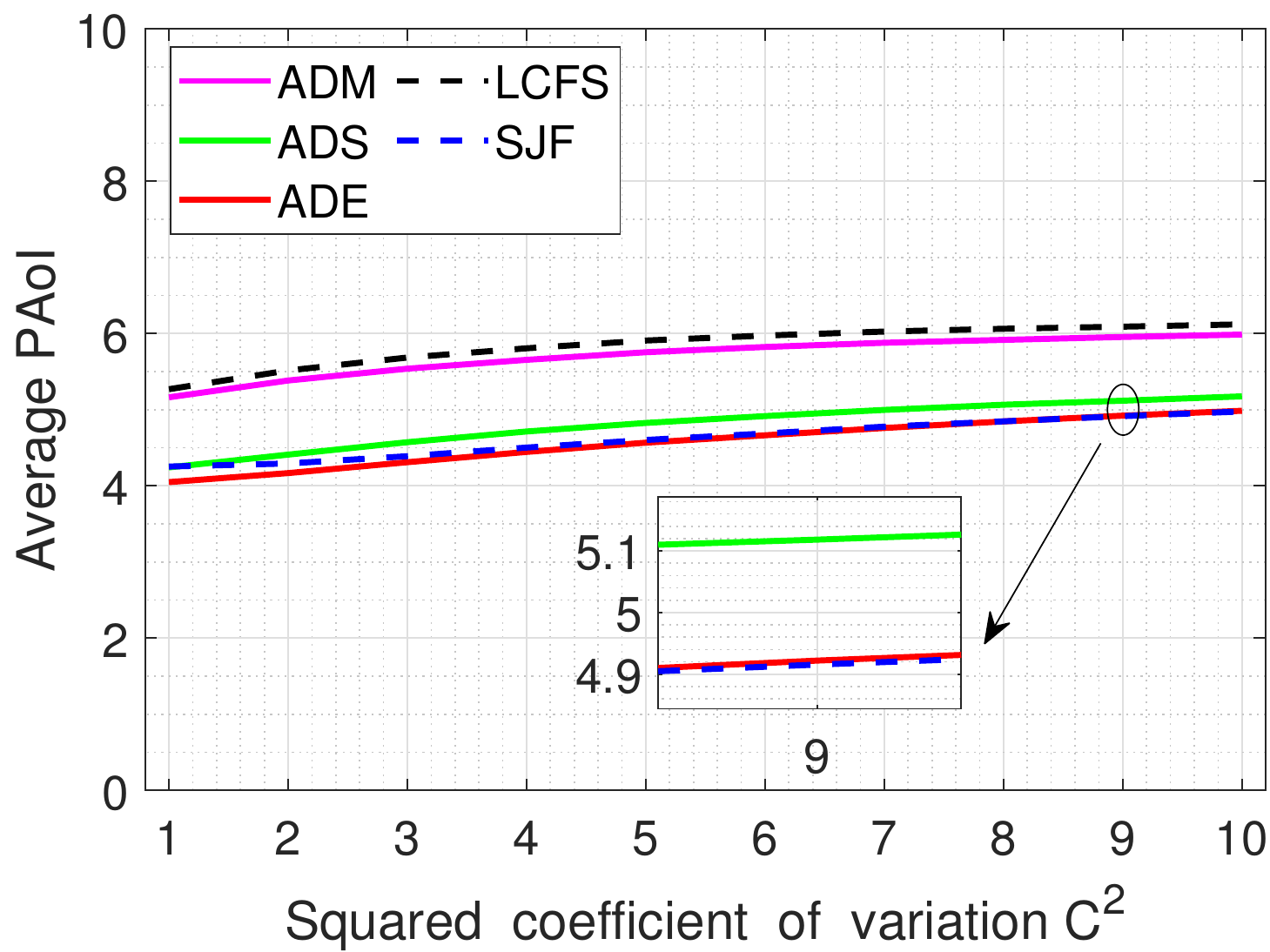}}	
	\subfigure[Interarrival time: Gamma;  
     Update size: Gamma ($\mu=1$)]{
		\label{fig:gam-gam-PAoI-age} 
		\includegraphics[width=0.312\textwidth]{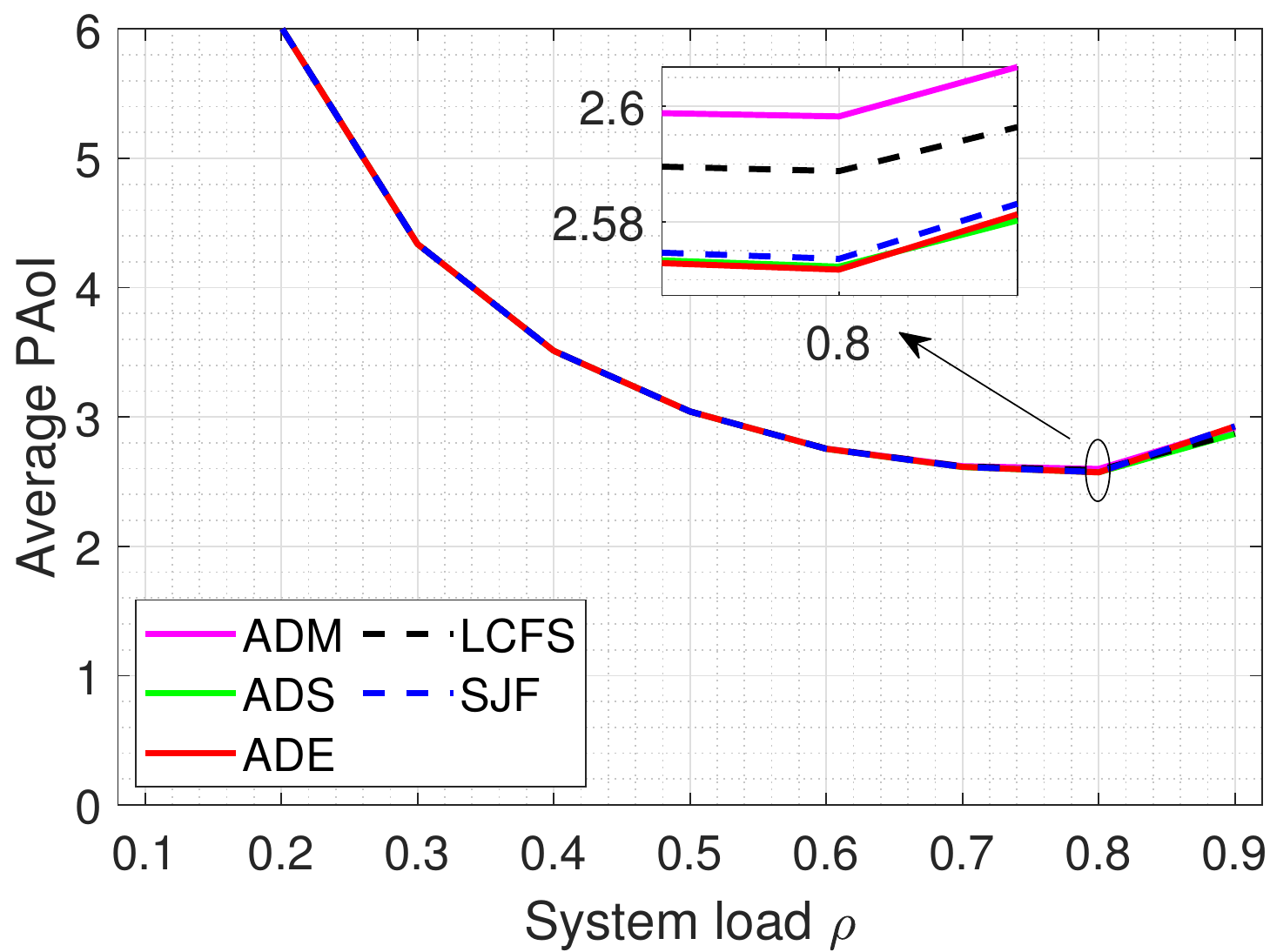}}	
	\hspace{0.5em}	
	\subfigure[Interarrival time: Log-normal;  
     Update size: Log-normal ($\mu=1$)]{
		\label{fig:log-log-PAoI-age} 
		\includegraphics[width=0.312\textwidth]{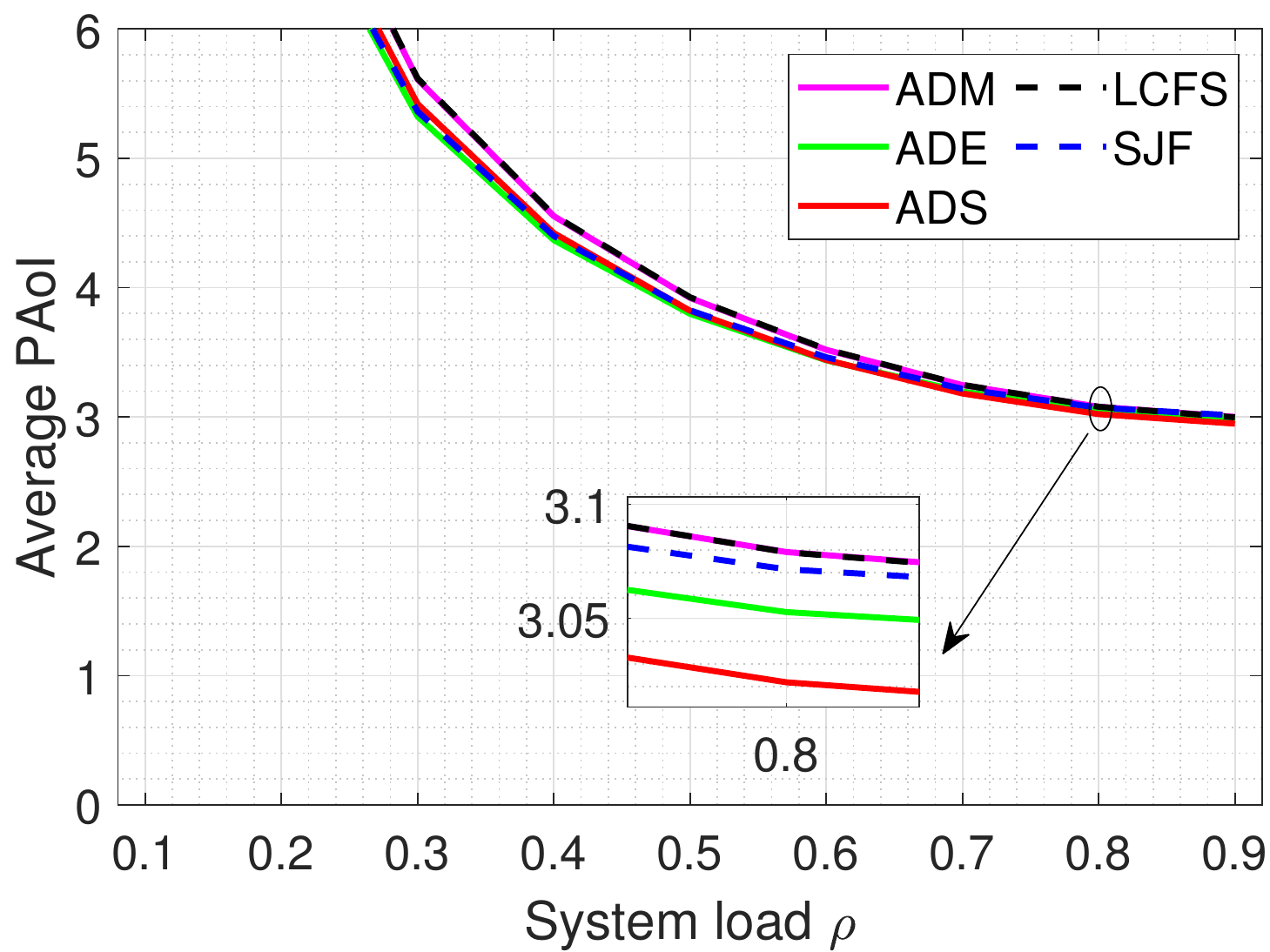}}
	\hspace{0.5em}	
	\subfigure[Interarrival time: Pareto;  
     Update size: Pareto ($\mu=1$)]{
		\label{fig:par-par-paoi-agebased} 
		\includegraphics[width=0.312\textwidth]{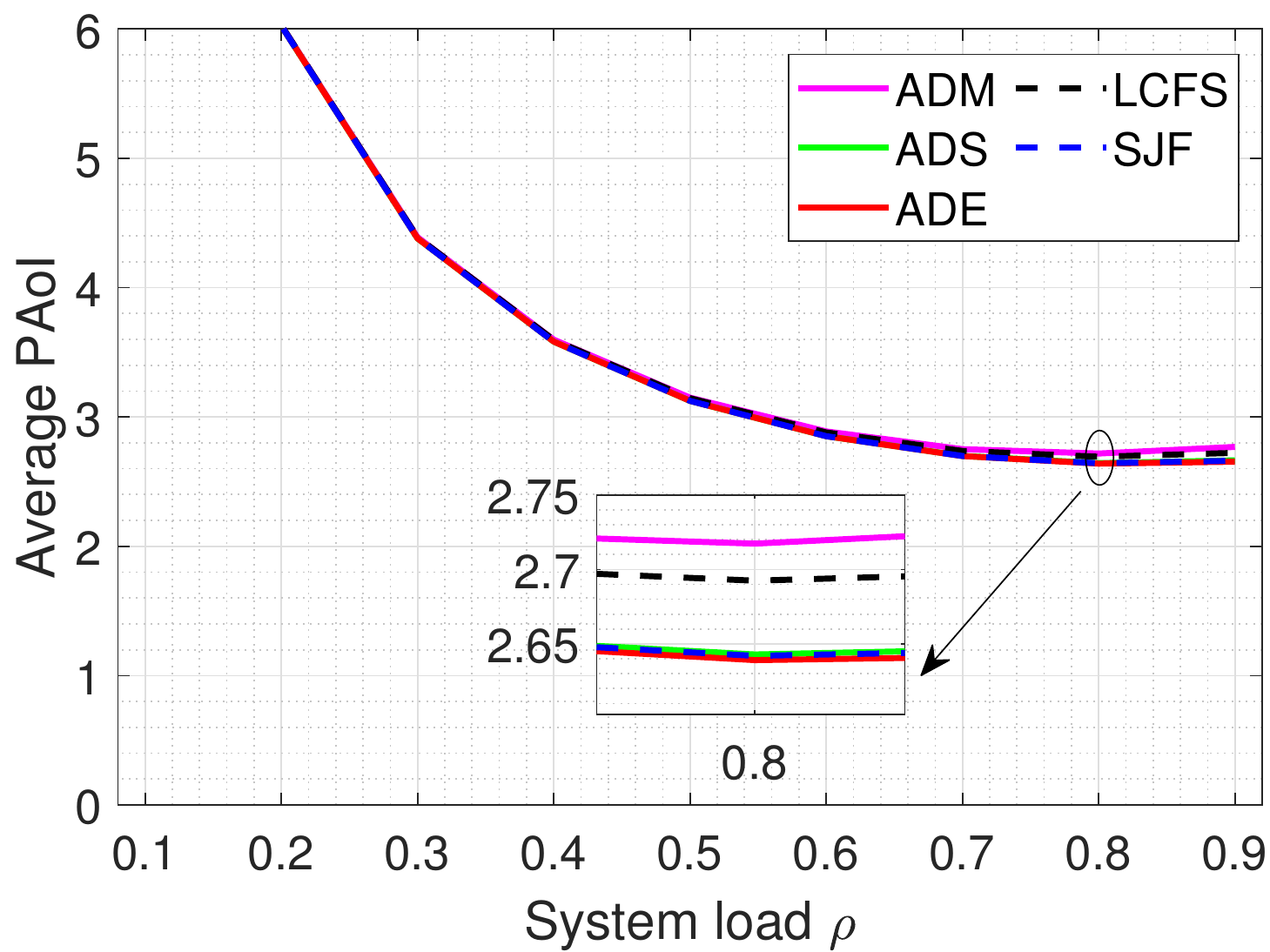}}		
	\caption{Comparisons of the average PAoI performance under different distributions: AoI-based policies vs. non-AoI-based policies}
	\label{fig:aoi-based-PAoI}
	\vspace{-5pt}
\end{figure*}
\begin{figure*}[!t]
    \centering
    \subfigure[Interarrival time: Weibull (${C^{\rm{2}}}{\rm{ = 10}}$); 
    Update size: Exponential ($\mu=1$) ]{
		\label{fig:for-wei-exp-AoI} 
		\includegraphics[width=0.312\textwidth]{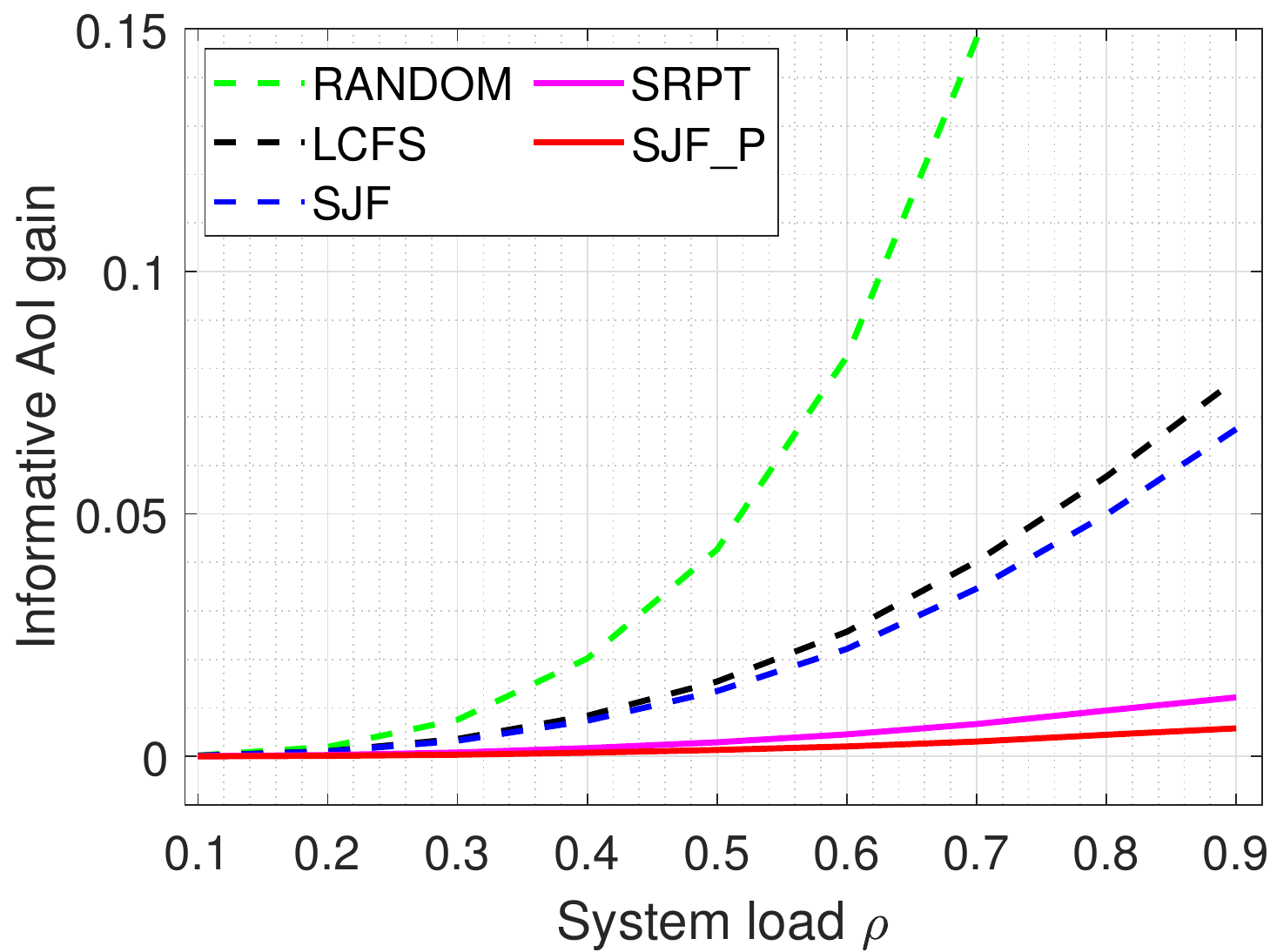}}
	\hspace{0.45em}	
	\subfigure[Interarrival time: Weibull (${C^{\rm{2}}}{\rm{ = 10}}$);  
	Update size: Weibull ($\mu=1$ and ${C^{\rm{2}}}{\rm{ = 10}}$)]{
		\label{fig:infor-wei-wei-AoI} 
		\includegraphics[width=0.312\textwidth]{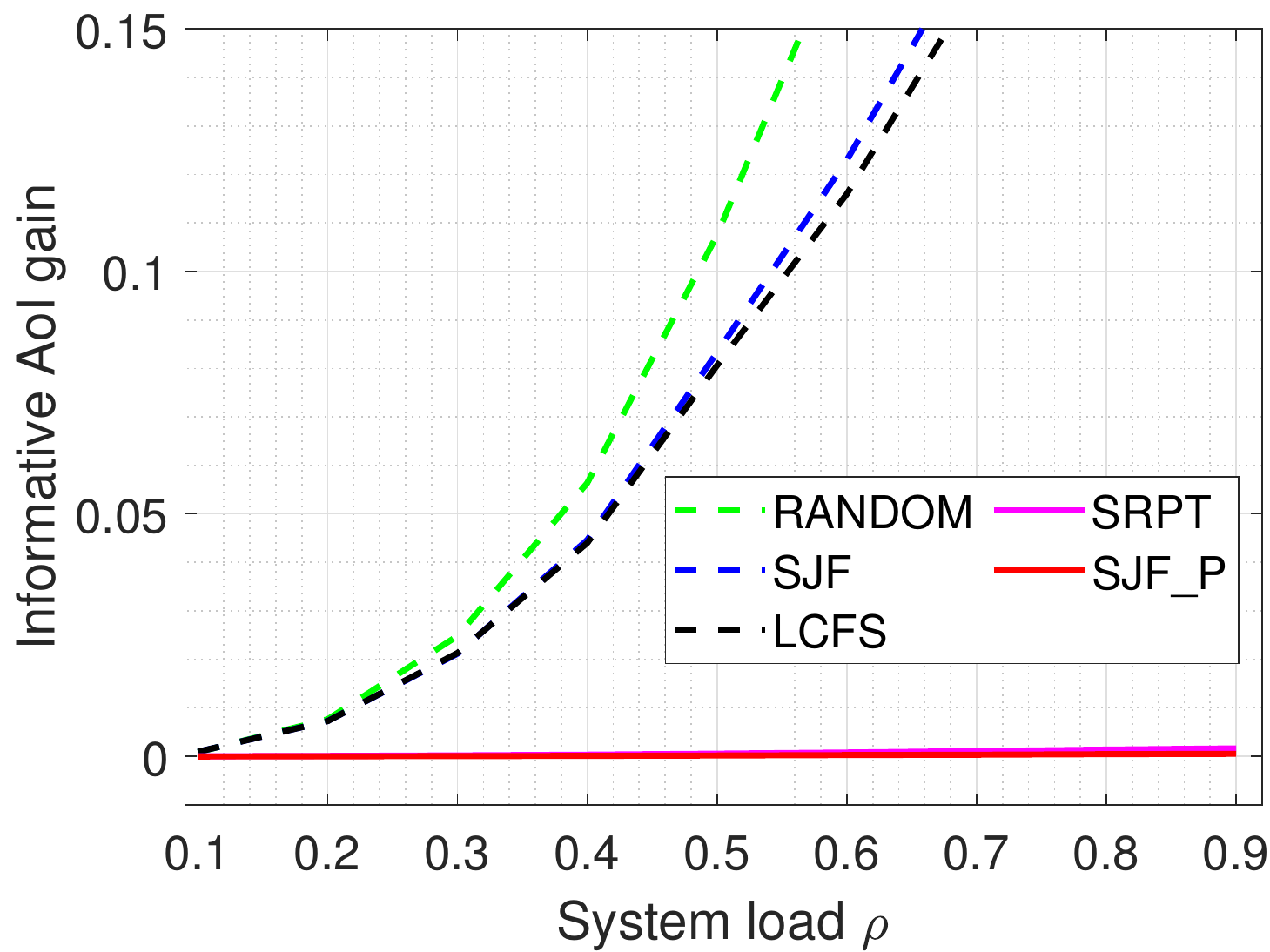}}
	\hspace{0.5em}	
	\subfigure[Interarrival time: Weibull (${C^{\rm{2}}}{\rm{ = 10}}$);  
	Update size: Weibull ($\mu=1$ and $\rho=0.7$)]{
		\label{fig:infor-wei-wei-variance-aoi} 
		\includegraphics[width=0.312\textwidth]{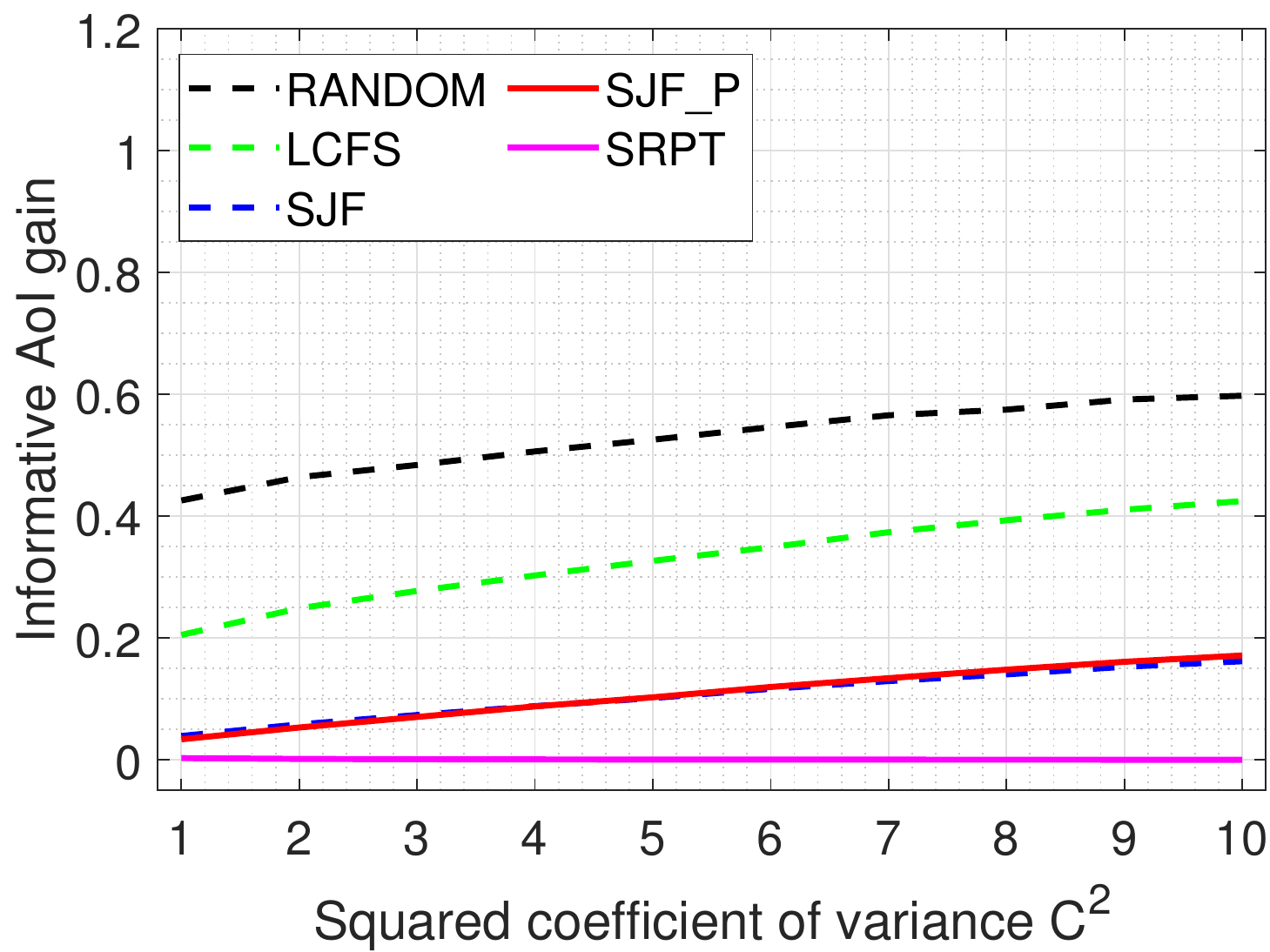}}
	\subfigure[Interarrival time: Gamma;  
     Update size: Gamma ($\mu=1$)]{
		\label{fig:gam-gam-AoI-informative} 
		\includegraphics[width=0.312\textwidth]{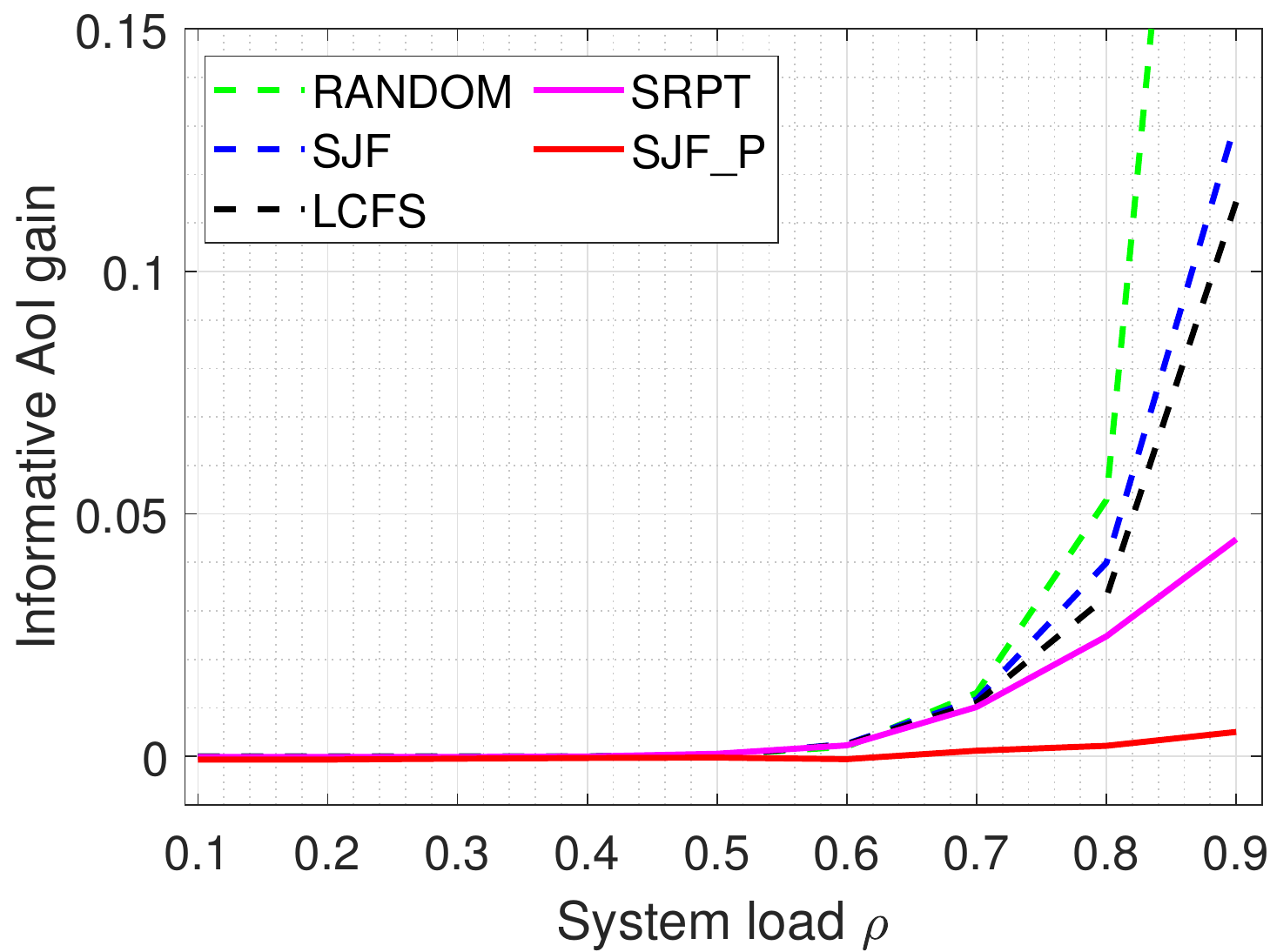}}
	\hspace{0.5em}	
	\subfigure[Interarrival time: Log-normal;  
     Update size: Log-normal ($\mu=1$)]{
		\label{fig:log-log-AoI-informative} 
		\includegraphics[width=0.312\textwidth]{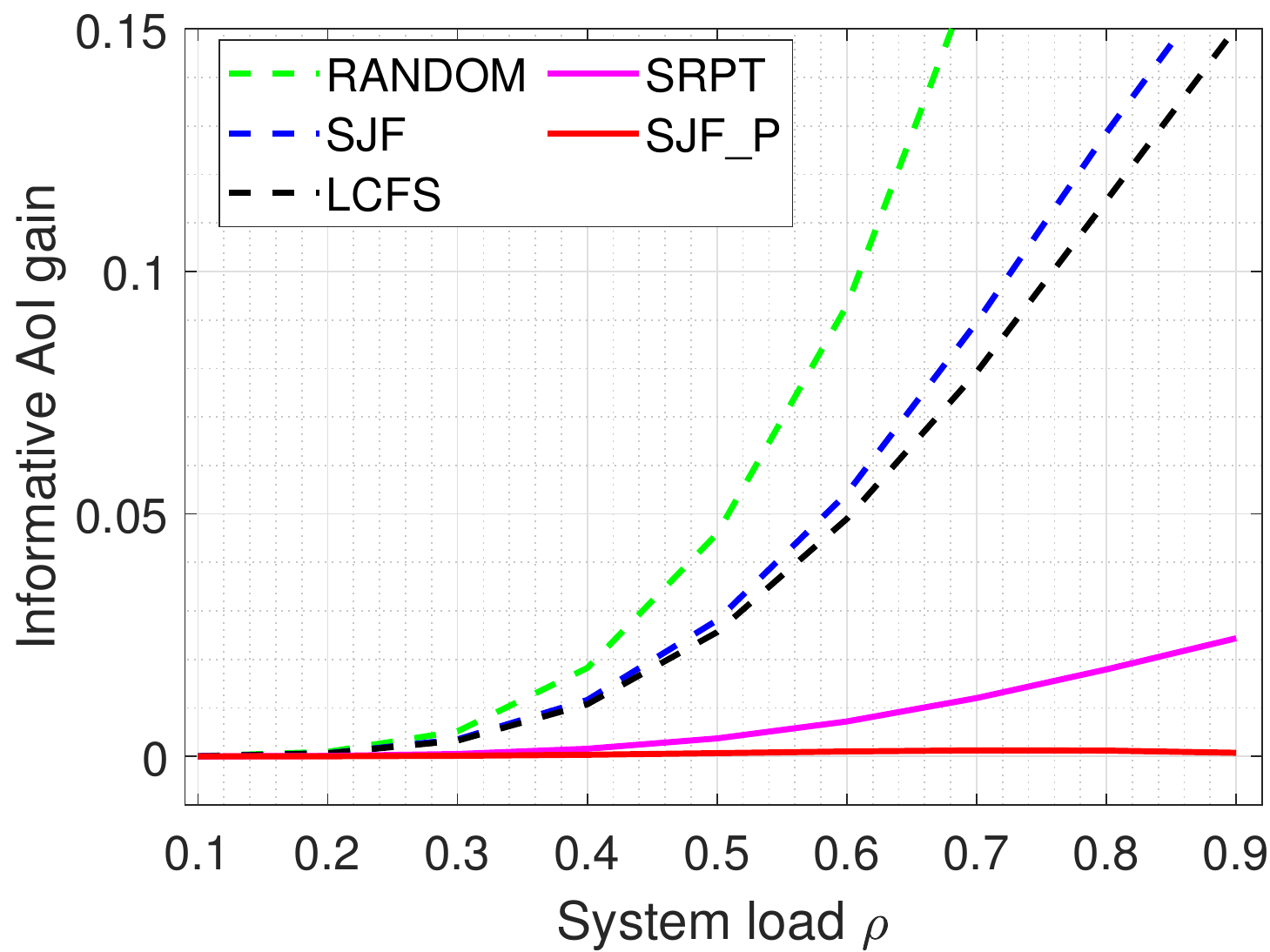}}
	\hspace{0.5em}	
	\subfigure[Interarrival time: Pareto;  
     Update size: Pareto ($\mu=1$)]{
		\label{fig:par-par-aoi-informative} 
		\includegraphics[width=0.312\textwidth]{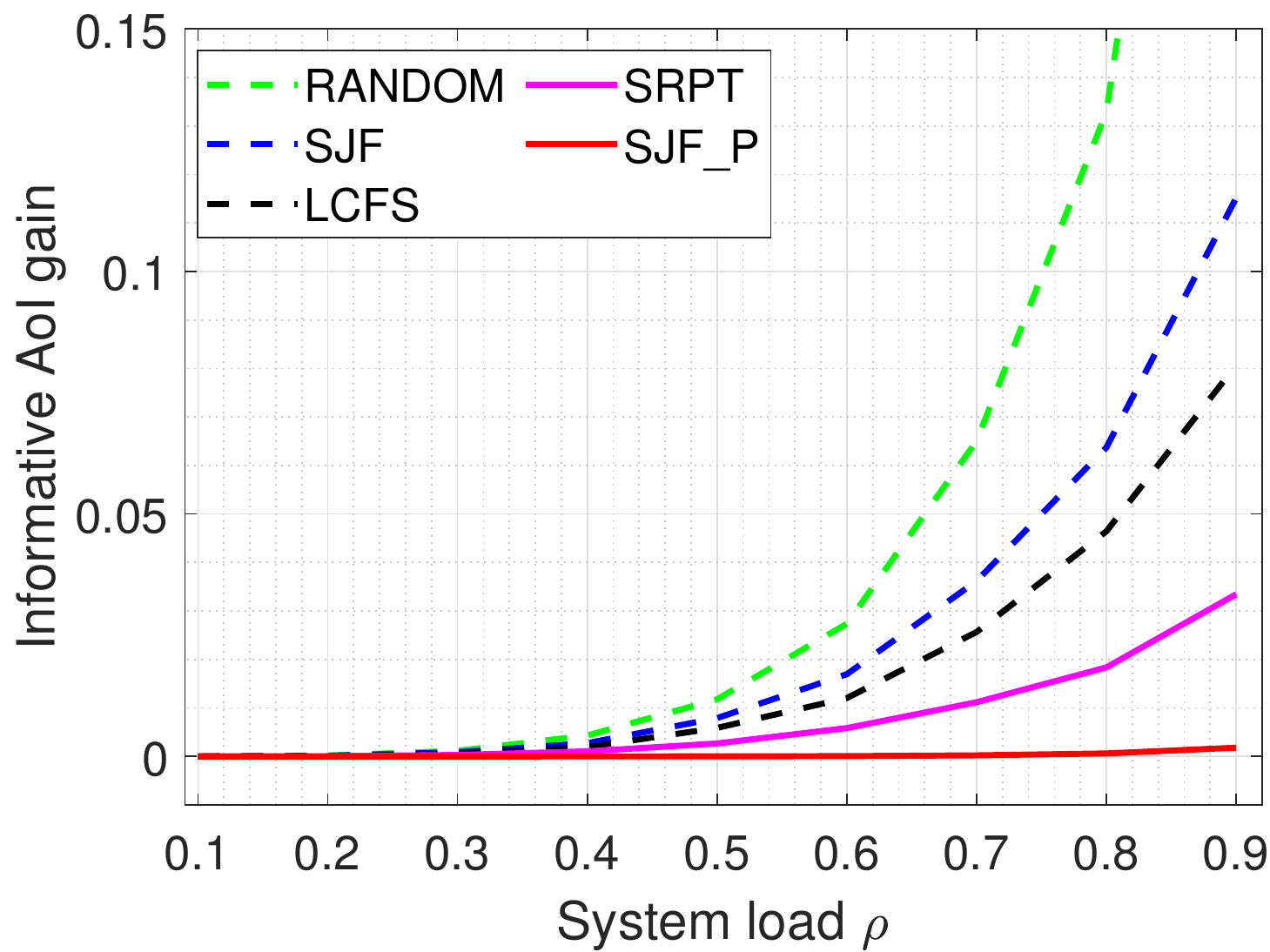}}		
	\caption{Comparisons of the average AoI performance under different distributions: informative policies vs. non-informative policies}
	\label{fig:for-AoI}
	\vspace{-5pt}
\end{figure*}
\begin{figure*}[!t]
    \centering
    \setlength{\abovecaptionskip}{-1pt}
    \subfigcapskip=-1pt
    \subfigure[Interarrival time: Weibull (${C^{\rm{2}}}{\rm{ = 10}}$);  
    Update size: Exponential ($\mu=1$) ]{
		\label{fig:for-wei-exp-PAoI} 
		\includegraphics[width=0.312\textwidth]{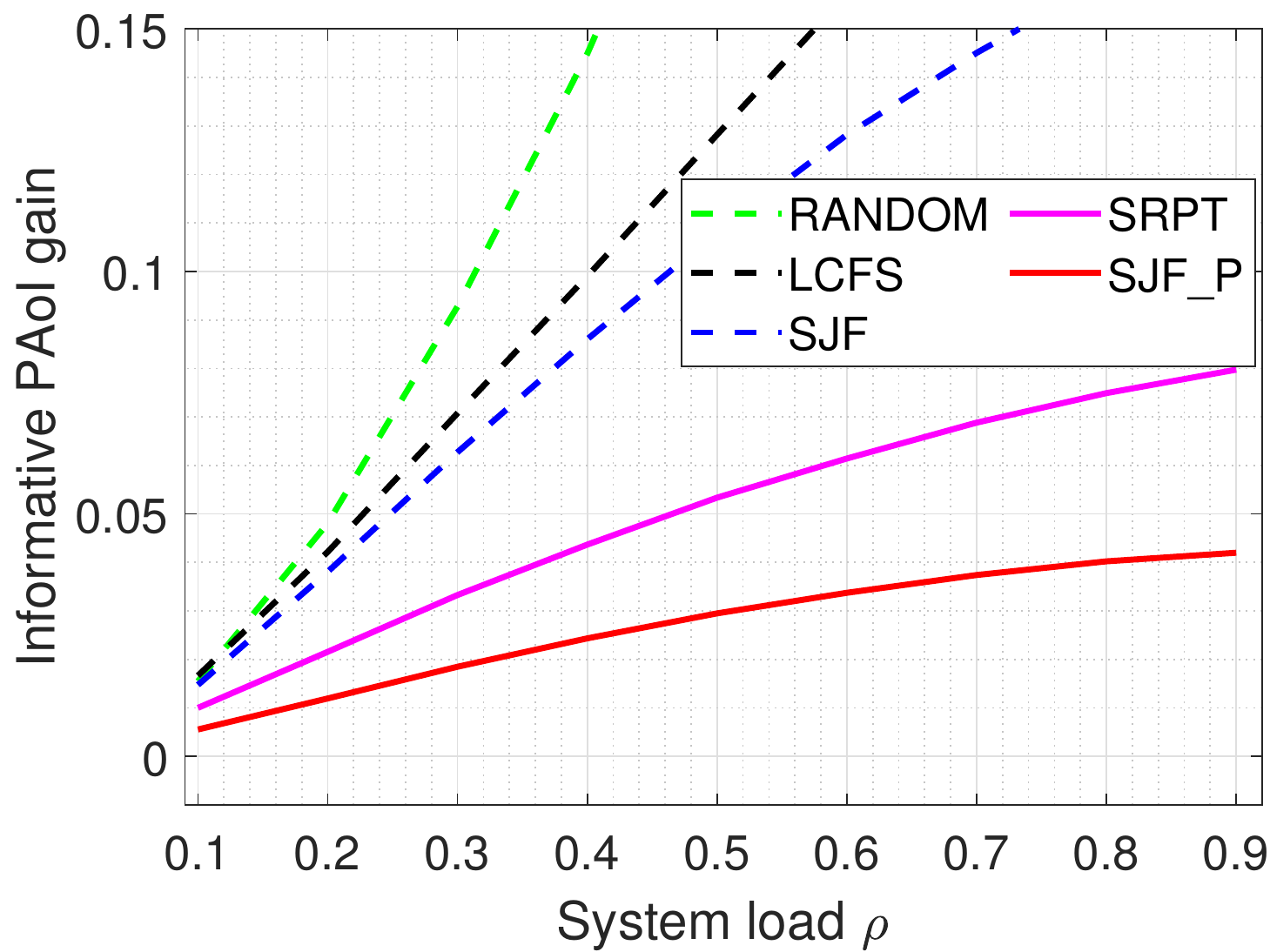}}
	\hspace{0.45em}	
	\subfigure[Interarrival time: Weibull (${C^{\rm{2}}}{\rm{ = 10}}$);  
	Update size: Weibull ($\mu=1$ and ${C^{\rm{2}}}{\rm{ = 10}}$)]{
		\label{fig:infor-wei-wei-PAoI} 
		\includegraphics[width=0.312\textwidth]{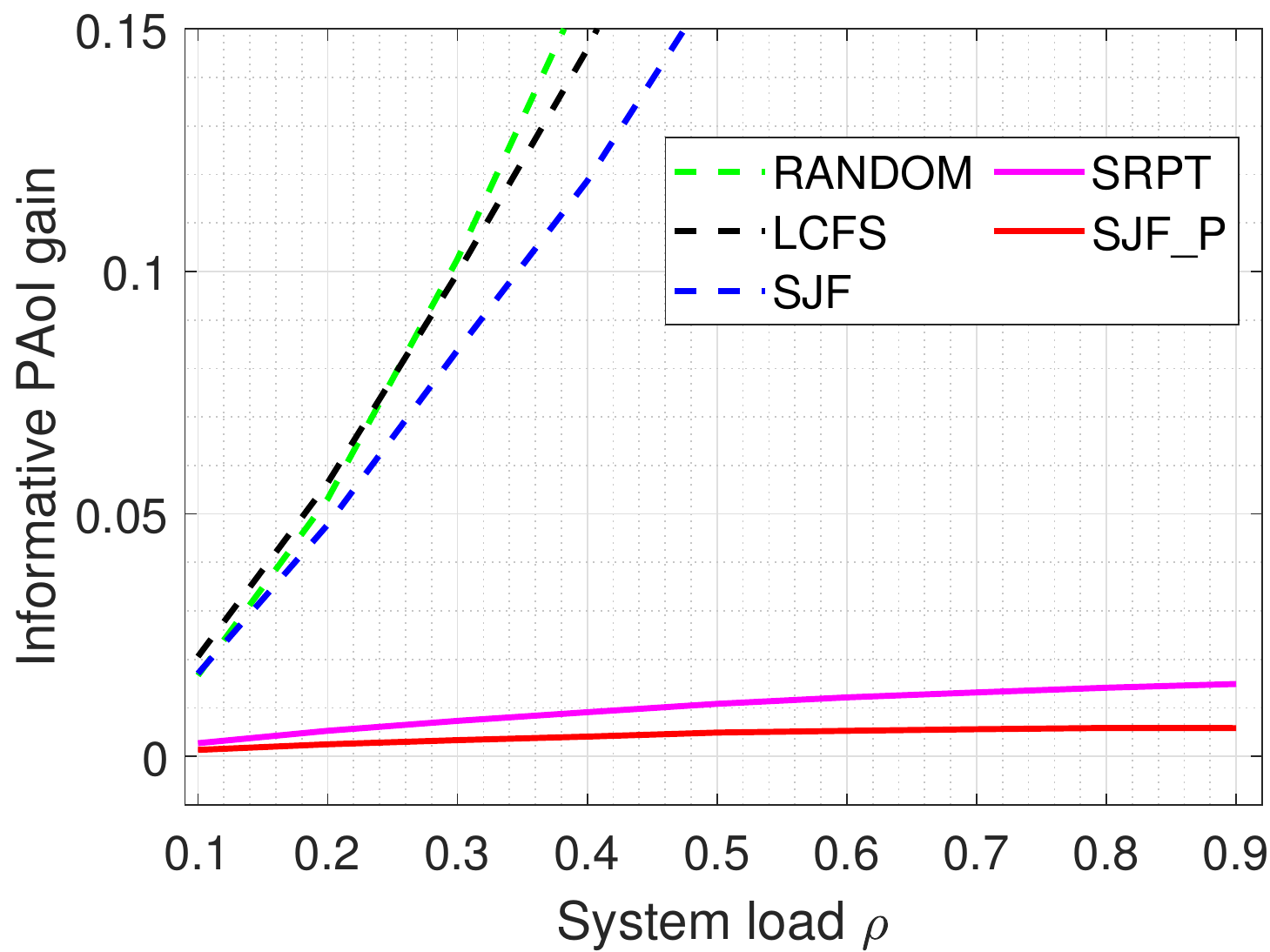}}
	\hspace{0.5em}	
	\subfigure[Interarrival time: Weibull (${C^{\rm{2}}}{\rm{ = 10}}$);  
	Update size: Weibull ($\mu=1$ and $\rho=0.7$)]{
		\label{fig:infor-wei-wei-variance-paoi} 
		\includegraphics[width=0.312\textwidth]{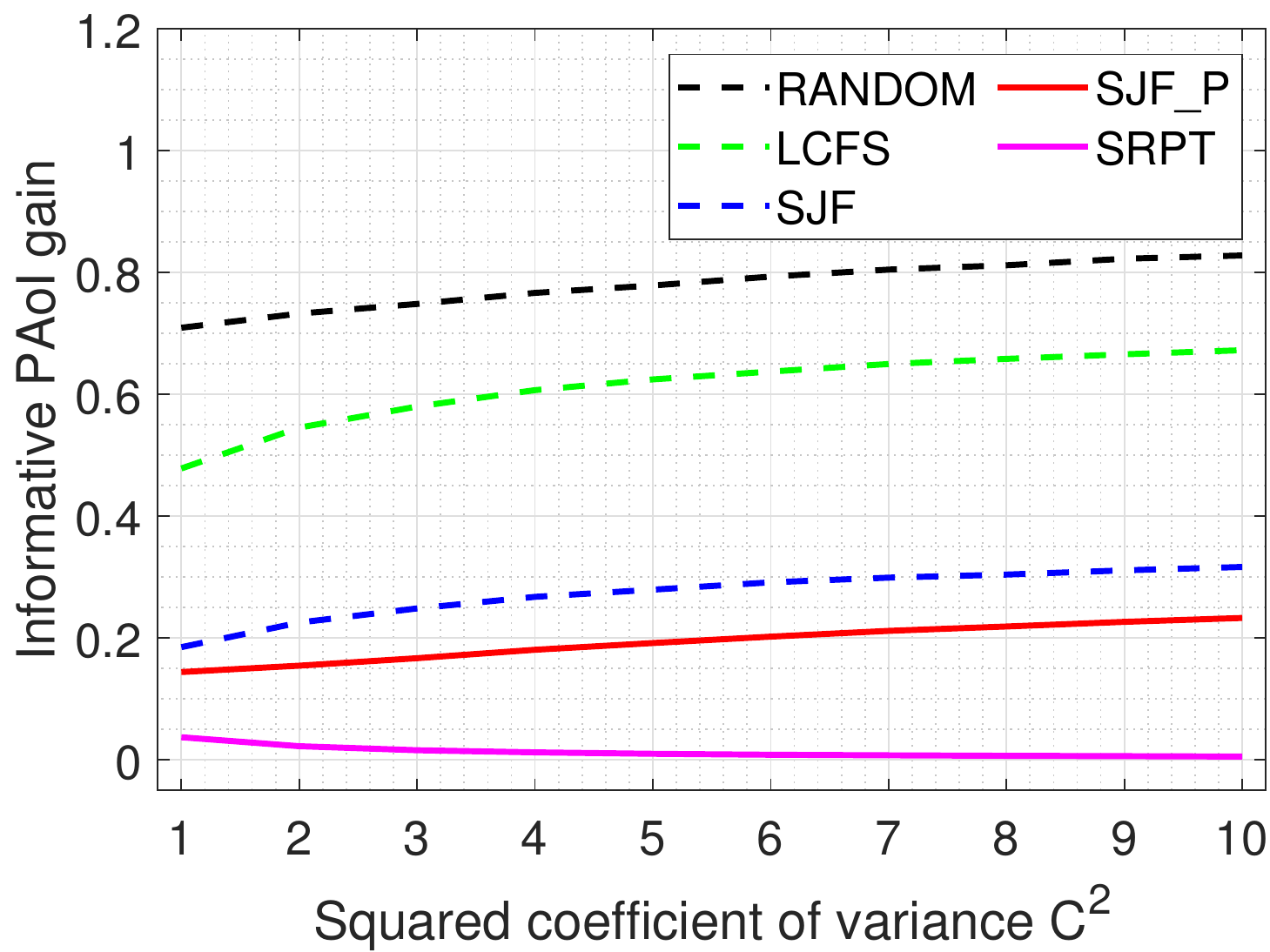}}	
	\subfigure[Interarrival time: Gamma;  
     Update size: Gamma ($\mu=1$)]{
		\label{fig:gam-gam-PAoI-informative} 
		\includegraphics[width=0.312\textwidth]{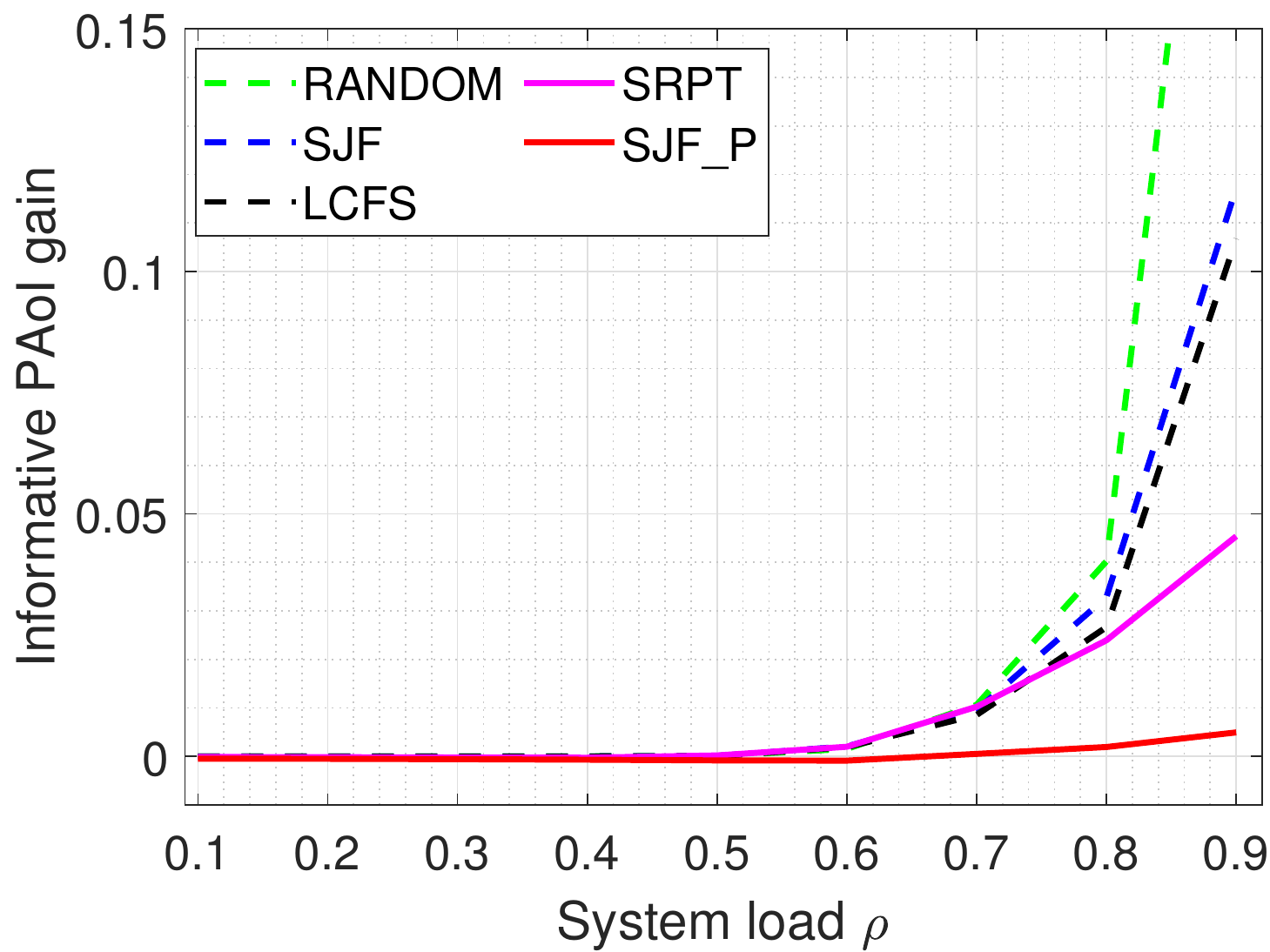}}	
	\hspace{0.5em}	
	\subfigure[Interarrival time: Log-normal;  
     Update size: Log-normal ($\mu=1$)]{
		\label{fig:log-log-PAoI-informative} 
		\includegraphics[width=0.312\textwidth]{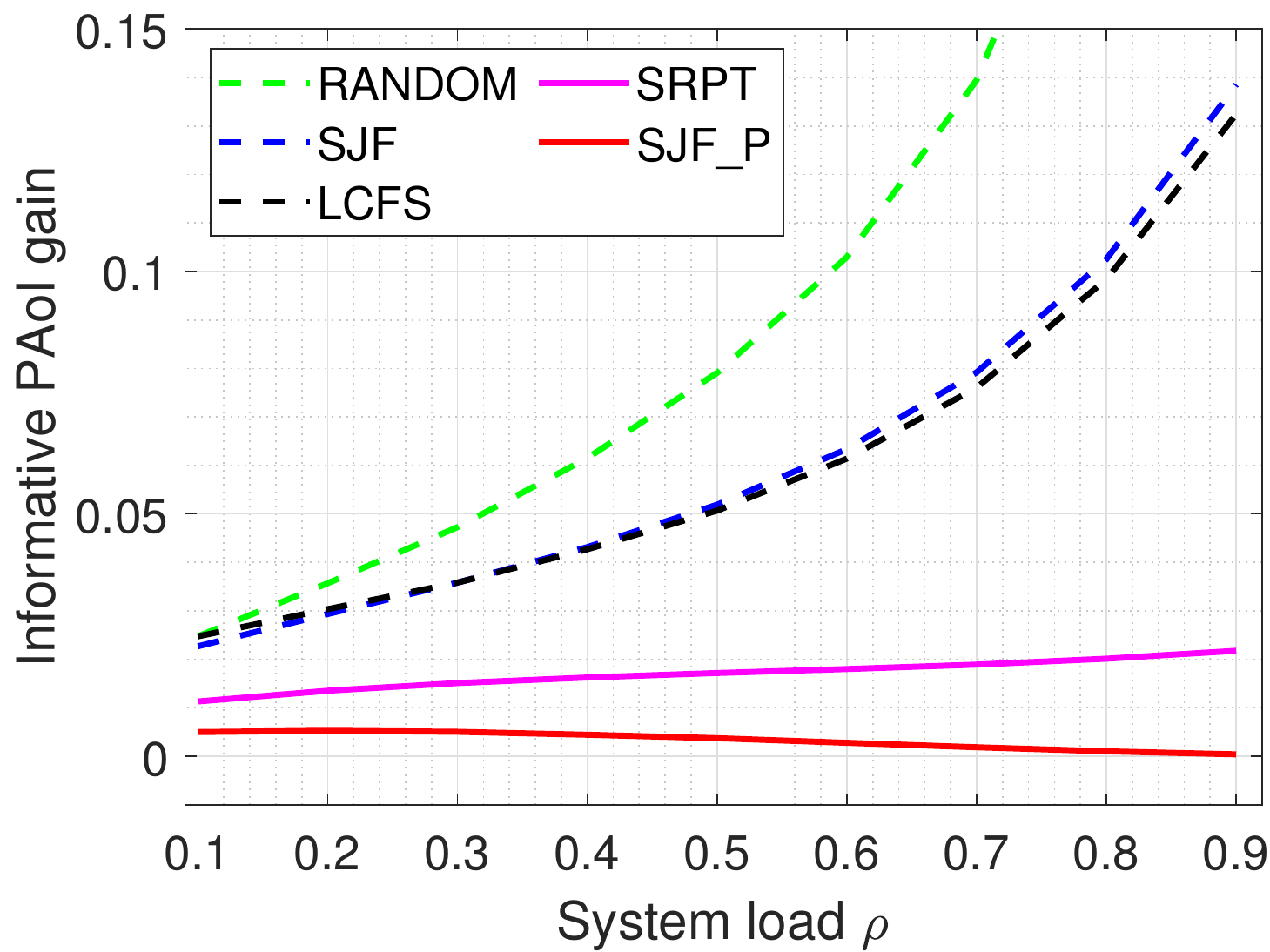}}	
	\hspace{0.5em}	
	\subfigure[Interarrival time: Pareto;  
     Update size: Pareto ($\mu=1$)]{
		\label{fig:par-par-paoi-informative} 
		\includegraphics[width=0.312\textwidth]{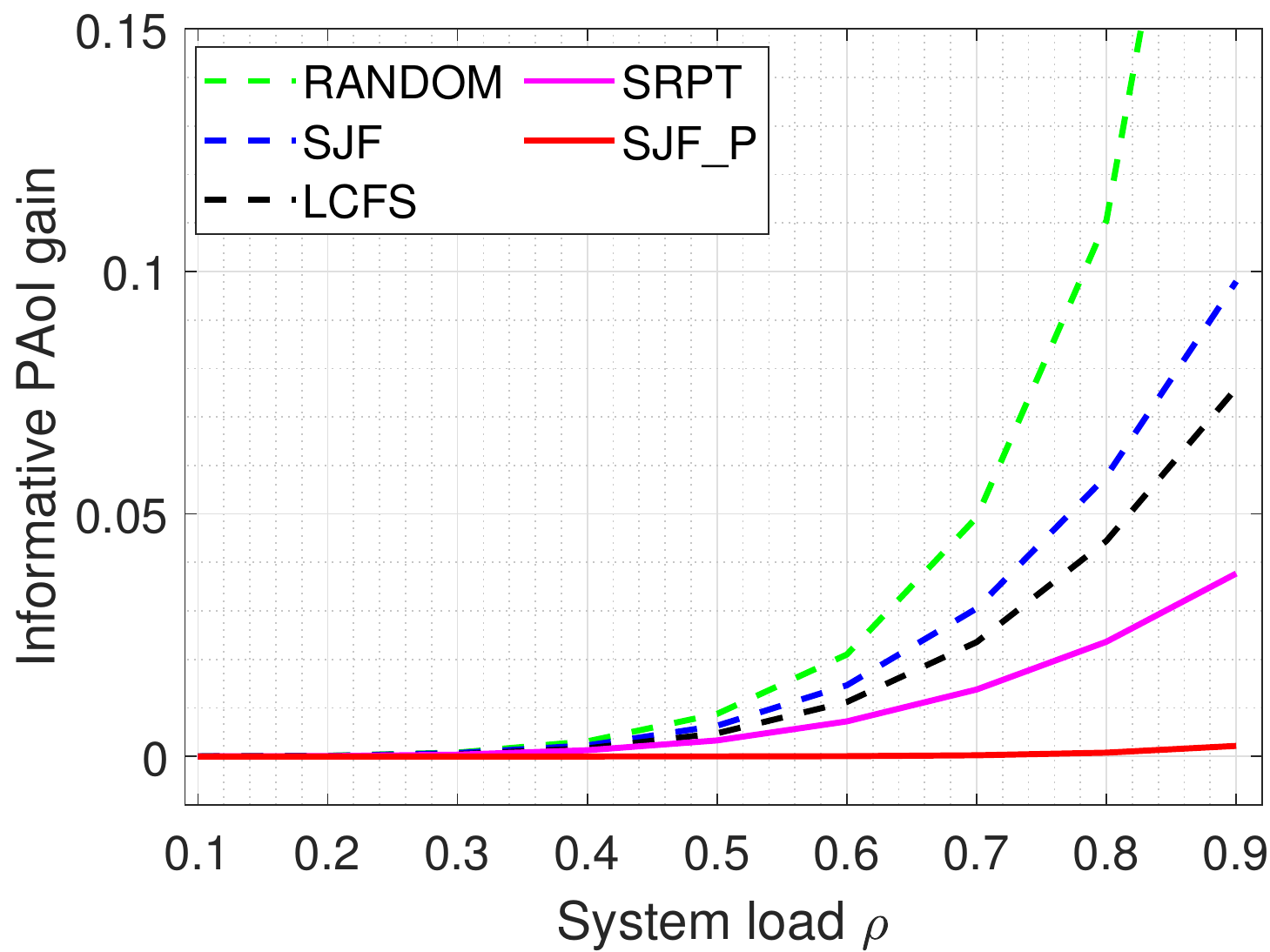}}	
	\caption{Comparisons of the average PAoI performance under different distributions: informative policies vs. non-informative policies}
	\label{fig:for-PAoI}
	\vspace{-5pt}
\end{figure*}

\newpage  

\begin{figure*}[!t]
    \centering
    \setlength{\abovecaptionskip}{-1pt}
    \subfigcapskip=-1pt
    \subfigure[Interarrival time: Weibull (${C^{\rm{2}}}{\rm{ = 10}}$);  
    Update size: Exponential ($\mu=1$)]{
		\label{fig:all-combined-wei-exp-aoi}
		\includegraphics[width=0.312\textwidth]{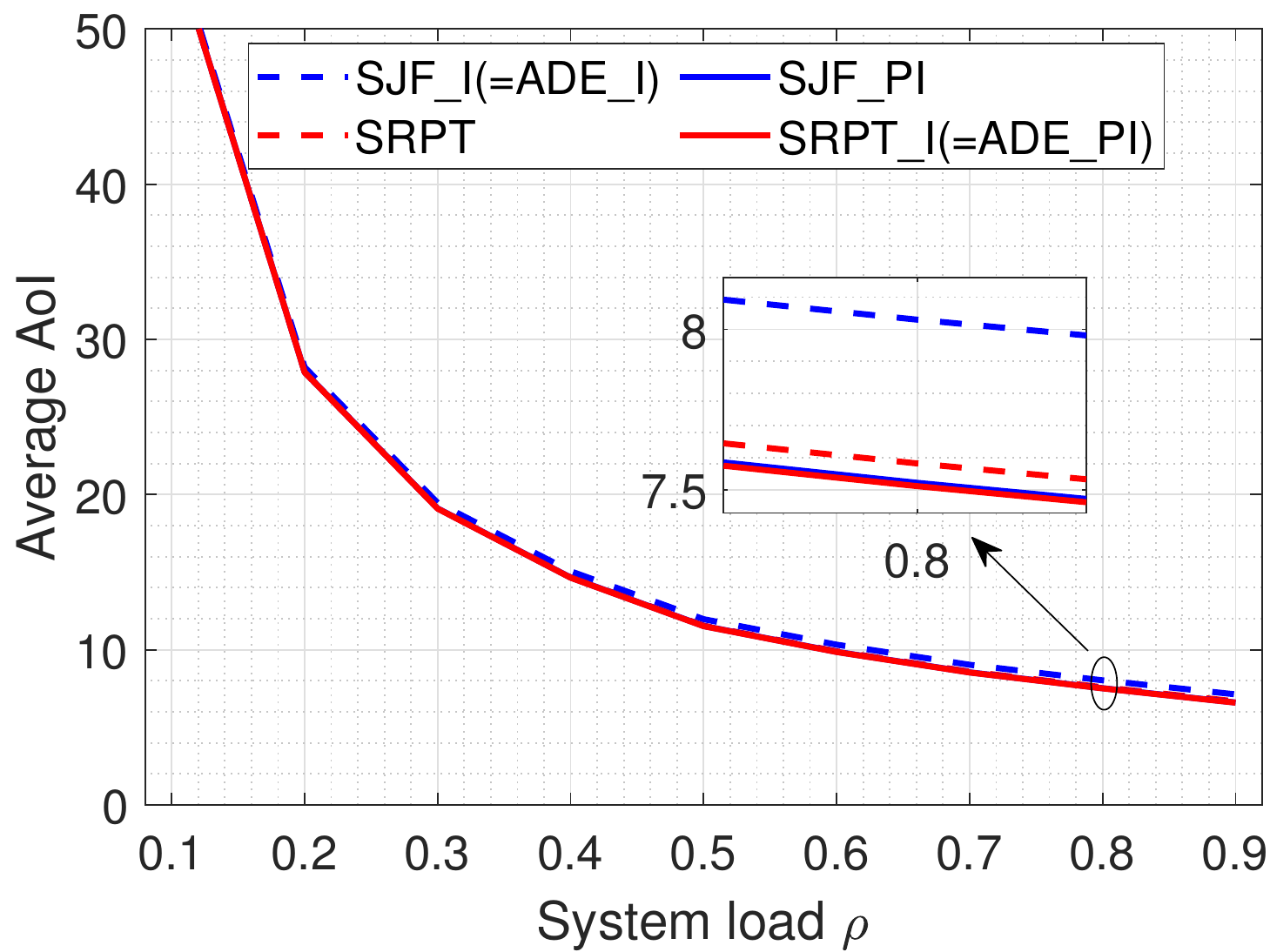}}
	\hspace{0.47em}	
    \subfigure[Interarrival time: Weibull (${C^{\rm{2}}}{\rm{ = 10}}$);  
    Update size: Weibull ($\mu=1$ and ${C^{\rm{2}}}{\rm{ = 10}}$)]{
    \label{fig:all-combined-wei-wei-aoi}
		\includegraphics[width=0.312\textwidth]{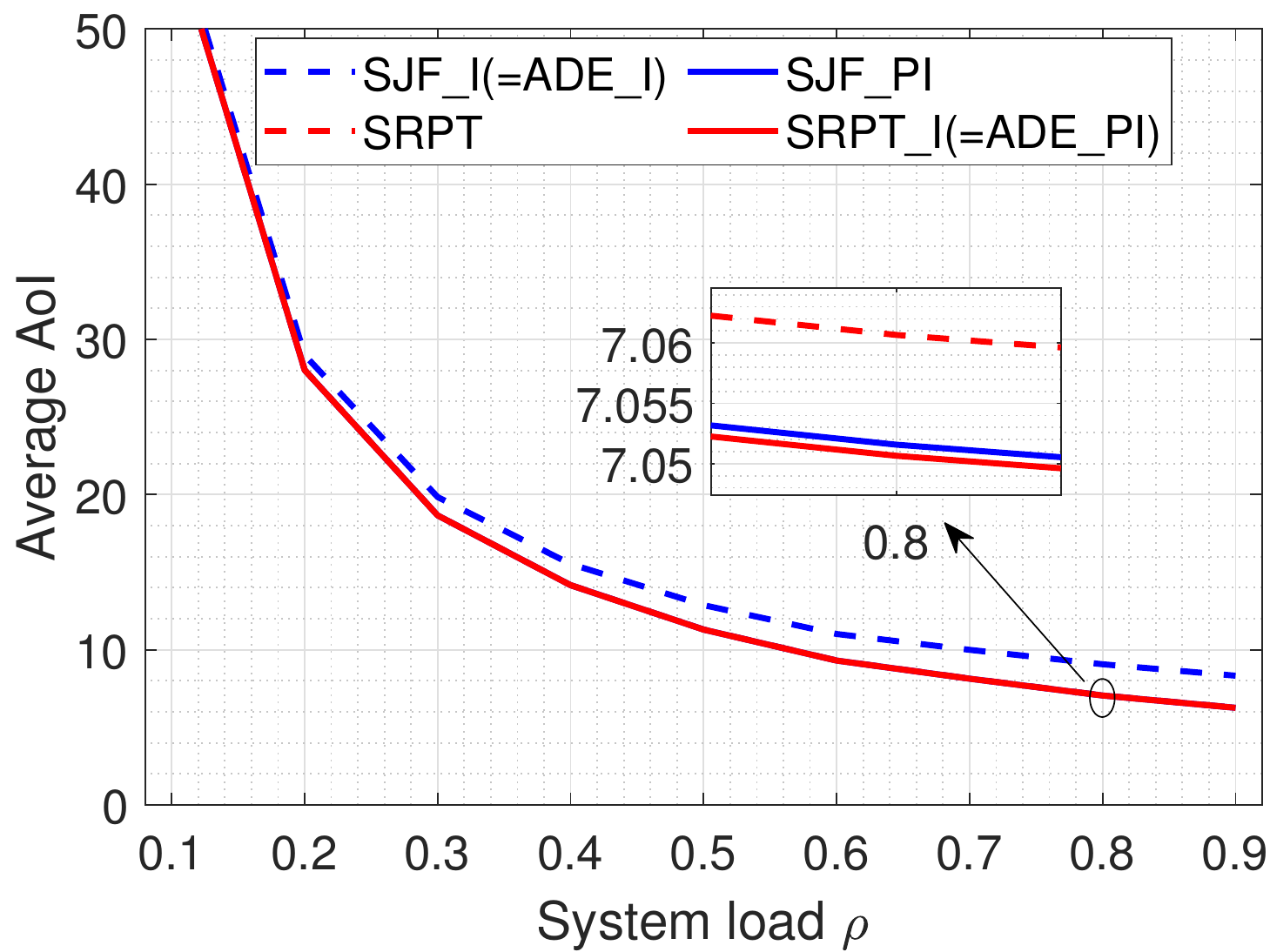}}
	\hspace{0.5em}	
	\subfigure[Interarrival time: Weibull (${C^{\rm{2}}}{\rm{ = 10}}$);  
	Update size: Weibull ($\mu=1$ and $\rho=0.7$)]{
		\label{fig:all-combined-wei-variance-aoi} 
		\includegraphics[width=0.312\textwidth]{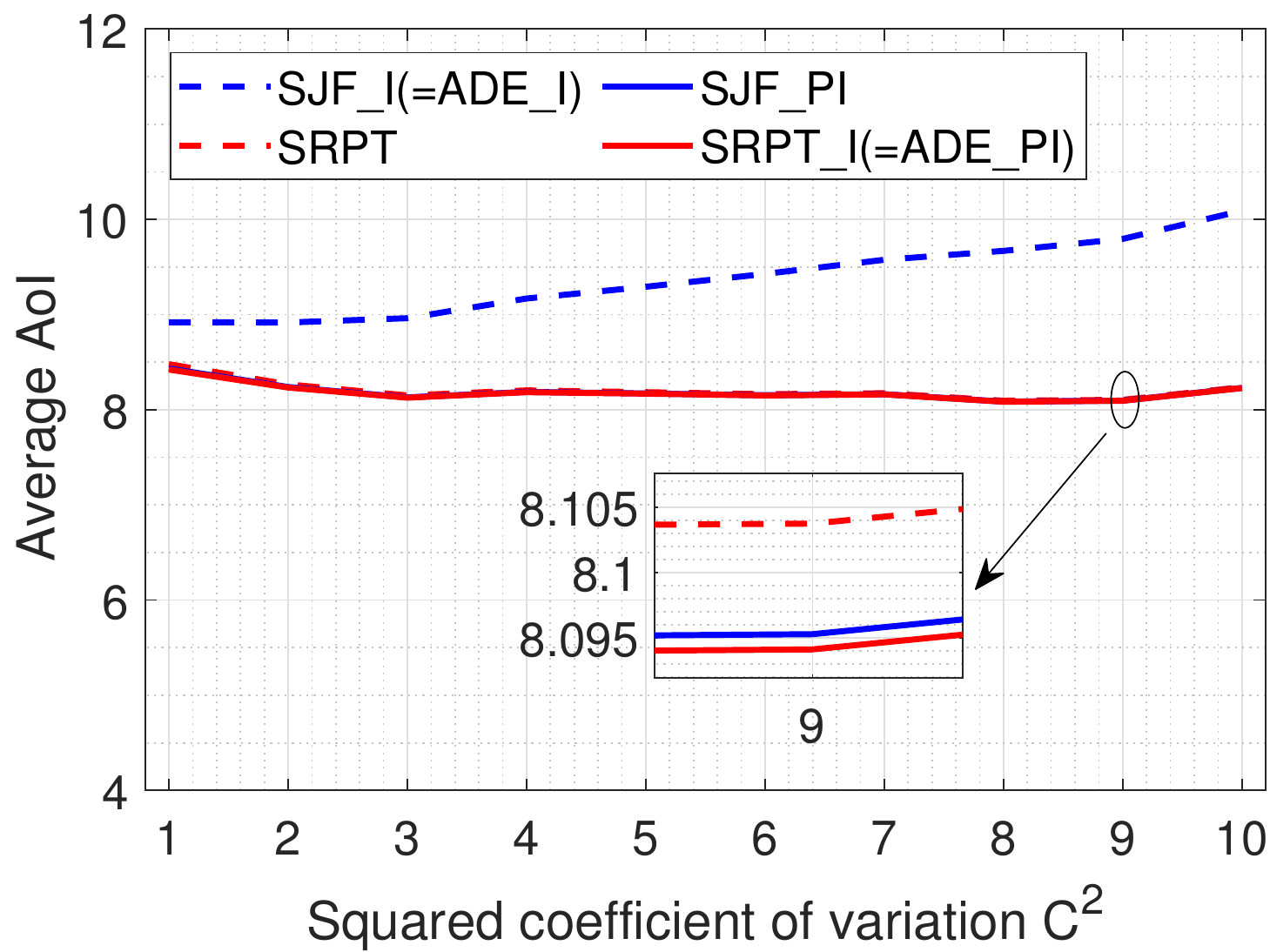}}	
	\subfigure[Interarrival time: Gamma;  
     Update size: Gamma ($\mu=1$)]{
		\label{fig:gam-gam-AoI-allAspects} 
		\includegraphics[width=0.312\textwidth]{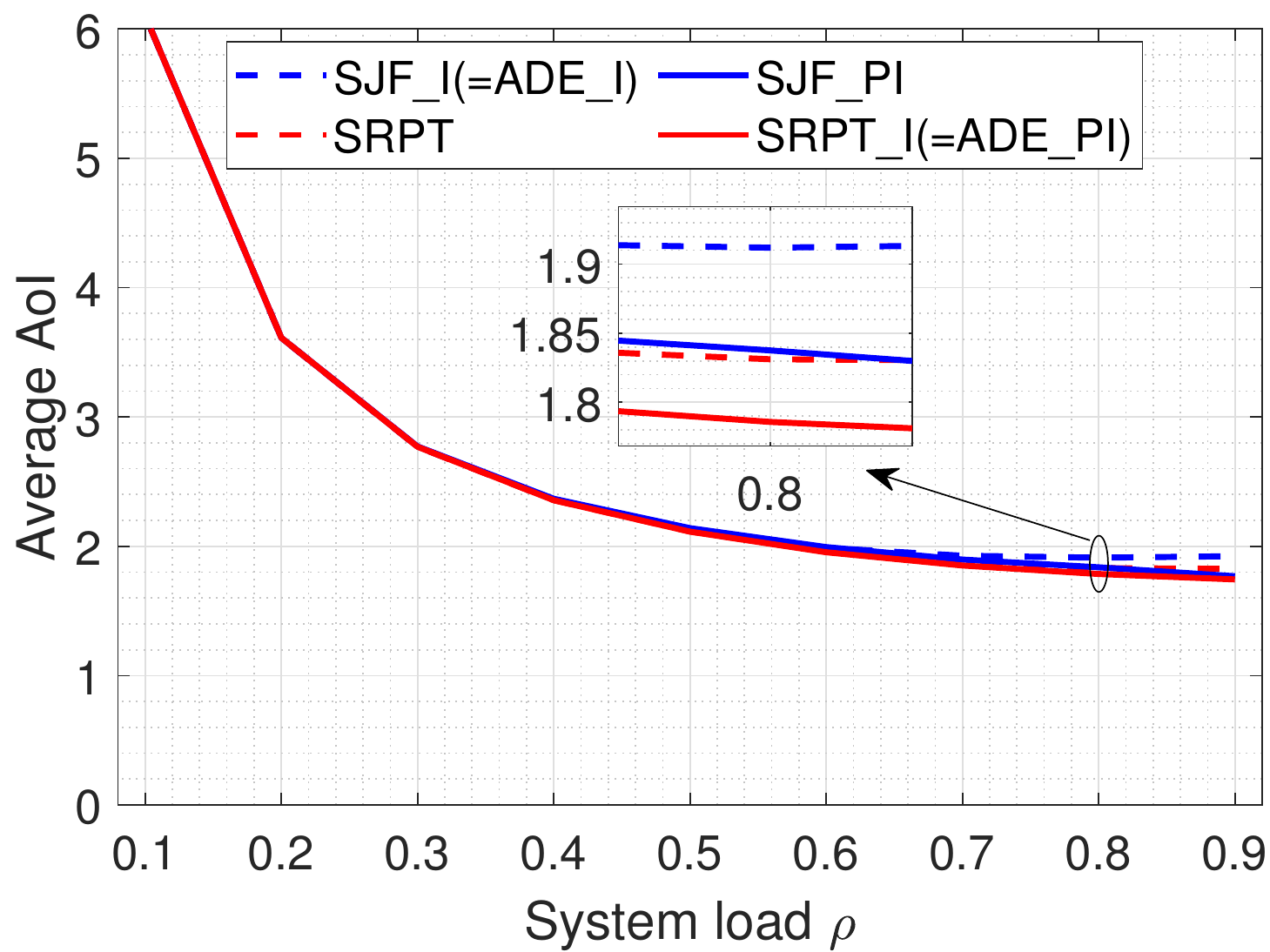}}	
	\hspace{0.5em}	
	\subfigure[Interarrival time: Log-normal;  
     Update size: Log-normal ($\mu=1$)]{
		\label{fig:log-log-AoI-allAspects} 
		\includegraphics[width=0.312\textwidth]{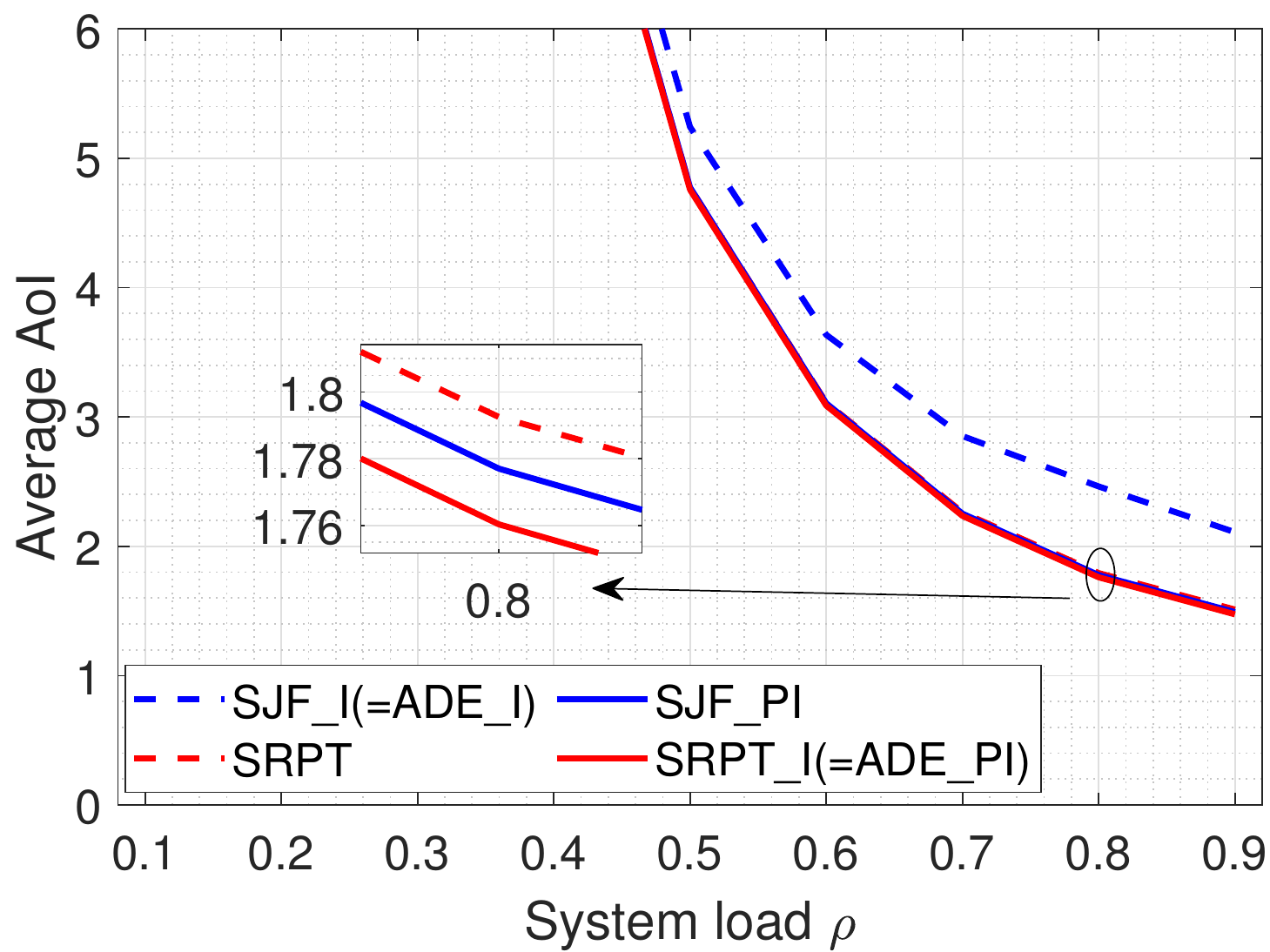}}
	\hspace{0.5em}	
	\subfigure[Interarrival time: Pareto;  
     Update size: Pareto ($\mu=1$)]{
		\label{fig:par-par-aoi-allAspects} 
		\includegraphics[width=0.312\textwidth]{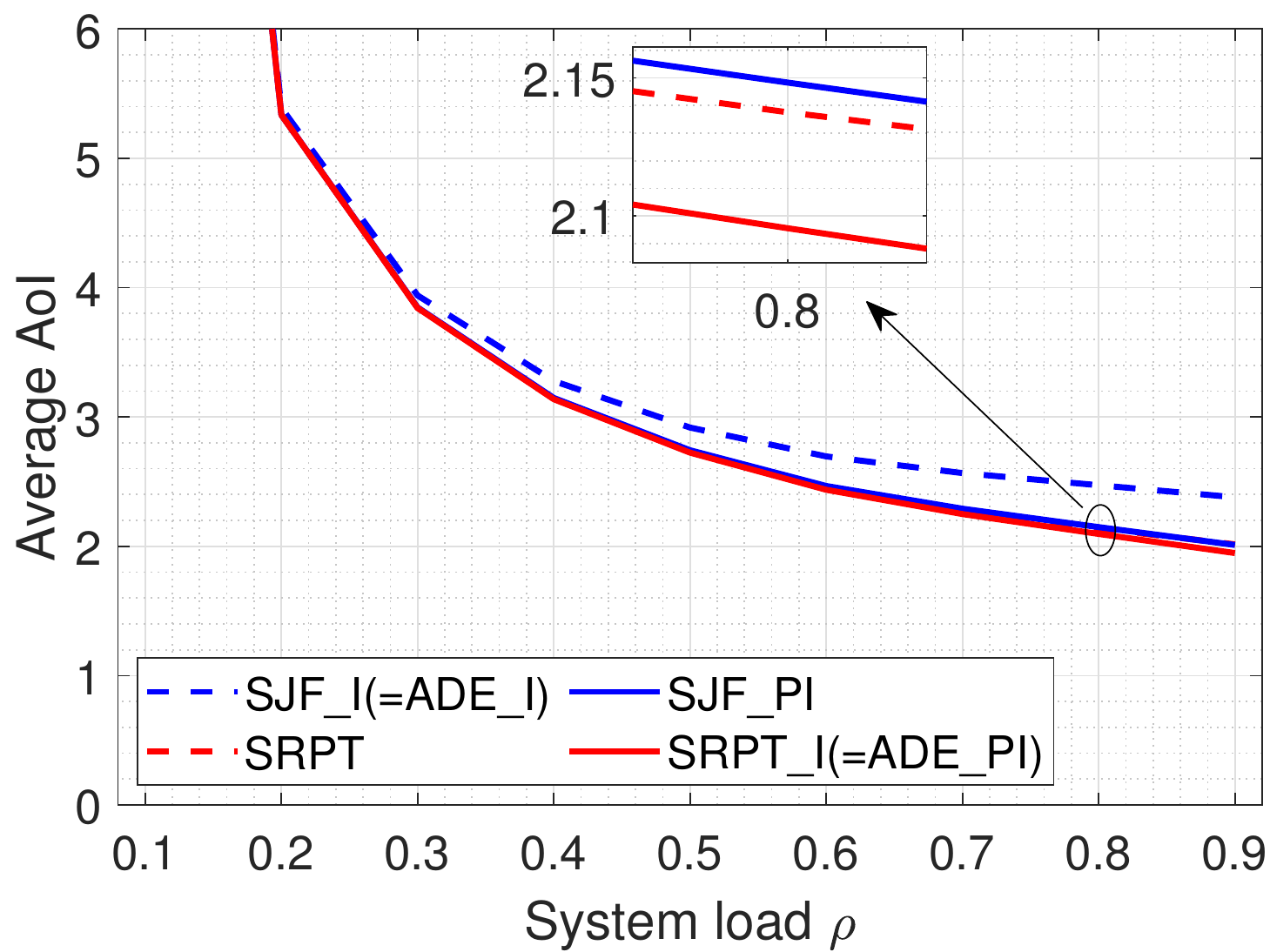}}		
	\caption{Comparisons of the average AoI performance under different distributions: preemptive, informative, AoI-based policies vs. others}
	\label{fig:all-combined-aoi}
\end{figure*}
\begin{figure*}[!t]
    \centering
    \setlength{\abovecaptionskip}{-1pt}
    \subfigcapskip=-1pt
    \subfigure[Interarrival time: Weibull (${C^{\rm{2}}}{\rm{ = 10}}$);  
    Update size: Exponential ($\mu=1$)]{
		\label{fig:all-combined-wei-exp-paoi}
		\includegraphics[width=0.312\textwidth]{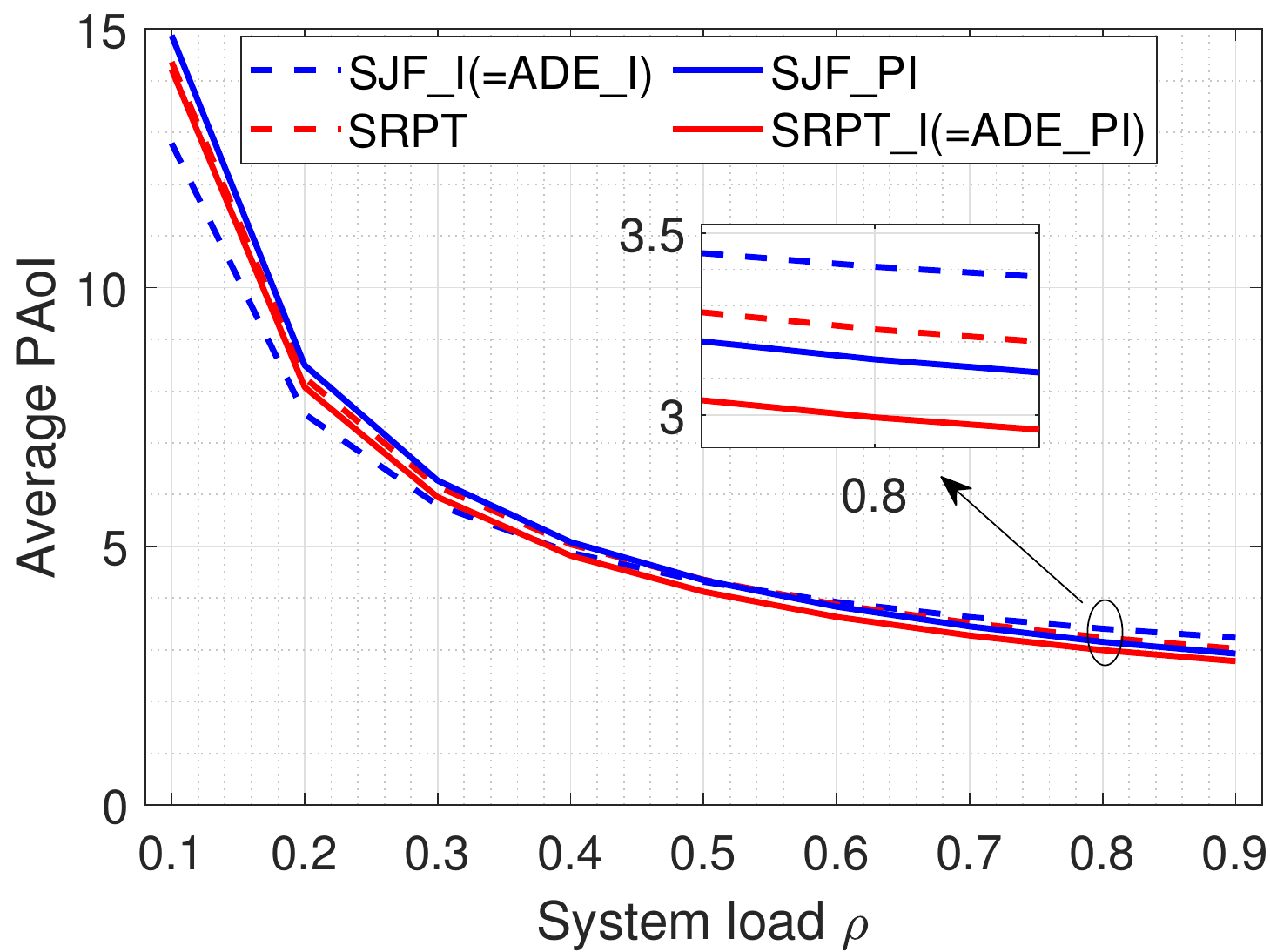}}
	\hspace{0.47em}	
    \subfigure[Interarrival time: Weibull (${C^{\rm{2}}}{\rm{ = 10}}$);  
    Update size: Weibull ($\mu=1$ and ${C^{\rm{2}}}{\rm{ = 10}}$)]{
    \label{fig:all-combined-wei-wei-paoi}
		\includegraphics[width=0.312\textwidth]{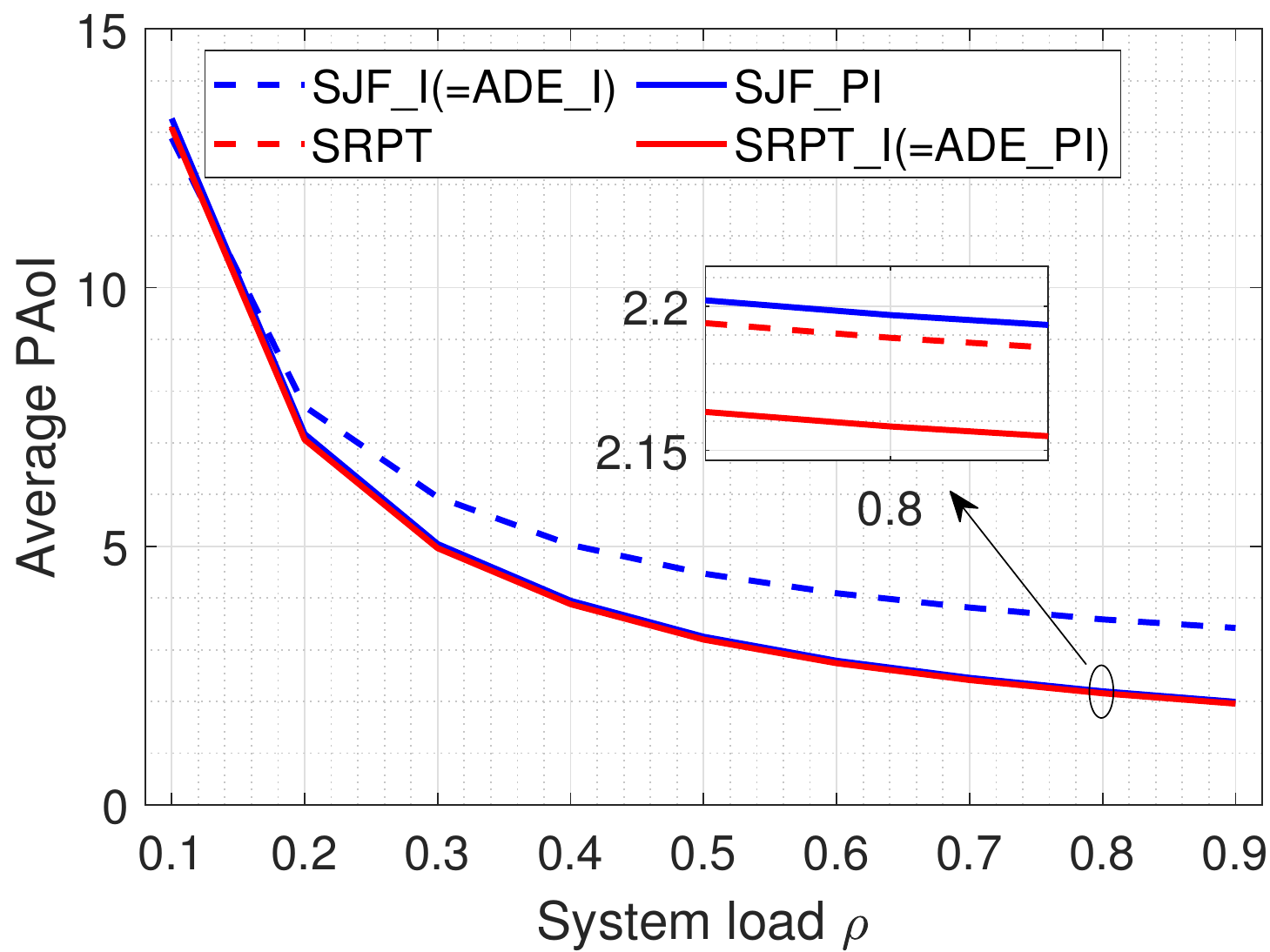}}
	\hspace{0.5em}	
	\subfigure[Interarrival time: Weibull (${C^{\rm{2}}}{\rm{ = 10}}$);  
	Update size: Weibull ($\mu=1$ and $\rho=0.7$)]{
		\label{fig:all-combined-wei-variance-paoi} 
		\includegraphics[width=0.312\textwidth]{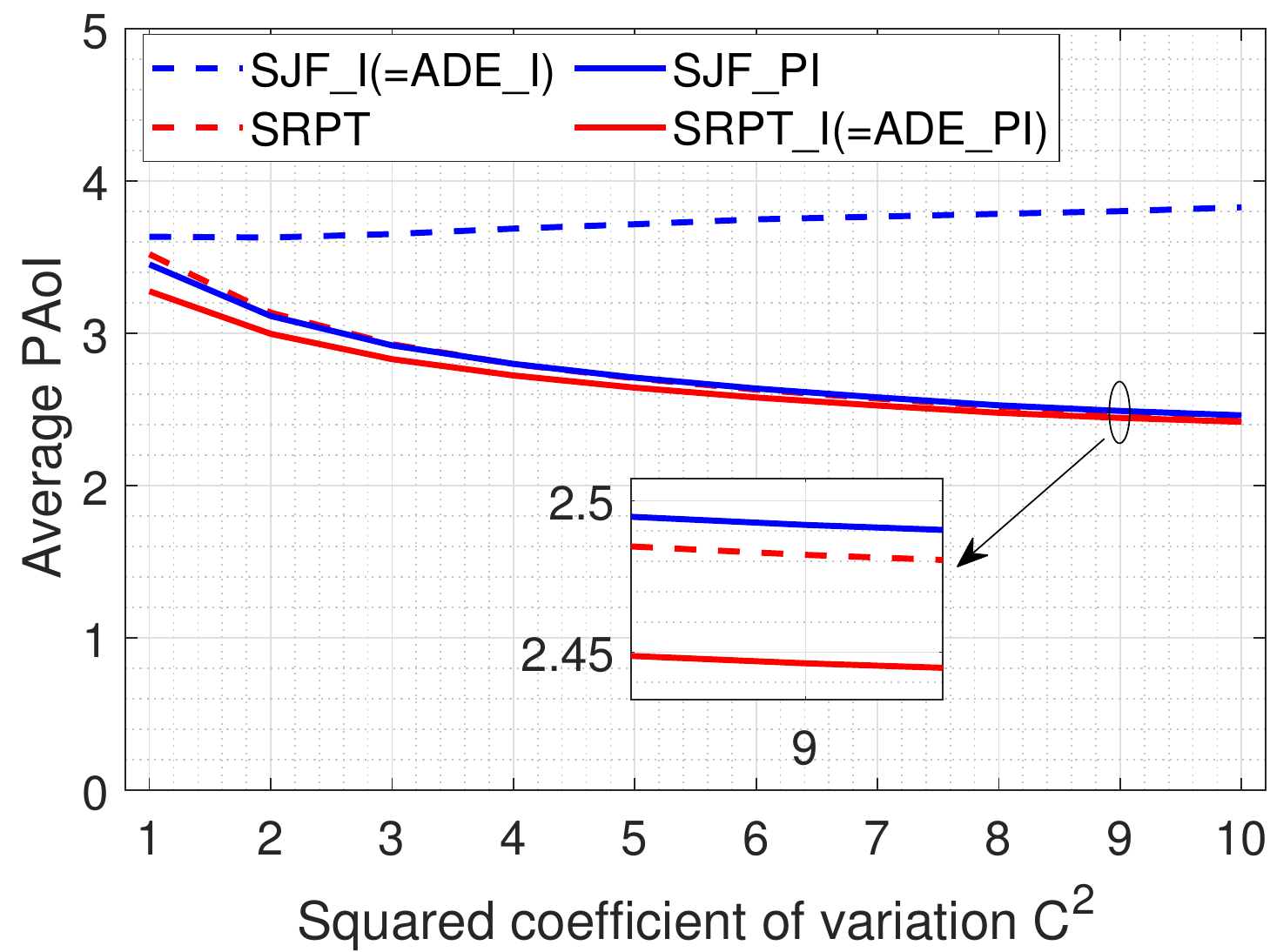}}
	\subfigure[Interarrival time: Gamma;  
     Update size: Gamma ($\mu=1$)]{
		\label{fig:gam-gam-PAoI-allAspects} 
		\includegraphics[width=0.312\textwidth]{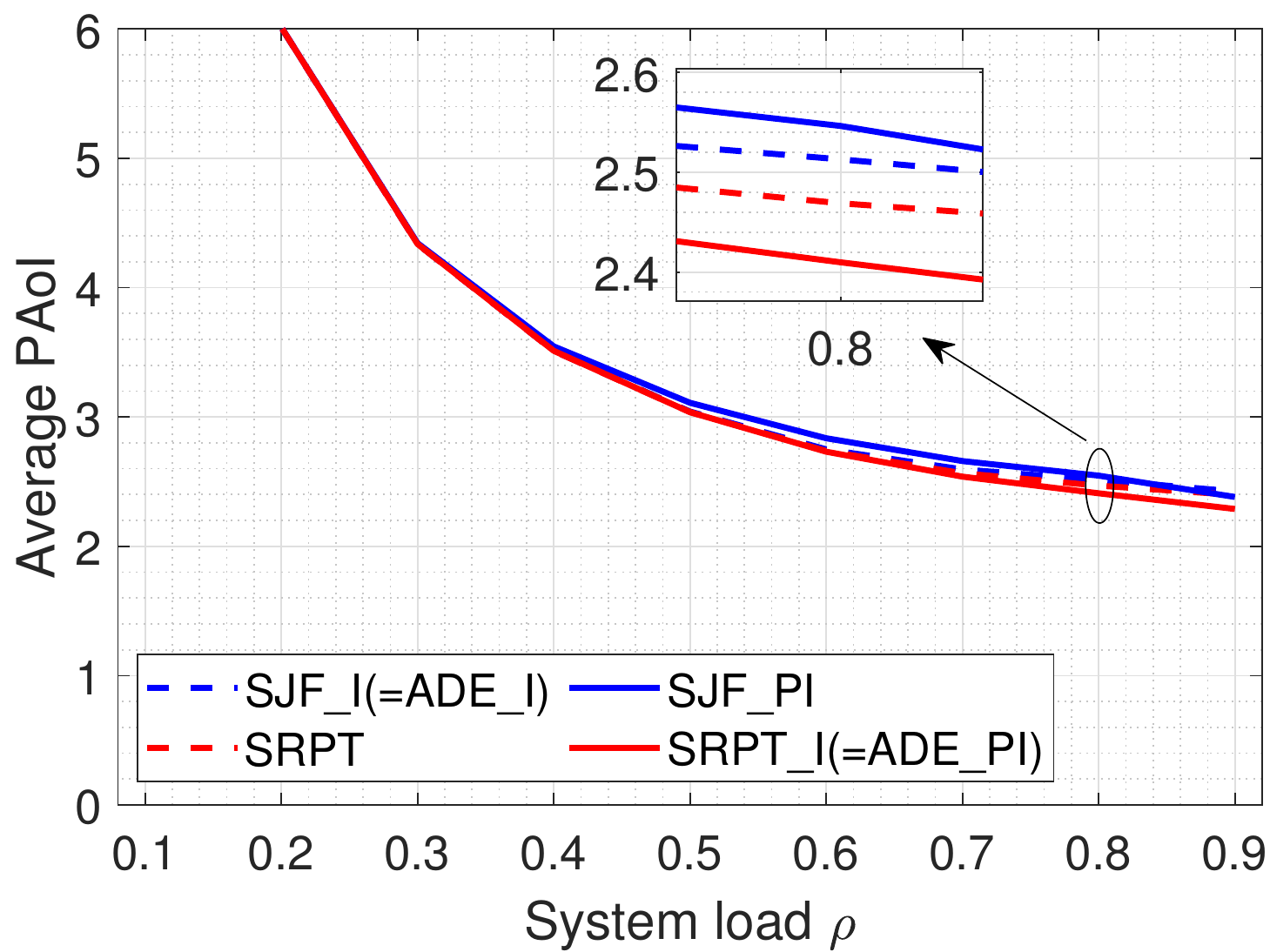}}
	\hspace{0.5em}	
	\subfigure[Interarrival time: Log-normal;  
     Update size: Log-normal ($\mu=1$)]{
		\label{fig:log-log-PAoI-allAspects} 
		\includegraphics[width=0.312\textwidth]{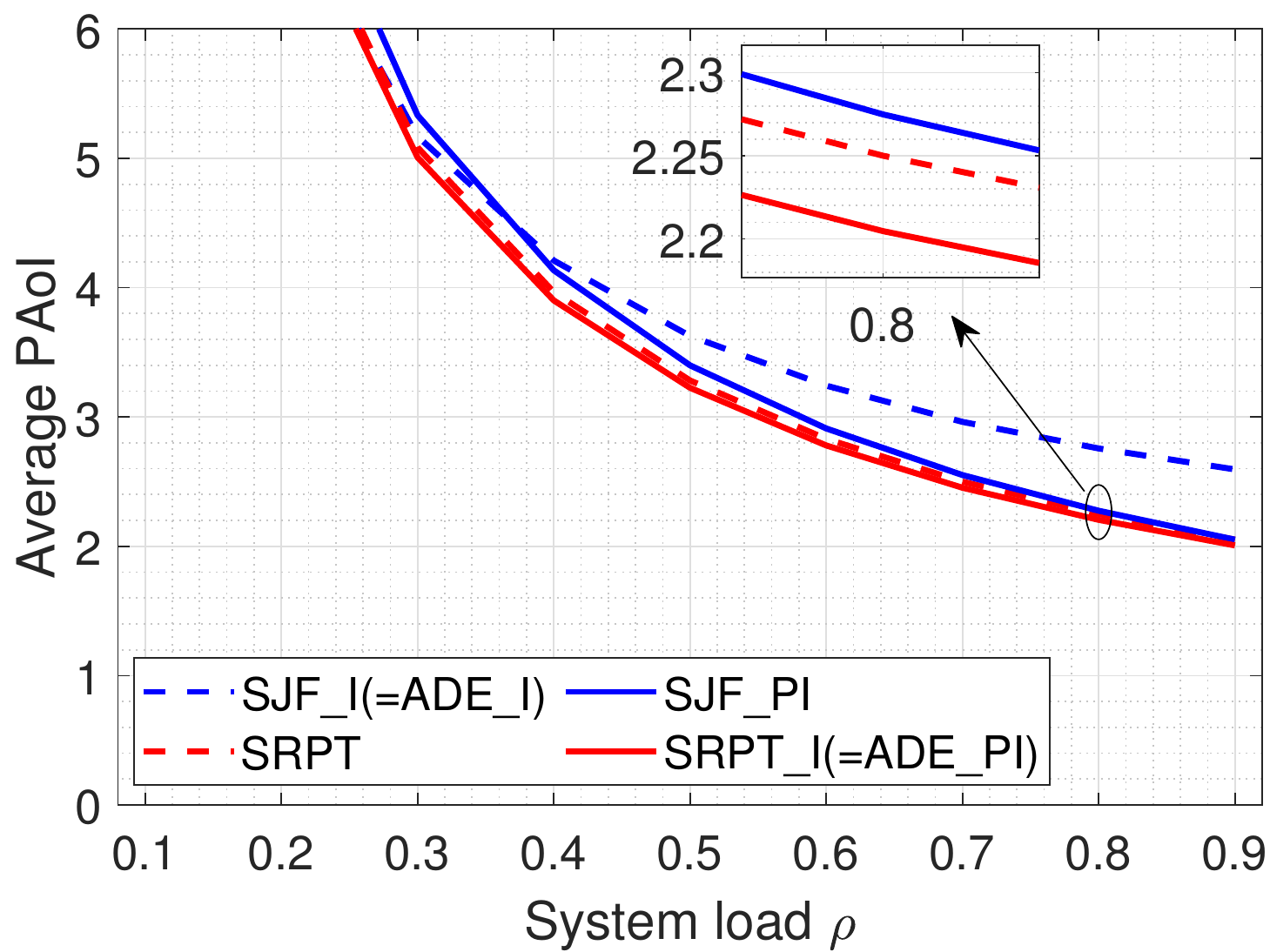}}
	\hspace{0.5em}	
	\subfigure[Interarrival time: Pareto;  
     Update size: Pareto ($\mu=1$)]{
		\label{fig:par-par-paoi-allAspects} 
		\includegraphics[width=0.312\textwidth]{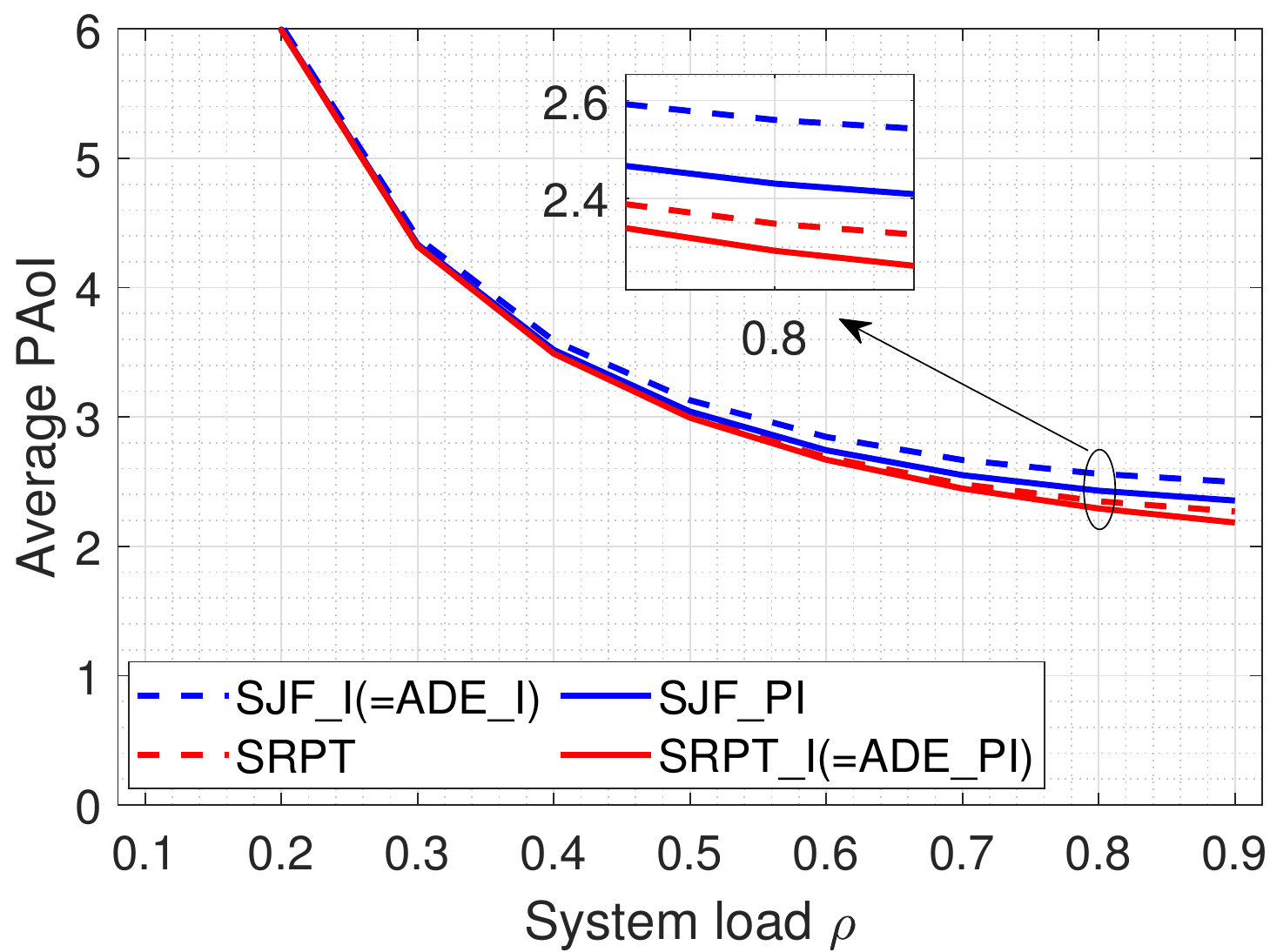}}			
	\caption{Comparisons of the average PAoI performance under different distributions: preemptive, informative, AoI-based policies vs. others}
	\label{fig:all-combined-paoi}
\end{figure*}

\end{document}